\documentclass[envcountsect,runningheads,a4paper]{llncs}\usepackage[subn]{optional}
\usepackage{hyperref}
\opt{subn}{\pagestyle{headings}}

\opt{normal}{%Added packages
    \usepackage{fullpage}
    \usepackage[utf8]{inputenc}
    \usepackage[english]{babel}
    \usepackage{hyperref}
}

\usepackage{cite}
\renewcommand{\paragraph}[1]{\noindent{\bf #1.} }

\usepackage{amsmath,amssymb,url,color,mleftright,mathrsfs}

\numberwithin{equation}{section}

\usepackage[linesnumbered]{algorithm2e}
\usepackage[textsize=scriptsize]{todonotes} 
\newcommand{\matr}[1]{\mathbf{#1}}
\renewcommand{\vec}[1]{\mathbf{#1}}
\newcommand{\reach}{\Rightarrow^*}
\usepackage{verbatim}
\usepackage{graphicx,algorithm2e}
\usepackage{float}
\usepackage{tikz}
\usepackage{tikz-qtree}

\opt{normal,sub}{\usepackage{amsthm}
\theoremstyle{plain}}

    \opt{sub}{%In sub, defining new theorems
    \newtheorem{conjecture}[theorem]{Conjecture}
    \newtheorem{conj}[theorem]{Conjecture}
    \newtheorem{cor}[theorem]{Corollary}
    
    \newtheorem{lem}[theorem]{Lemma}%{\bfseries}{\itshape}
   \newtheorem{thm}[theorem]{Theorem}%{\bfseries}{\itshape}
   \newtheorem{defn}[theorem]{Definition}%{\bfseries}{\itshape}
    
    \newtheorem{proposition}[theorem]{Proposition}
    \newtheorem{prop}[theorem]{Proposition}
    \newtheorem{problem}[theorem]{Problem}
    \newtheorem{prob}[theorem]{Problem}
    \newtheorem{notation}[theorem]{Notation}
    \newtheorem{nt}[theorem]{Notation}
    \newtheorem{ex}[theorem]{Example}
    \newtheorem{rem}[theorem]{Remark}
    \newtheorem{obs}[theorem]{Observation}
    }

    \opt{normal}{%In sub, defining new theorems
    \newtheorem{theorem}{Theorem}
    \numberwithin{equation}{section}
    \numberwithin{theorem}{section}
    {\bfseries}{\itshape}
    {\bfseries}{\itshape}
    \newtheorem{obs}[theorem]{Observation}{\bfseries}{\itshape}
    {\bfseries}{\itshape}
    \newtheorem{cor}[theorem]{Corollary}{\bfseries}{\itshape}
    {\bfseries}{\itshape}
    \newtheorem{prop}[theorem]{Proposition}{\bfseries}{\itshape}
    {\bfseries}{\itshape}
    \newtheorem{prob}[theorem]{Problem}{\bfseries}{\itshape}
    {\bfseries}{\itshape}
    \newtheorem{nt}[theorem]{Notation}{\bfseries}{\itshape}
     {\bfseries}{\itshape}
    \newtheorem{conj}[theorem]{Conjecture}{\bfseries}{\itshape}
    {\bfseries}{\itshape}
    
    \newtheorem{lem}[theorem]{Lemma}{\bfseries}{\itshape}
   {\bfseries}{\itshape}
   \newtheorem{defn}[theorem]{Definition}{\bfseries}{\itshape}
    
    \newtheorem{rem}[theorem]{Remark}{\bfseries}{\itshape}
    {\bfseries}{\itshape}
    {\bfseries}{\itshape}
    \newtheorem{example}[theorem]{Example}{\bfseries}{\itshape}
    }

\newcommand{\N}{\mathbb{N}}
\newcommand{\Z}{\mathbb{Z}}
\newcommand{\R}{\mathbb{R}}
\newcommand{\Q}{\mathbb{Q}}

\newcommand{\calC}{{\mathcal{C}}}
\newcommand{\tme}{\tilde{\mathcal{E}}}
\newcommand{\E}{\mathcal{E}}

\newcommand{\rag}{\rangle}
\newcommand{\lag}{\langle}

\newcommand{\vx}{\vec{x}}
\newcommand{\vp}{\vec{p}}
\newcommand{\vr}{\vec{r}}
\newcommand{\vv}{\vec{v}}

\newcommand{\ve}{\vec{e}}
\newcommand{\NP}{\mathsf{NP}}
\newcommand{\Pspace}{\mathsf{PSPACE}}
\newcommand{\Complete}{\mathsf{complete}}
\newcommand{\Space}{\textbf{\textrm{SPACE}}}

\newcommand{\range}{\textrm{range}}

\usepackage[normalem]{ulem} %% for strike-out \sout

\newcommand{\SubsetFixedAtomic}{\textsc{Subset-Fixed-Atomic}}

\newcommand{\SubsetAtomic}{\textsc{Subset-Atomic}}
\newcommand{\ReachablyFixedAtomic}{\textsc{Reachably-Fixed-Atomic}}
\newcommand{\MonotoneOneThree}{\textsc{Monotone-1-In-3-Sat}}
\newcommand{\ReachablyAtomic}{\textsc{Reachably-Atomic}}
\newcommand{\RevReachablyAtomic}{\textsc{Reversibly-Reachably-Atomic}}
\newcommand{\ReachableReach}{\textsc{Reachable-Reach}}

\newcommand{\IP}{\textsc{IP}}

\def\longrightharpoonup{\relbar\joinrel\rightharpoonup}
\def\longleftharpoondown{\leftharpoondown\joinrel\relbar}
\def\longrightleftharpoons{\mathop{\vcenter{\hbox{\ooalign{\raise1pt\hbox{$\longrightharpoonup\joinrel$}\crcr\lower1pt\hbox{$\longleftharpoondown\joinrel$}}}}}}
\def\revrxn{\mathop{\rightleftharpoons}\limits}

\usepackage{microtype}%if unwanted, comment out or use option "draft"

\opt{sub,normal}{\bibliographystyle{plainurl}}% the recommended bibstyle
\opt{subn}{\bibliographystyle{splncs03}}

\opt{sub}{%Headings in submitted version
\title{Computational Complexity of Atomic Chemical Reaction Networks\footnote{This work was supported by NSF grant 1619343.}}
\titlerunning{Computational Complexity of Atomic Chemical Reaction Networks} %optional, in case that the title is too long; the running title should fit into the top page column

\author[1]{David Doty}
\author[1]{Shaopeng Zhu}
\affil[1]{University of California, Davis, Computer Science Department, CA, USA\\
  \texttt{doty@ucdavis.edu,spzhu@ucdavis.edu}}

\authorrunning{D. Doty and S. Zhu} %mandatory. First: Use abbreviated first/middle names. Second (only in severe cases): Use first author plus 'et. al.'

\Copyright{David Doty and Shaopeng Zhu}%mandatory, please use full first names. LIPIcs license is "CC-BY";  http://creativecommons.org/licenses/by/3.0/

%Dummy classification -- please refer to \url{http://www.acm.org/about/class/ccs98-html}}% mandatory: Please choose ACM 1998 classifications from http://www.acm.org/about/class/ccs98-html . E.g., cite as "F.1.1 Models of Computation". 
% mandatory: Please provide 1-5 keywords
\EventEditors{David Doty, Shaopeng Zhu}
\EventNoEds{2}
\EventLongTitle{The 25th Annual European Symposium on Algorithms (ESA 2017)}
\EventShortTitle{ESA 2017}
\EventAcronym{ESA}
\EventYear{2017}
\EventDate{Sept $4$-$8$, 2017}
\EventLocation{Vienna, Austria}
\EventLogo{}
\SeriesVolume{25}
\ArticleNo{}
}

\opt{normal}{%Headings in normal version
\title{Computational Complexity of Atomic Chemical Reaction Networks}
\author{David Doty \\ \texttt{doty@ucdavis.edu} \and Shaopeng Zhu \\ \texttt{spzhu@ucdavis.edu}}
}

\begin{document}

\opt{subn}{%Headings in submitted version
\title{Computational Complexity of Atomic Chemical Reaction Networks\footnote{This work was supported by NSF grant 1619343.}}
\titlerunning{Computational Complexity of Atomic Chemical Reaction Networks} %optional, in case that the title is too long; the running title should fit into the top page column
\toctitle{Computational Complexity of Atomic Chemical Reaction Networks}
%\author[1]{David Doty}
%\author[1]{Shaopeng Zhu}
\author{David Doty\inst{1} \and Shaopeng Zhu\inst{2}}

\authorrunning{D. Doty and S. Zhu} %mandatory. First: Use abbreviated first/middle names. Second (only in severe cases): Use first author plus 'et. al.'
\tocauthor{D. Doty and S. Zhu}

\institute{University of California, Davis, Computer Science Department, \email{doty@ucdavis.edu}\and University of Maryland, College Park, Computer Science Department,
\email{szhu@terpmail.umd.edu}}%\email{doty@ucdavis.edu}}%\and University of California, Davis,\email{spzhu@ucdavis.edu}}
%\affil[1]{University of California, Davis, Computer Science Department, CA, USA\\
 % \texttt{doty@ucdavis.edu,spzhu@ucdavis.edu}}
%\Copyright{David Doty and Shaopeng Zhu}%mandatory, please use full first names. LIPIcs license is "CC-BY";  http://creativecommons.org/licenses/by/3.0/

%\subjclass{F.1.1}%Dummy classification -- please refer to \url{http://www.acm.org/about/class/ccs98-html}}% mandatory: Please choose ACM 1998 classifications from http://www.acm.org/about/class/ccs98-html . E.g., cite as "F.1.1 Models of Computation". 
%\keywords{chemical reaction network, atomic, computational complexity}% mandatory: Please provide 1-5 keywords
% Author macros::end %%%%%%%%%%%%%%%%%%%%%%%%%%%%%%%%%%%%%%%%%%%%%%%%%
%Editor-only macros:: begin (do not touch as author)%%%%%%%%%%%%%%%%%%%%%%%%%%%%%%%%%%
%\EventEditors{David Doty, Shaopeng Zhu}
%\EventNoEds{2}
%\EventLongTitle{The 25th Annual European Symposium on Algorithms (ESA 2017)}
%\EventShortTitle{ESA 2017}
%\EventAcronym{ESA}
%\EventYear{2017}
%\EventDate{Sept $4$-$8$, 2017}
%\EventLocation{Vienna, Austria}
%\EventLogo{}
%\SeriesVolume{25}
%\ArticleNo{}
% Editor-only macros::end %%%%%%%%%%%%%%%%%%%%%%%%%%%%%%%%%%%%%%%%%%%%%%%
}

\opt{subn}{%In sub, defining new theorems
    %\spnewtheorem{conjecture}[theorem]{Conjecture}{\bfseries}{}
    \spnewtheorem{conj}[theorem]{Conjecture}{\bfseries}{\itshape}
    \spnewtheorem{cor}[theorem]{Corollary}{\bfseries}{\itshape}
    %\spnewtheorem{proposition}[theorem]{Proposition}
    \spnewtheorem{prop}[theorem]{Proposition}{\bfseries}{\itshape}
   % \spnewtheorem{problem}[theorem]{Problem}{\bfseries}{\itshape}
   \spnewtheorem{lem}[theorem]{Lemma}{\bfseries}{\itshape}
   \spnewtheorem{thm}[theorem]{Theorem}{\bfseries}{\itshape}
   \spnewtheorem{defn}[theorem]{Definition}{\bfseries}{\itshape}
    \spnewtheorem{prob}[theorem]{Problem}{\bfseries}{\itshape}
    \spnewtheorem{notation}[theorem]{Notation}{\bfseries}{\itshape}
    \spnewtheorem{nt}[theorem]{Notation}{\bfseries}{\itshape}
    \spnewtheorem{ex}[theorem]{Example}{\bfseries}{\itshape}
    \spnewtheorem{rem}[theorem]{Remark}{\bfseries}{\itshape}
    \spnewtheorem{obs}[theorem]{Observation}{\bfseries}{\itshape}
    }

\maketitle

\begin{abstract}
%[MAYBE COMBINE SOME OF THE UNIMPORTANT THEOREMS AND REVIVE SECTION 7 (Gnacadja).]

    Informally, a chemical reaction network is ``atomic'' if each reaction may be interpreted as the rearrangement of indivisible units of matter.
    There are several reasonable definitions formalizing this idea.
    We investigate the computational complexity of deciding whether a given network is atomic according to each of these definitions.
    
    %Our first definition, 
    \textbf{\emph{Primitive atomic}}, which requires each reaction to preserve the total number of atoms, 
    is shown to be equivalent to mass conservation. 
    Since it is known that it can be decided in polynomial time whether a given chemical reaction network is mass-conserving~\cite{mayr2014framework}, the equivalence we show gives an efficient algorithm to decide primitive atomicity.
   % Another definition, 
   
   \textbf{\emph{Subset atomic}} further requires %that 
   all atoms %are
   be species.
    We show that deciding %whether
    if a %given
    network is subset atomic is in $\NP$, and %the problem 
    ``whether a network is subset atomic with respect to a given atom set'' is strongly $\NP$-$\Complete$.
   
    %Our third definition, \
    \textbf{\emph{Reachably atomic}},
    studied by Adleman, Gopalkrishnan \emph{et al.}~\cite{Adleman2014,ManojPrivateCommunication}, 
    further requires that each species has a sequence of reactions splitting it into its constituent atoms. 
    Using a combinatorial argument, we show that there is a \textbf{polynomial-time} algorithm to decide whether a given network is reachably atomic, 
    improving upon the result of Adleman \emph{et al.} that the problem is \textbf{decidable.} 
    We show that the reachability problem for reachably atomic networks is $\Pspace$-$\Complete$. 
    
    Finally, we demonstrate equivalence relationships between our definitions and some %special
    cases of %another 
    an existing definition of atomicity due to Gnacadja~\cite{gnacadja2011reachability}.
 \end{abstract}

\section{Introduction}\label{introintro}
A \emph{chemical reaction network} is a set of reactions such as $A+B \revrxn C$ and $X \to 2Y$,
intended to model molecular \emph{species} that interact, possibly combining or splitting in the process.
%The model is syntactically equivalent to Petri nets:
%molecules correspond to ``tokens'', species correspond to ``places'', reactions correspond to ``transitions'', and configurations correspond to ``markings''.
%Indeed, the study of chemical reaction networks has profited from this connection~\cite{recrn, angeli2007petri, esparza2016verification}.
%However, due to their different modeling goals (concurrent systems and well-mixed chemistry, respectively), 
%sometimes different questions are germane in each model.
For 150 years~\cite{guldbergwaage1864}, 
the model has been a popular language for describing natural chemicals that react in a well-mixed solution.
%Several recent wet-lab experiments demonstrate the systematic engineering of \emph{custom-designed} chemical reactions~\cite{chen2013programmable, srinivas2015programming, PadiracFujiiRondelez2013, montagne2011programming, SeeSolZhaWin06, qian2011scaling, qian2011neural}, 
%and i
It is known that in theory \emph{any} set of reactions can be implemented by synthetic DNA complexes~\cite{SolSeeWin10}.
%Thus 
Syntactically equivalent to Petri nets~\cite{recrn, angeli2007petri, esparza2016verification}, chemical reaction networks are now equally appropriate as a \emph{programming} language that can be compiled into real chemicals.
With advances in synthetic biology heralding a new era of sophisticated biomolecular engineering~\cite{chen2013programmable, srinivas2015programming, PadiracFujiiRondelez2013, montagne2011programming, SeeSolZhaWin06, qian2011scaling, qian2011neural}, 
chemical reaction networks are expected to gain prominence as a natural high-level language for designing molecular control circuitry. %It is also worth mentioning that The model is s.

There has been a flurry of recent progress in understanding the ability of chemical reaction networks to carry out computation: 
computing functions~\cite{CheDotSolNaCo, DotHajLDCRNNaCo, CheDotSol14, recrn, p1ccrnJournal, salehi2016chemical, salehi2015markov, esparza2016verification, DotyTCRN2014, SpeedFaultsDIST, SolCooWinBru08, AlistarhAEGR2017}, 
as well as other computational tasks such as 
space- and energy-efficient search~\cite{ThaCon12}, 
signal processing~\cite{salehi2014asynchronous, jiang2013discrete}, 
linear I/O systems~\cite{oishi2011biomolecular}, 
machine learning~\cite{napp2013message},
and even identifying function computation in existing biological chemical reaction networks~\cite{cardelli2012cell}.
These studies generally assume that any set of reactions is permissible, 
but not all are physically realistic.
Consider, for example, the reaction $X \to 2X$,
which appears to violate the law of conservation of mass.
Typically such a reaction is a shorthand for a more realistic reaction such as $F + X \to 2X$, 
where $F$ is an anonymous and plentiful source of ``fuel'' providing the necessary matter for the reaction to occur.
The behavior of the two is approximately equal only when the number of executions of $X \to 2X$ is far below the supplied amount of $F$, and if $F$ runs out then the two reactions behave completely differently.
Thus, although $X \to 2X$ may be implemented approximately, to truly understand the long-term behavior of the system requires studying its more realistic implementation $F+X \to 2X$.
A straightforward generalization of this ``realism'' constraint 
is that each chemical species $S$ may be assigned a \emph{mass} $m(S) \in \R^+$, 
where in each reaction the total mass of the reactants equals that of the products.
Indeed, \emph{conservative} Petri nets formalize this very idea~\cite{mayr2014framework, lien1976note},
and it is straightforward to decide algorithmically if a given network is conservative by reducing to a question of linear algebra.

The focus of this paper is a more stringent condition: 
that the network should be \emph{atomic}, 
i.e., each reaction rearranges discrete, indivisible units (atoms), 
which may be of different noninterchangeable types.\footnote{This usage of the term ``atomic'' is different from its usage in traditional areas like operating system or syntactic analysis, where an ``atomic'' execution is an uninterruptable unit of operation~\cite{silberschatz2013operating}.}
(In contrast, mass conservation requires each reaction to rearrange a conserved quantity of continuous, generic ``mass''.)
We emphasize that this is not intended as a study of the atoms appearing in the periodic table of the elements.
Instead, we aim to model chemical systems whose reactions rearrange certain units, 
but never split, create, or destroy those units.
For example, DNA strand displacement systems~\cite{SolSeeWin10, yurke2000dna} 
%~\cite{SolSeeWin10, cardelli-two-domain, LakPhi11, ConHuManThaRoyal, Seelig_Soloveichik2009time, ChiDotKarSek10, QianSoloveichikWinfree2011lncs, Qian_Winfree2011theory, leaklessDNA21, reachabilityBoundsDSDNaCo, ThaCon12, strandDisplacementBayesDNA18, yordanov2013functional, lakin2013modular, qian2014parallel, dalchau2014computational, lakin2014abstract, dalchauprobabilistic, chandran2012meta, Zhang_Seelig2010classifiers} 
have individual DNA strands as indivisible components,
and each reaction merely rearranges the secondary structure among the strands 
(i.e., which bases on the strands are hybridized to others).

Contrary to the idea of mass conservation,
there is no ``obviously correct'' definition of what it means for a chemical reaction network to be atomic, 
as we will discuss.
Furthermore, at least two inequivalent definitions exist in the literature~\cite{Adleman2014, gnacadja2011reachability}.
It is not the goal of this paper to identify a single correct definition.
Instead, our goal is to evaluate the choices that must be made in formalizing a definition, 
to place existing and new definitions in this context to see how they relate to each other, 
and to study the computational complexity of deciding whether a given network is atomic.
This is a step towards a more broad study of the computational abilities of ``physically realistic'' chemical reaction networks.

% \footnote{In \cite{CheDotSolNaCo}, the Chemical Reaction Network for computing the function $f(x_1,x_2)=\lfloor\frac{\max(x_1,x_2)}{2}\rfloor$ contains a reaction ``$K+W\rightarrow \emptyset$'', while in the real world no reactants in any reaction can vanish without producing any product with the same mass as the reactants. Similarly, the reaction ``$Y+V_0\rightarrow Y$'' used in the Chemical Reaction Decider for the language $\{ \vec{c}\in \N^{\{X,Y\}}\setminus \{(0,0)\}\mid \vec{c}(X)\not\equiv \vec{c}(Y)\mod m\}$ ($m\geq 2$), which appears in \cite{recrn}, is also impossible for ``real'' chemical reactions since $V_0$ vanishes with $Y$ being its catalyst. An example that involves multiple reactions would be \begin{eqnarray*}
%  A+B&\rightarrow &C\\
%  A+C&\rightarrow &B
% \end{eqnarray*}
% In the example above, no single reaction directly violates the physical mass conservation law; however, let $\vec{m}(X)>0$ ($X\in\{A,B,C\}$) denote the mass of species $X$, then the first reaction requires $\vec{m}(C)=\vec{m}(B)+\vec{m}(A) > \vec{m}(B)$, while the second requires $\vec{m}(B)=\vec{m}(C)+\vec{m}(A) > \vec{m}(C)$. The contradiction implies that no possible assignment of mass to $A,B$ and $C$ --- unless non-positive mass which does not exist in chemical reactions --- would enable the system to conserve mass.}  

\opt{normal}{
    Consider the following system with $2$ (real world) reactions for example:

    \begin{eqnarray*}
    \textrm{H}_2\textrm{SO}_4 + \textrm{Cu(OH)}_2&\rightarrow& 2\textrm{H}_2\textrm{O} + \textrm{CuSO}_4\\
    \textrm{Fe} + \textrm{CuSO}_4 &\rightarrow& \textrm{FeSO}_4 + \textrm{Cu}
    \end{eqnarray*}

    If we ignore the change of oxidation number in the second reaction (which Chemical Reaction Networks will not be modelling anyway), we may as well just let the ``atom'' set-- set of invariant component in this reaction system -- be 
    \begin{eqnarray*}
    \Delta_1 &=& \{X:=\textrm{H}, Y:=\textrm{OH}, Z:=\textrm{Cu}, U:= \textrm{Fe}, W:= \textrm{SO}_4\}
    \end{eqnarray*} where we define all \emph{indecomposable components} in the \emph{specific} reaction system as atoms; and the let the molecule\footnote{In this abstract model when we call  $\textrm{H}_2\textrm{SO}_4$ and $\textrm{CuSO}_4$ ``molecules'', we also ignore the caveat that substance like $\textrm{H}_2\textrm{SO}_4$ and $\textrm{CuSO}_4$ does not exist as an actual molecule (in the chemical sense) in water solution.} 
    set be: 
    \begin{eqnarray*}
    M&=&\{S:=\textrm{H}_2\textrm{SO}_4, T:= \textrm{Cu(OH)}_2, G:= \textrm{H}_2\textrm{O}, N:=\textrm{CuSO}_4, P:=\textrm{FeSO}_4\}
    \end{eqnarray*}
    and rewrite the reaction system as 
    \begin{eqnarray*}
    S+T&\rightarrow& 2G + N\\
    U + N &\rightarrow& P + Z
    \end{eqnarray*}
    and note that the collection of decompositions (i.e., atomic makeups): 
    \begin{eqnarray*}
    \{S&\mapsto& \{2X,1W\}, T\mapsto \{1Z,2Y\}, G\mapsto\{1X,1Y\}, \\N&\mapsto& \{1Z,1W\}, P\mapsto\{1U,1W\},A\mapsto\{1A\} (\forall A\in \Delta_1) \}
    \end{eqnarray*}
    will maintain all the decomposition information in the two-reaction system. Instead of using the entire periodic table 
    , \emph{we now shrink the size of atom set into $5$}. Note, also, that we may alternatively define $\Delta$ in a more ``canonical'' way---the subset of elements in the periodic table that are actually used in the system, which in this case will be $\Delta_2=\{\textrm{H,O,S,Cu,Fe}\}$, and the corresponding set of decompositions shall be:
    \begin{eqnarray*}
    \{S&\mapsto& \{2\textrm{H},4\textrm{O},1\textrm{S}\}, T\mapsto \{2\textrm{H}, 2\textrm{O},1\textrm{Cu}\}, G\mapsto\{2\textrm{H},1\textrm{O}\},\\ N&\mapsto& \{4\textrm{O},1\textrm{S},1\textrm{Cu}\}, P\mapsto\{4\textrm{O},1\textrm{S},1\textrm{Fe}\},A\mapsto\{1A\} (\forall A\in \Delta_2) \}
    \end{eqnarray*}
    which shows that \emph{we may have multiple choices of ``atoms'' set} and corresponding set of decompositions when modelling the chemical reactions.
}%End of Real World Example that comes only in the normal version.

% On the other hand, some Chemical Reaction Networks in the computational modelling sense may have no ``indecomposable components''. Take, for example, the following Chemical Reaction Network with $3$ species $\{A,B,C\}$ and $3$ reactions:
% \begin{eqnarray*}
%  A&\rightarrow& 2B\\
%  B&\rightarrow&2C\\
%  C&\rightarrow&2A
% \end{eqnarray*}

% Here, each species is ``decomposable'' into $2$ units of another species, making the set of indecomposable components empty. However such reactions cannot happen in the real world, for their violation of mass conservation law: observe that the first and second reactions imply mass of $A$ is fourfold the mass of $C$, while the third gives $\vec{m}(C) = 2\vec{m}(A)$.

\subsection{Summary of Results and Connection with Existing Work}\label{previous111}
%Initially motivated by the interest to inspect into such phenomena involving mass conservation and existence of indecomposable (i.e. atomic) components in Chemical Reaction Networks model and to thoroughly understand the complexity of corresponding decision problems, we researched on existing works,  discovered that previous literature has shown significant meaning,  mathematically, computationally and otherwise, in exploring the atomicity of Chemical Reaction Networks as well as its relationship with mass conservation, and benefited our research therefrom: 

The most directly related previous work is that of Adleman, Gopalkrishnan, Huang, Moisset, and Reishus~\cite{Adleman2014} and of Gnacadja~\cite{gnacadja2011reachability}, which we now discuss in conjunction with our results.

We identify two fundamental questions to be made in formalizing a definition of an ``atomic'' chemical reaction network:

\begin{enumerate}
  \item
  \label{ques:atoms-are-species} 
  Are atoms required to be species? 
  (for example, if the only reaction is $2\textrm{H}_2 + \textrm{O}_2 \revrxn 2\textrm{H}_2\textrm{O}$; 
  then H and O are atoms but not species that appear in a reaction)
  
  \item
  \label{ques:species-reachable-to-atoms} 
  Is each species required to be separable into its constituent atoms via reactions?
\end{enumerate}

A negative answer to~\eqref{ques:atoms-are-species} implies a negative answer to~\eqref{ques:species-reachable-to-atoms}.
(If some atom is not a species, then it cannot be the product of a reaction.)
Thus there are three non-degenerate answers to the above two questions:
no/no,
yes/no,
and
yes/yes.
We respectively call these 
\emph{primitive atomic},
\emph{subset atomic},
and 
\emph{reachably atomic},
defined formally in Section~\ref{newdefns}.
Intuitively, a network is \emph{primitive atomic} if each species may be interpreted as composed of one or more atoms, which themselves are not considered species (a species can be composed of just a single atom, but they will have different ``names'').
More formally, if $\Lambda$ is the set of species, there is a set $\Delta$ of atoms, 
such that each species $S \in \Lambda$ has an \emph{atomic decomposition} $\vec{d}_S \in \N^{\Delta} \setminus \{\vec{0}\}$ describing the atoms that constitute $S$, 
such that each reaction preserves the atoms.
A network is \emph{subset atomic} if it is primitive atomic and the atoms are themselves considered species; i.e., if $\Delta \subseteq \Lambda$.
A network is \emph{reachably atomic} if it is subset atomic, and furthermore, for each species $S \in \Lambda$, there is a sequence of reactions, starting with a single copy of $S$, resulting in a configuration consisting only of atoms. (If each reaction conserves the atomic count, then this configuration must be unique and equal to the atomic decomposition of $S$.)

A long-standing open problem in the theory of chemical reaction networks is the global attractor conjecture~\cite{craciun2009toric, horn1974}, of which even the following special case remains open: is every network satisfying detailed balance \emph{persistence}, i.e., if started with all positive concentrations, do concentrations stay bounded away from 0?
Adleman, Gopalkrishnan, Huang, Moisset, and Reishus~\cite{Adleman2014} defined reachably atomic chemical reaction networks and proved the global attractor conjecture holds for such networks.
Gnacadja~\cite{gnacadja2011reachability}, attacking similar goals, 
%(``Motivated by the intuition that a reaction network should be persistent if it is constructed from building blocks that cannot be depleted''),
defined a notion of atomicity called ``species decomposition'' and showed a similar result.
We establish links between our definitions and those of both~\cite{gnacadja2011reachability} and~\cite{Adleman2014} in Section \ref{acc}. 
We discuss related complexity issues in Sections~\ref{52cva} and~\ref{acc}. 
In particular, Adleman \emph{et al.}~\cite{Adleman2014} showed that it is decidable whether a given network is reachably atomic. 
This is not obvious since the condition of a species being separable into its constituent atoms via reactions appears to require an unbounded search.
We improve this result, showing it is decidable in polynomial time.

Mayr and Weihmann~\cite{mayr2014framework} proved that configuration reachability graphs for mass conserving chemical reaction networks (i.e., conservative Petri nets) are at most exponentially large in the size of the binary representation of the network, 
implying via Savitch's theorem~\cite{savitch1970relationships} a polynomial-space algorithm for deciding reachability in mass-conserving networks.
We use these results in analyzing the complexity of reachability problems in reachably atomic chemical reaction networks in Section~\ref{52cva}.%, and 

It is clear that any reasonable definition of atomicity should imply mass conservation: simply assign all atoms to have mass 1, noting that any reaction preserving the atoms necessarily preserves their total count.
Perhaps surprisingly, the conditions of primitive atomic and mass-conserving are in fact equivalent, so it is decidable in polynomial time whether a network is primitive atomic and what is an atomic decomposition for each species.
A key technical tool is Chubanov's algorithm~\cite{chubanov2015polynomial} for finding exact rational solutions to systems of linear equations with a strict positivity constraint.
%, which is itself recognized as a milestone in the field of computational linear algebra, 
%We employ this algorithm to prove equivalence between mass conservation and primitive atomicity. %, a concept to be discussed in detail in section \ref{newdefns}. 

In their work on autocatalysis of reaction networks \cite{deshpande2013autocatalysis}, Abhishek and Manoj showed that a consistent reaction network is self-replicable if and only if it is critical. Since weak-reversibility implies consistency and our definition of %reversible chemical reaction network
reversibility implies weak-reversibility% defined in \cite{deshpande2013autocatalysis}
, we obtain the following equivalence:
let a chemical reaction network $\calC$ be reversible. Then $\calC$ is mass conserving if and only if there does not exist $\vec{c}_1 < \vec{c}_2 \in\N^{\Lambda}$ 
    such that $\vec{c}_1 \Rightarrow^* \vec{c}_2$.% for reversible chemical reaction networks: mass-conservation $\iff$

Lastly, we note that there have been other models addressing different aspects of atomicity (not necessarily using the term ``atomic'').
They focus on features of chemical reaction networks not modeled in this paper. \opt{sub,subn}{For discussions on these works, please see Section \ref{relatedbutnotsame} of the Appendix.}

\newcommand{\Egrelatedbutnotsameone}{For example,  Johnson, Dong and Winfree~\cite{johnson2016verifying} study, using concept called \emph{bisimulation}, 
the complexity of deciding whether one network $\mathcal{N}_1$ ``implements'' another $\mathcal{N}_2$.
The basic idea involves assigning species $S$ in $\mathcal{N}_1$ to represent sets of species $\{A_1,\ldots,A_k\}$ in $\mathcal{N}_2$. 
Thus $S$ may be intuitively thought of as being ``composed'' of $A_1,\ldots,A_k$; 
however, an allowed decomposition is $\emptyset$, 
to account for the fact that some species in $\mathcal{N}_1$ have the goal of mediating interactions between species in $\mathcal{N}_2$  without ``representing'' any of them.
Molecules, on the other hand, are always composed of at least one atom.

%in finding a valid interpretation between formal and implementation chemical reaction networks, imposing the ``atomic condition'' as one of the requirements for a correct interpretation;\footnote{Intuitively, an \emph{interpretation} is a function between the species sets of two chemical reaction networks that can be extended to states, preserving the linearity of formal sum for all reactions.} Our model describes a different phenomenon, disallows species to vanish violating mass conservation,\footnote{E.g, a reaction of the form $X\rightarrow \emptyset$ cannot appear in our model.} and found different complexity results. There may be more subtle differences between these models, presumably (partially) invoked by allowing non mass-conserving reactions in the bi-simulation model.
}

\opt{normal}{
    %%Cannot make multiple empty lines for this %%\opt{normal}, because this normal is a footnote. %%Tried making multiple empty lines, they looked %%terrible.
    %%%
    %%%
    \Egrelatedbutnotsameone
    \footnote {One may tend to think at first glance that, for a given chemical reaction network $\calC = (\Lambda, R)$ atomic with respect to $\Delta$ via decomposition matrix $\matr{D} = (\vec{d}_S)_{S\in\Lambda}$, the mapping \begin{eqnarray*}\vec{d}:\Lambda&\rightarrow& \N^\Lambda\\S&\mapsto& \vec{d}_S,\end{eqnarray*} as defined in our future chapters, gives an interpretation between $\calC = (\Lambda, R)$ and $\calC' = (\Delta, R')$ where $R'\subseteq\N^{\Delta}$ is the natural linear extension of $R\subseteq\N^{\Delta}$ -- that is, for all $(\vec{r},\vec{p})\in R$, construct $(\vec{r}',\vec{p}')\in R'$ s.t. $\vec{r}'=\sum_{S\in [\vec{r}]}\vec{d}_S, \vec{p}'=\sum_{S\in [\vec{p}]}\vec{d}_S$. For example, when $\calC = (\{S_1,S_2\}, \{S_1\rightarrow 2S_2\})$, $\calC'=(\{S_2\}, \{2S_2\rightarrow 2S_2\})$, and the interpretation would be $m=\{(S_1,2S_2),(S_2,S_2)\}$. (To avoid confusion one may also relabel the species in $\calC'$: for example, in the previous case, $\calC' = (\{S_2^{'}\},\{2S_2^{'}\rightarrow 2S_2^{'}\})$ and $m=\{(S_1,2S_2^{'}),(S_2,S_2^{'})\}$.

    But this reduction does not necessarily work, as the correct interpretation defined in \cite{johnson2016verifying} requires ``equivalence of trajectories''. Consider a slightly modified example from \cite{johnson2016verifying}: let $\calC = \{\Lambda, R\}$ where $\Lambda = \{A,B_1,B_2,C,D\}$ and $R=$
    \begin{eqnarray*}A&\rightarrow& B_1\\A&\rightarrow&B_2\\B_1&\rightarrow& C\\C&\rightarrow& A\\\forall S\in\Lambda \textrm{ where }S\neq D:\ S&\rightarrow & 3D
    \end{eqnarray*} 
    
    By our conjecture above, this would give an interpretation from $\calC$ to $\calC'$, where $\calC'=(\{D'\},\{3D'\rightarrow 3D'\})$ and $m=\{(A,3D'),(B_1,3D'),(B_2,3D'),(C,3D'),(D,D')\}$. Yet this is NOT a correct interpretation, for, as \cite{johnson2016verifying} pointed out, $\calC$ is subject to deadlock yet $\calC'$ is not.}
}%End of the optional footnote that only happen in normal version.
%Similar as above, I cannot make multiple white lines here, because I donnot want the "In Atom Mapping..." sentence to be beginning of a NEW paragraph.
\newcommand{\Egrelatedbutnotsametwo}{While studying catalysis in reaction networks, Manoj \cite{gopalkrishnan2010catalysis} developed another definition of ``atomic'' requiring that each non-isolated \emph{complex} ``decompose'' uniquely and explicitly to its atomic decomposition in a weakly reversible network (variables in isolated complexes are atoms). Our definition of ``reversibly reachably atomic'' is stronger than this.%, while ``subset atomic'' is weaker. % definition.
\footnote{Consider the network consisting of a single reaction $X+Y\rightleftharpoons 4Z$, which is ``atomic'' by the definition in \cite{gopalkrishnan2010catalysis} but not reversibly reachably atomic. On the other hand, for reversibly reachably atomic networks one may explicitly decompose each molecular species and obtain the unique atomic decomposition for each complex, so reversibly reachably atomic implies ``atomic'' in \cite{gopalkrishnan2010catalysis}.} %On the other hand, guaranteed existence of isolated complexes ensures the subset-atomicity.}% definition.}
Manoj \cite{gopalkrishnan2010catalysis,ManojPrivateCommunication} went ahead to prove that the ideal $(\mathcal{E}_G)$ generated by the event binomials is prime for ``atomic'' networks.\footnote{For example, $(\mathcal{E}_G) = (XY-Z^4)$ for the single-reaction network $X+Y\rightleftharpoons 4Z$.}

Mann, Nahar, Schnorr, Backofen, Stadler, and Flamm~\cite{mann2014atom} studied a notion of atomicity based on the idea that shared electrons between atoms must be conserved.
This model is more chemically detailed than ours, but also more limited, being unable to model higher-level notions of ``shared subcomponents'' such as DNA strands whose secondary structure, but not their primary structure, is altered by strand displacement reactions~\cite{SolSeeWin10, yurke2000dna}. Mercedes and Alicia \cite{millan2016structure} explored ``MESSI'' networks, a biological network where all reactions are mono or bi-molecular, and the set of non-intermediate species can be partitioned into isomeric classes. On contrary, isomers are not allowed for atoms in our model, while we do not restrict reactions to be mono or bimolecular.
Tapia and Faed~\cite{tapia2013atomizer} developed software for finding atomic decompositions in a similarly detailed setting motivated by specific biochemical experiments.

% In Atom Mapping with Constraint Programming\cite{mann2014atom}, Martin  Mann,  Feras  Nahar,  Norah  Schnorr,  Rolf  Backofen,  Peter  F. Stadler  and  Christoph  Flamm % mainly concerns``shared electrons''
% studied the conservation law on shared electrons between atoms, which is different from our model both phenomenon-wise and computation-wise, as moledules in our model do not share actual copies of atoms. In The Atomizer: Extracting Implicit Molecular Structure from
% Reaction Network Models\cite{tapia2013atomizer},  Jose-Juan Tapia and James R. Faed developped a software for finding atomic compositions \emph{in certain bio-chemical setting}, but refrained from discussion on computational complexity of their algorithms.
}

\opt{normal}{
    \Egrelatedbutnotsametwo
    \subsection{Structure}\label{subsection02}
    With the expectation to fully study the computational power as well as other computationally interesting properties of atomic chemical reaction networks in future, we decided to focus this paper on \emph{different definitions of atomicity} and the \emph{computational complexity for deciding atomicity} under each definition.

    Section \ref{prelim123} provides some definitions and notations as background for subsequent sections.

    Section \ref{newdefns} introduces the core concepts of this paper: \emph{primitive atomicity, subset atomicity, reachable atomicity and reversibly-reachable atomicity}. Intuitively and informally, \emph{primitive atomicity} requires that 
    
    \begin{enumerate}
        \item There exists a decomposition for each species into indecomposable components such that each reaction preserves the total count of each type of indecomposable components.    
        \item Each \emph{composite species} can be decomposed into the linear combination of indecomposable components (and thus its decomposition should contain at least $2$ units of indecomposable components), while each \emph{elementary species} --- which are indecomposable components that are themselves species, should be its own ``decomposition''. And, 
        \item Each elementary species appears in the decomposition of at least one composite species; thus no elementary species is ``redundant''.
    \end{enumerate}

    It is called ``primitive'' because it captures our primitive sense of what properties a chemical reaction network with indecomposable components (``atoms'') should have. \emph{Subset atomicity}, on the other hand, requires all the above as well as an additional condition: all indecomposable components are themselves elementary species. As mentioned in previous sections, \emph{reachable-atomicity} further requires that each composite species can be actually decomposed into elementary species via sequences of reactions. Lastly, \emph{reversibly-reachable atomicity} requires in addition to all that are true for \emph{reachably atomic} networks, plus that each composite species can be constructed from its atomic decompositions via sequences of reactions. It will be shown in subsequent sections that the problems of deciding different atomicity type belong to different computation complexities.  

    Section \ref{MCPAEquiv} proves the equivalence between the concept of mass conservation and primitive atomicity. Intuitively, when the set of atoms might be defined independently from the set of species, one may translate between ``a mass assignment to each species that preserves the total mass'' and ``an assignment of number of \emph{same type of atoms} that preserves the total count of atoms'' using some algebraic manipulation of the given assignment.

    Section \ref{CSFA} attends to the complexity of deciding whether a chemical reaction network is subset atomic. We define two languages, $\SubsetFixedAtomic$ and $\SubsetAtomic$, the first describing chemical reaction networks paired with a \emph{given} atom set with respect to which it is indeed subset atomic, and the second describing subset atomic chemical reaction networks without a given atom set. We proved both languages in $\NP$, with the first one $\NP$-\textsf{hard}. We conjecture that the second is also $\NP$-\textsf{hard}. 

    Proof that both languages are in $\NP$ is based on a reduction to the Integer-Programming problem, which is itself in $\NP$. The reduction is necessary because no obvious short witness for $\SubsetFixedAtomic$ and $\SubsetAtomic$ exists due to potentially large size of decomposition. Proof of the strong $\NP$-hardness of $\SubsetFixedAtomic$ is based on a reduction from the ``Monotone $1$-in-$3$'' problem.

    Section \ref{52cva} looks into the complexity of deciding reachable-atomicity as well as the configuration reachability problem under the reachably atomic constraint. By exhibiting a polynomial time bottom-up algorithm based on the special structure of configuration reachability graphs of reachably atomic chemical reaction networks, we solved the complexity of deciding reachable-atomicity; By first applying  Savitch's Theorem and then modifying \cite{mayr2014framework}'s result into a simulation of polynomial space Turing machines with reachably atomic chemical reaction networks, we proved that reachability problem for  reachably atomic chemical reaction networks are $\Pspace$-$\Complete$.

    Section \ref{acc}, as mentioned above, mainly discusses about the relationships between our model and some models established in \cite{gnacadja2011reachability}. By defining the corresponding mappings and then treading through the definitions, we proved that a chemical reaction network is subset atomic if and only if it admits a `core composition', disallows isomerism of elementary species and contains no ``redundant'' elementary species. We then gave two proofs, one by defining an auxilliary subspace of $\R^{n}$ ($n$ being the number of types of elementary species) and arguing its set-containment relationship with other relevant subspaces, and another one by directly applying a Theorem in \cite{gnacadja2011reachability}, that reachable-atomicity implies constructivity of chemical reaction networks. Lastly, we gave a combinatorics-based proof that reversibly-reachable atomcity is equivalent to explicitly-reversibly constructivity plus the restriction that no isomerism is allowed among elementary species.

    Section \ref{openopen} mentioned some open problems left for future works.
}%End of Normal text for Intro to Structure of Paper.

\section{Preliminaries}\label{prelim123}

%We introduce some basic definitions used in this paper. The definitions and notations introduced in this section are mostly drawn from \cite{CheDotSol14}, with certain modifications.
Let $\Z,\N,\R$ respectively denote the set of integers, nonnegative integers, and reals.
Let $\Lambda$ be a finite set.
We write $\N^{\Lambda}$ to denote $\{f:\Lambda\rightarrow \N\}$. 
Equivalently, by assuming a ``canonical'' ordering on $\Lambda$, 
an element $\vec{c}\in\N^\Lambda$ can also be viewed as a $|\Lambda|$-dimensional vector of natural numbers, 
with each coordinate labeled by $S\in\Lambda$ interpreted as the count of $S$. 
$\vec{c}\in\N^\Lambda$ interpreted this way is called a \emph{configuration}.
We sometimes use multiset notation, e.g., $\{3A, 2B\}$ to denote the configuration with 3 copies of $A$, 2 of $B$, and 0 of all other species.
$\Z^\Lambda,\R^{\Lambda},\N^{\Lambda\times\Delta},\N^\Delta$ (where $\Delta$ is also a finite set) are defined analogously.

We write $\vec{c} \leq \vec{c}'$
to
denote that $(\forall X\in\Lambda)\ \vec{c}(X) \leq \vec{c}'
(X)$, and $\vec{c} <\vec{c}'$
if $\vec{c}\leq  \vec{c}'$ and $\vec{c}\neq  \vec{c}'$.
We say $\vec{c}$ and $\vec{c}'$ are \emph{incomparable} if $\vec{c} \not\leq \vec{c}'$ and $\vec{c} \not\geq \vec{c}'$.

% \begin{example}
% \normalfont{}
% For $\Lambda = \{S_1,S_2\}$, configurations $(1,2)$ and $(2,1)$ are incomparable. 
% The first configuration indicates that there exist $1$ unit of $S_1$ and $2$ units of $S_2$ in the system, 
% while the second indicates $2S_1$'s and $1S_2$'s.
% \end{example}

\begin{defn}\label{CRN}%\normalfont{}
    Given a finite set of chemical species $\Lambda$, 
    a \emph{reaction}  over $\Lambda$ is a pair $\alpha = (\vec{r},\vec{p}) \in \N^\Lambda \times \N^\Lambda$,
    specifying the stoichiometry of the reactants and products respectively.\footnote{There is typically a positive real-valued \emph{rate constant} associated to each reaction, but we ignore reaction rates in this paper and consequently simplify the definition.}
    
    A \emph{chemical reaction network} is a pair $\calC = (\Lambda, R)$, 
    where $\Lambda$ is a finite set of chemical species, 
    and $R$ is a finite set of reactions over $\Lambda$.
\end{defn}

%For the purpose of this paper, assume that $k = 1$, so that $\alpha = (\vec{r}, \vec{p}, 1)$ is also represented by the pair $(\vec{r}, \vec{p})$. We also assume that $(\forall (\vec{r},\vec{p})\in R)\vec{r}\neq\vec{p}$. 

% \opt{normal}{\begin{rem}\normalfont{}
%     Note that the last assumption in Definition \ref{CRN} not only simplifies our proof for certain theorems (e.g., Theorem \ref{VisAtomP}), but also embodies the philosophy that a reaction needs to do some change to the system: a ``reaction'' $\vec{r}\rightarrow \vec{r}$ does not do any change to the system at all.\end{rem}

%     Without statement otherwise, in the rest of this section we assume $(\Lambda, R)=\calC$ is a Chemical Reaction Network.
% }

%\begin{definition}%\normalfont{}
A chemical reaction network is \emph{reversible} if $(\forall(\vec{r},\vec{p})\in R)\ (\vec{p},\vec{r})\in R$.
%\end{definition}

%\begin{definition}[Reachability of Configurations]%\normalfont{}
For configurations $\vec{c}_1,\vec{c}_2\in \N^\Lambda$, %we 
we write $\vec{c}_1\Rightarrow^{*}_{\calC} \vec{c}_2$ (read ``$\calC$ reaches $\vec{c}_2$ from $\vec{c}_1$'') if there exists a finite reaction sequence (including the empty sequence) that starts with $\vec{c}_1$ and ends with $\vec{c}_2$. For %the sake of 
simplicity, %we 
write $\vec{c}_1\Rightarrow^{*} \vec{c}_2$ (read ``$\vec{c}_2$ is reachable from $\vec{c}_1$'') when $\calC$ is clear. %from context.
%\end{definition}

\opt{normal}{
\begin{rem}\normalfont{}
    Note that $\vec{r}\Rightarrow^* \vec{r}$ is vacuously true for any $(\vec{r},\vec{p})\in R$, by setting the reaction sequence to be empty (length-$0$). There may also be positive-length reaction sequences for certain networks and certain $(\vec{r},\vec{p})\in R$ to reach $\vec{r}$ from $\vec{r}$: for example, $\calC = (\{S_1,S_2,S_3,S_4\},\{S_1\rightarrow S_2 + S_3, S_2+S_3\rightarrow S_4, S_4\rightarrow S_1 \})$. Consider the reaction $(\{1S_1\},\{1S_2,1S_3\})\in R$, we have $\{1S_1\}\Rightarrow^*\{1S_1\}$ by executing the three reactions in order. 

    Note also that this does not contradict our previous assumption that $(\forall(\vec{r},\vec{p})\in R)\vec{r}\neq\vec{p}$., for $\vec{r}\Rightarrow^*\vec{r}$ does not require that $(\vec{r},\vec{r})\in R$.%, and the reachability is satisfied by either executing no reaction or some actual reactions that does  
\end{rem}
}%End of opt{normal} for the 2nd remark. seems not necessary for the sub version.

\begin{defn}%\normalfont{}
%\normalfont{}
    Given $\vec{c}\in\N^{\Lambda}$ (or $\Z^\Lambda, \R^\Lambda$ etc. analogously), the \emph{support} of $\vec{c}$, written as $[\vec{c}]$, is the set $\{S\in\Lambda\mid \vec{c}(S)\neq 0\}$.
\end{defn}

A few more notational conventions are listed here: write $\ve_A\in\N^\Lambda$ as the unit vector that has count $1$ on $A\in\Lambda$ and $0$ on everything else. Given a vector $\vx \in \N^\Lambda$, write %$\|\vx\| =
$ \|\vx\| = \|\vx\|_1 = \sum_{S \in \Lambda} \vx(S)$. 
When $\cdot$ is any data, write $\lag\cdot\rag$ for its binary representation as a string, 
so $|\lag\cdot\rag|$ is the length of the binary representation of $\cdot$. Given $f:A\rightarrow B$ and $C\subseteq A$, $f\restriction C$ %, which reads ``$f$ restricted to $C$'', 
is the function %defined on $C$
$C\rightarrow B$, $c\mapsto f(c)$ ($\forall c\in C$). Lastly, when $\matr{M}$ is a matrix, write $\matr{M}^T$ as its transposition.

\section{Definitions of ``Atomic''}\label{newdefns}

This section addresses definitions of several classes of networks, some computational complexity result of which will be exhibited later.

%\subsection{Primitive Atomic}
    
    %We now arrive at our first definition of ``atomic'', which we call \emph{primitive atomic}.
    %All later definitions of atomic will imply this definition (i.e, the set of primitive atomic networks is a superset of the set defined by other definitions).%}
    
    Intuitively, $\calC = (\Lambda, R)$ is primitive atomic if all species can be decomposed into combinations of some atoms.
    Atoms are not required to be %a subset of the 
    species. %set $\Lambda$. %the set of species.
    Each reaction conserves the total count of each type of atom in the species involved 
    (i.e., the reaction can only rearrange atoms but not create or destroy them). 
    
     Note that the purpose of studying the primitive-atomic model (as well as all other types of atomic later) is not to analyze ``real-world'' atoms. Instead, we are trying to study how molecules can be interpreted as decomposable into \emph{exchangeable parts}. In particular, if we know only the reactions but not those exchangeable parts, we are interested in whether the reactions can tell us how the molecules are composed from parts. Proposition \ref{primitivemcequiv} below, for example, shows that this information can be retrieved by finding a mass distribution vector.
    %; each nonatomic species composition must have at least two atoms, and each atomic species' composition has exactly one atom (equal to itself). 
    %We refer to non-atomic species, i.e. species in $\Lambda\setminus\Delta$, as \emph{molecular} species. 
    
   % \todoi{DD: I don't understand why $\Delta$ is being treated as a possible subset of $\Lambda$ here. I wrote what I think is a better definition below. If it looks okay, please update the rest of the paper to correct any references to it and delete the above definition.}
    
    \begin{defn}[primitive atomic]\label{primitiveatomic}
        Let $\Delta$ be a nonempty finite set and $\calC=(\Lambda,R)$ a chemical reaction network.
        $\calC$ is \emph{primitive atomic with respect to $\Delta$} if 
        %there is a nonempty finite set $\Delta$ such that 
        for all $S \in \Lambda$,
        there is $\vec{d}_S \in \N^{\Delta} \setminus \{\vec{0}\}$ 
        such that
        
        \begin{enumerate}
            \item
            \label{primitive1} 
            $(\forall (\vec{r},\vec{p}) \in R) (\forall A \in \Delta)
            \sum\limits_{S \in \Lambda} \vec{r}(S) \cdot \vec{d}_{S}(A) = 
            \sum\limits_{S \in \Lambda} \vec{p}(S) \cdot \vec{d}_{S}(A)$
            (reactions preserve atoms), and
            
            %\item\label{primitive2} $(\forall W\in \Lambda\setminus \Delta)$ $||\vec{d}_{W}||_1\geq 2$; $\forall A\in\Delta\cap\Lambda$, $\vec{d}_A = \{1A\}$, and
            
            \item 
            \label{primitive3} 
            $(\forall A\in\Delta)(\exists S\in \Lambda)\ \vec{d}_S(A)\neq 0$.
            (each atom appears in the decomposition of some species)
        \end{enumerate}     
    
        For $S\in\Lambda$, call $\vec{d}_S$ in Condition (\ref{primitive1}) the (atomic) \emph{decomposition} of $S$. %\emph{compositions} of species are defined analogously. 
        We say $\calC$ is \emph{primitive atomic} if there is a nonempty finite set $\Delta$ such that $\calC$ is primitive atomic with respect to $\Delta$. In the cases above, $\Delta$ is called the \emph{set of atoms}.
    \end{defn}

    \newcommand{\RemarkWhyConditionTwo}{\begin{rem}\label{Whyjustwhy}\normalfont{}We included Condition $(2)$ to assist the proof of Lemma \ref{EConservative} in the Appendix. It appears for that proof (and other proofs relevant) to work, it doesn't matter if some atom appears only in the decomposition of a catalyst (a species that explicitly showed up on both sides of the reaction and has the same stochiometric coefficient).\end{rem}}

    Condition~\eqref{primitive1} embodies the intuition above. 
    Condition~\eqref{primitive3} prescribes that each atom appears in the decomposition of at least one species.
    \opt{normal}{\footnote{\RemarkWhyConditionTwo}}\opt{sub,subn}{(See Remark \ref{Whyjustwhy} for some comment on condition 2.)} Consider the network $\mathcal{C}=(\{X,Y,W,Z\},$ $\{((2,1,0,1)^T,(0,0,2,1)^T),((1,2,1,1)^T,(0,1,1,2)^T)\})$. One may write $\calC$ as: $$\{2X + Y + Z\rightarrow 2W + Z,\ X + 2Y + W +Z \rightarrow  Y+W + 2Z.\}$$
        %\begin{equation*}\left\{\begin{array}{rl}
         %   
          %  \end{array}
           % \right.
        %\end{equation*}
    %This ensures that every element in the set of atoms does actually participate in forming some molecular species. 

$\calC$ is primitive-atomic with respect to, say, $\Delta = \{H,O\}$, via the decomposition vector $\vec{d}_X=(2,0)^T, \vec{d}_Y=(0,2)^T$, $\vec{d}_W=(2,1)^T, \vec{d}_Z=(2,2)^T$. Here $\vec{d}_X = (2,0)^T$ means the species $X$ is composed of $2$ units of atom $H$ and $0$ unit of atom $O$, and $\vec{d}_Y,\vec{d}_W,\vec{d}_Z$ can be interpreted likewise. Observe that each of the two reactions in $\calC$ preserves the total count of each type of atom on both sides of reactions.%, etc.
    
   Next, we introduce the definitions of stoichiometric matrix and decomposition matrix. %, in order to facilitate a more terse and transparent way of expressing conservation laws. 
   In particular, $\matr{A}$ encodes the net change of species caused by execution of one reaction, %that applying a reaction causes,
   and $\matr{D}$ compiles all decomposition vectors into one data structure.
  
   \begin{defn}[Stoichiometric Matrix]\label{stoi}
The stoichiometric matrix $\matr{A}$ for a chemical reaction network $\mathcal{C}=(R,\Lambda)$ is the $|R|\times |\Lambda|$ matrix where the entry $\matr{A}_{(\vec{r},\vec{p}),S} = \vec{p}(S)-\vec{r}(S)$ for each $(\vec{r},\vec{p})\in R$ and $S\in \Lambda$. 
\end{defn}

Notation-wise, $\matr{A}_{(\vec{r},\vec{p}),S}$ is the entry whose row is labeled by the reaction $(\vec{r},\vec{p})$ and column by the species $S$. Each row of the stoichiometric matrix represents the change of count of each species via execution of $1$ unit of $(\vec{r},\vec{p})$. 
    \opt{sub,subn}{For more illustration%intuition
    , see Example \ref{H20}.} %Don't want to put multiple empty lines before this, for don't want to make this sentence ("For some intuition") a new paragraph.

\newcommand{\ExampleOfStoichiometricMatrix}{
    \begin{example}\label{H20}\normalfont{} 
        In the network $\mathcal{C}$ mentioned above: %$\mathcal{C}=(\{X,Y,W,Z\},\{((2,1,0,1)^T,(0,0,2,1)^T),((1,2,1,1)^T,(0,1,1,2)^T)\})$. The set of reactions can be described as: 
        \begin{equation*}\left\{\begin{array}{rl}
            2X + Y + Z\rightarrow& 2W + Z    \\
            X + 2Y + W + Z\rightarrow& Y + W + 2Z,
            \end{array}
            \right.
        \end{equation*}
        %For $\mathcal{C}$, 
        the stoichiometric matrix $\matr{A}=\begin{bmatrix} -2 & -1 & 2 & 0\\-1  & -1 & 0 & 1 \end{bmatrix}$.
    \end{example}
}%Example H20. This goes to Appendix and hence gets kept.

    \opt{normal}{\ExampleOfStoichiometricMatrix}

\begin{defn}[Decomposition Matrix] \label{DecMat}
    Let $\calC=(\Lambda, R)$ be primitive atomic with respect to $\Delta$. 
    The decomposition matrix, denoted as $\matr{D}_{\Delta}$ \emph{for $\calC$ with respect to $\Delta$} 
    is the $|\Lambda|\times |\Delta|$ matrix whose row vectors are $(\vec{d}_S)^T \ (S\in \Lambda)$.%, where $\vec{d}_S\restriction\Delta$ is the vector $\vec{d}_S$ viewed as a function $\Lambda\to \N$ restricted to $\Delta$.
\end{defn}

 Note that the set of decomposition vectors $\{\vec{d}_S\}_{S\in\Lambda}$ is in general not unique for primitive atomic chemical reaction networks -- for example, $A+B\to C$ is primitive atomic with respect to $\Delta = \{D\}$ via $(k,k,2k) (\forall k\in \mathbb{N}_{>0})$. Correspondingly, $\matr{D}_{\Delta}$'s are defined with respect to each set %of decomposition vectors 
 $\{\vec{d}_S\}_{S\in\Lambda}$. %In fact, each $\matr{D}_{\Delta}$ composes all information of a certain set $\{\vec{d}_S\}$ in the form of a matrix.
  \opt{sub,subn}{See Example \ref{H20Revisited} and Remark in \ref{AD} for some more discussion of decomposition matrices.}
  
\newcommand{\EgAndRemarkOnDecompositionMatrix}{%Example to appear only in appendix and full version
    \begin{rem}\label{AD}\normalfont{} We note that a decomposition matrix has the following properties: 
    \begin{enumerate} 
    \item Let $\vec{c} \in\N^{\Lambda}$ %denote the vector with each entry $\vec{c}(S)$ being the current count of species $S\in\Lambda$ in the system
    be a configuration for primitive atomic $\calC=(\Lambda, R)$, then the vector $\matr{D}_{\Delta}^{T}\cdot \vec{c}\in \N^\Delta$ illustrates the current count of each atom in the system. That is, $(\matr{D}_{\Delta}^{T}\cdot \vec{c})(A)$, the entry of $\matr{D}_{\Delta}^{T}\cdot \vec{c}$ labeled by $A\in\Delta$, represents the current count of atom $A$ in the system. In fact,  
    for each $A\in\Delta$, 
    \begin{eqnarray}
        \nonumber%\vec{a}(b) &=& 
        \textrm{current count of $A$ in system} &=& \sum\limits_{S\in\Lambda}\vec{d}_{S}(A)\cdot \vec{c}(S)\\\nonumber
         &=& (\matr{D}^{T}_{\Delta})_{(A,:)}\cdot \vec{c},\nonumber
    \end{eqnarray}
    
    where $(\matr{D}^{T}_{\Delta})_{(A,:)}$ stands for the %$b$-th 
    row vector of $\matr{D}^{T}_{\Delta}$ labeled by $A$. 
    
    \item If $\calC=(\Lambda,R)$ is primitive atomic with respect to $\Delta$, then $\matr{A}\cdot\matr{D}_{\Delta} = \matr{0}$, where $\matr{A}$ is the stoichiometric matrix in Definition \ref{stoi} above, and $\matr{0}$ is the $|R|\times |\Delta|$ zero matrix. Intuitively, this illustrates the fact that the number of each type of atom is preserved throughout all reactions. Indeed, %in an atomic system,
    for each $(\vec{r},\vec{p})\in R$, $A\in\Delta$, the entry
    \begin{eqnarray}
            (\nonumber\matr{A}\cdot\matr{D}_{\Delta})_{((\vec{r},\vec{p}), A)} & = & \sum\limits_{S\in\Lambda} (\vec{p}(S)-\vec{r}(S))\cdot \vec{d}_S(A)\\\nonumber
            & = & \sum\limits_{S\in\Lambda} 0\quad\quad\textrm{(by -Atomicity)} \\\nonumber
            & = & 0
    \end{eqnarray}
    
    Conversely, let $\Delta$ be a set and $\matr{D}_{\Delta}$ be a $|\Lambda|\times |\Delta|$ matrix with rows labeled by each $S\in\Lambda$ and columns labeled by each $A\in\Delta$. If $\matr{A}\cdot \matr{D}_{\Delta}=\matr{0}$, then by the same arithmetics above and by Definition \ref{primitiveatomic}, $\calC$ is primitive atomic with respect to $\Delta$. It follows that one may rewrite Condition (\ref{primitive1}) of Definition \ref{primitiveatomic} as:
    \begin{enumerate}
        \item $(\exists %\text{ nonempty, finite }
                \Delta\mid 0< |\Delta| < \infty) (\exists \matr{D}_{\Delta}\in\N^{\Lambda\times\Delta}\mid \forall S\in \Lambda, (\matr{D}_{\Delta})_{(S,:)}\neq (0^\Delta)^T)$ 
            such that $\matr{A}\cdot\matr{D}_{\Delta}=\matr{0}$. 
    \end{enumerate}
    
    We use these two expressions for condition $(1)$ interchangeably in future parts of this paper. Furthermore, when $\calC,\Delta$ is clear from the context, we use $\matr{D}$ as shorthand for $\matr{D}_{\Delta}$, and we say \emph{$\calC$ is primitive atomic with respect to $\Delta$ via $\matr{D}$}.
    \end{enumerate}
    \end{rem}
    
    \begin{example}\label{H20Revisited}\normalfont{}
        For the network $\calC$ used in Example \ref{H20}:
        \begin{equation*}\left\{\begin{array}{rl}
            2X + Y + Z\rightarrow& 2W + Z    \\
            X + 2Y + W + Z\rightarrow& Y + W + 2Z,
            \end{array}
            \right.
        \end{equation*}%$$\mathcal{C}=(\{X,Y,W,Z\},\{((2,1,0,1)^T,(0,0,2,1)^T),((1,2,1,1)^T,(0,1,1,2)^T)\})$$ 
        
        and the set of atoms $\Delta = \{H,O\}$, a valid decomposition matrix could be:\begin{equation*}
            \matr{D} = \begin{bmatrix}2 & 0\\ 0 & 2\\ 2 & 1\\ 2 & 2\end{bmatrix},
        \end{equation*} 
    
    With respect to the set of decomposition vectors: $\{\vec{d}_X=(2,0)^T, \vec{d}_Y=(0,2)^T,\vec{d}_Z =(2,1)^T,\vec{d}_W = (2,2)^T\}$. %Again this example show that $\{d_S\}_{S\in\Lambda}$ is non-unique for a certain network. 
    It can also be verfied that \begin{align*}
            \matr{A}\cdot\matr{D} =\begin{bmatrix} -2 & -1 & 2 & 0\\-1  & -1 & 0 & 1 \end{bmatrix}\cdot \begin{bmatrix}2 & 0\\ 0 & 2\\ 2 & 1\\ 2 & 2\end{bmatrix}\\
             = \begin{bmatrix}0 & 0\\0 & 0\end{bmatrix}
        \end{align*} 
        
        So $\calC$ is primitive atomic with respect to $\Delta =\{H,O\}$ via $\matr{D}$.
    \end{example}}%End of E.g. and Rmk on Decompo Matrix

    \opt{normal}{\EgAndRemarkOnDecompositionMatrix}

%The next definitions impose some computational bounds on \emph{atomic} in order to facilitate our future analysis of properties of networks as a language.

%\subsection{Subset Atomic}

The next definition requires all atoms to be species. %themselves.

% We now place an extra restriction on the set of primitive atomic networks, 
% obtaining a new class of ``-atomic'' network. 
% Further than the primitive atomic conditions, we now have $\Delta\subseteq\Lambda$, which means that the atoms are a subset of species, which, if the atoms actually appear in reactions, constrains the possible atomic compositions.  
\begin{defn}[subset atomic]\label{subatom} 
	Let $\calC=(\Lambda,R)$ be a chemical reaction network and let $\Delta \subseteq \Lambda$ be nonempty.	
	We say that $\calC$ is \emph{subset-$\Delta$-atomic} if $\calC$ is primitive atomic with respect to $\Delta$ and,
	for each $S \in \Lambda$:
	\begin{enumerate}
	    \item $S \in \Lambda\cap \Delta = \Delta \implies \vec{d}_S = \{S\}$, and
	    \item $S \in \Lambda\setminus \Delta \implies \|\vec{d}_S\| \geq 2$.
	\end{enumerate}%$ and .
	We say $\calC$ is \emph{subset atomic} if %there exists 
	$\exists \emptyset\neq \Delta \subseteq \Lambda$ such that $\calC$ is subset-$\Delta$-atomic.
\end{defn}

	%%%%%%%%Deliberately kept comments.
	%Alternative statement:%%%%%%%%%%%%
	%. And, %$(\forall S \in\Lambda) 
		   % (\exists \vec{d}_S \in \N^\Delta \setminus \{\vec{0}\}) 
		    %(\forall (\vec{r},\vec{p}) \in R) 
		   % (\forall A \in \Delta)
           % %\sum\limits_{S \in \Lambda} \vec{r}(S) \cdot \vec{d}_S(A) = 
           % %\sum\limits_{S \in \Lambda} \vec{p}(S) \cdot \vec{d}_S(A)$. %
            %
       %     %\item %\label{defn:subatom:condition-decomp}
       % %$(\forall(\vec{r},\vec{p})\in R)\vec{r}\neq\vec{p}$.    %$(\forall W\in \Lambda\setminus \Delta)\ \|\vec{d}_{W}\| \geq 2$; 
           % %$(\forall A\in\Delta)\ \vec{d}_A(A) = 1$ and $(\forall S \in \Lambda \setminus \{A\})\ \vec{d}_A(S) = 0$.
       %  \end{enumerate}
       %%%%%%%End of Alternative statement. Deliberately kept comments. 

By definition \ref{subatom}, no two \emph{atoms} can have the same atomic decomposition, 
but it is allowed that two distinct \emph{molecular} (i.e. non-atom) species to have the same decomposition.
In this case we say the two species are \emph{isomers} (reminiscent of isomers in nature that are composed of the same atoms in different geometrical arrangements). 
As for the requirement that each non-atom species decompose to a vector of size at least $2$, that is to incorporate the idea that generally a molecule should be composed of at least $2$ atoms.

For example, the network $\calC = \{2X+Y+Z\to 2W+Z,X+2Y+W+Z\to Y+W+2Z\}$ mentioned above is subset-atomic: just redefine $\Delta = \{X\}$ and $\vec{d}_X= (1),\vec{d}_Y= (2),\vec{d}_Z= (3),\vec{d}_W= (2)$. One may verify that in the first reaction, each side has $7$ atoms $X$, while in the second each side has $10$.

%\subsection{Reachably Atomic}\label{va111}
%The next
The next definition %imposes further requirement 
further requires that decomposition of each molecular species $S_i$ can be ``\emph{realized}'' via a sequence of reactions, given $\{1S_i\}$ as initial state. As discussed in Subsection \ref{previous111}, this definition was originally developed in \cite{Adleman2014} to help their approach to the Global Attractor Conjecture in the field of mass action kinetics. Considering the convention for most networks, we relax their requirement of reversibility for each reaction.

%\begin{definition}[reachably atomicity]
%\label{visat} 
%A chemical reaction network $\calC=(\Lambda,R)$ is \emph{reachably atomic} if 
   % \begin{enumerate}
     %   \item %it is 
      %  $\exists \Delta\subseteq\Lambda$ and $\matr{D}\in\N^{\Lambda\times\Delta}$ s.t. $\calC$ is \emph{subset atomic} with respect to $\Delta$ via $\matr{D}$, and
        
     %   \item\label{vis222} for each $S\in\Lambda\setminus\Delta$, $\{1S\} \reach \vec{d}_{S}$. 
   % \end{enumerate}
%\end{definition}
 
%\todoi{I don't understand why $\matr{D}$ appears in the above definition. I think it should be the following.}

\begin{defn}[reachably atomic]
\label{visat} 
A chemical reaction network $\calC=(\Lambda,R)$ is \emph{reachably atomic} if 
    \begin{enumerate}
        \item\label{vis111} $\calC$ is subset atomic with respect to some $\Delta \subseteq \Lambda$, and
        
        \item\label{vis222} for each $S\in\Lambda\setminus\Delta$, $\{1S\} \reach \vec{d}_{S}$. 
    \end{enumerate}
\end{defn}
 
Here and wherever necessary, with slight abuse of notation, $\vec{d}_{S}$, which represents the atomic decomposition of $S$, simultaneously represents a configuration in %$\N^\Delta$ 
$\N^{\Delta}$ reachable from $\{1S\}$. Observe that $\calC =  \{2X+Y+Z\to 2W+Z,X+2Y+W+Z\to Y+W+2Z\}$ is not reachably-atomic unless we add the following reactions: $Y\to 2X,Z\to 3X,W\to 2X$.

Condition \ref{vis222} is a strong restriction ensuring some nice properties. For example, the atom set of a reachably atomic network is unique:
\begin{lem}\label{uniquereachable}
    If $\calC=(\Lambda, R)$ is reachably atomic, then the choice of $\Delta$ with respect to which $\calC$ is reachably atomic is unique. 
    Moreover, for each $S\in\Lambda$, $\vec{d}_S$ is unique, 
    i.e.,
    if $\{1S\}\Rightarrow^* \vec{c}\in\N^\Delta$, then $\vec{c} = \vec{d}_S$.
\end{lem}

\newcommand{\DetailedProofOfUniquenessOfDecompositionInReachablyAtomic}{
    Assume for the sake of contradiction that for some reachably atomic network $\calC$, there exist $\Delta_1\neq\Delta_2$ with respect to both of which $\calC$ is reachably atomic, respectively via decomposition matrices $\matr{D}_1$ and $\matr{D}_2$. Note that $(\Delta_1\setminus\Delta_2)\cup (\Delta_2\setminus\Delta_1)\neq\emptyset$. Take $A\in (\Delta_1\setminus\Delta_2)\cup (\Delta_2\setminus\Delta_1)$:
    \begin{enumerate}
         \item If $A\in \Delta_1\setminus\Delta_2$, then $\{1A\}$ is decomposible into some $\vec{c}\in\N^\Lambda\mid [\vec{c}]\subseteq \Delta_{2}$ via a sequence of reactions, with $\|\vec{c}\|_1\geq 2$. There is no way for this sequence of reactions to preserve atoms with respect to $\Delta_1$, for the initial configuration has a single atom $A\in\Delta_1$ while the final configuration has no atom $A$.%at least $2$ atoms. 
         \item Similarly, if $A\in \Delta_2\setminus\Delta_1$, there will be a sequence of reactions violating preservation of atoms with respect to $\Delta_2$.
    \end{enumerate} 
 
    We next prove the uniqueness of decomposition vectors $\vec{d}_S$ for all $S\in\Lambda$, or equivalently, the uniqueness of decomposition matrix $\matr{D}$. Suppose not, then %for a certain choice of $\Delta\subseteq \Lambda$, 
    there exists $S\in\Lambda\setminus\Delta$ s.t. $\{1S\}\reach \vec{y}_1,\vec{y}_2\in\N^\Lambda$, $\vec{y}_1\neq \vec{y}_2$ and $[\vec{y}_1],[\vec{y}_2]\subseteq \Delta$, via reaction sequences $\vec{s}_1,\vec{s}_2$ %of reactions, 
    respectively. Assume without loss of generality that $\vec{s}_1$ preserves the number of atoms in each reaction, which means $\vec{y}_1=\vec{d}_S$. Then there must be one or more actions in $\vec{s}_2$ that does (do) not preserve the number of atoms, for $\vec{s}_1,\vec{s}_2$ share the same initial configuration $\{1S_i\}$ yet reach different final count of atoms, while no atoms are allowed to be isomeric to each other.
}

\begin{proof}
The intuition %of the proof %by contradiction 
is to show that %should there exist an element $A$ in the set difference between two atom sets
should there exist $\Delta_1\neq\Delta_2$ and without loss of generality, assume $\exists A\in \Delta_1\setminus\Delta_2$, then the decomposition of $A$ with respect to $\Delta_1$ violates the preservation of atoms in $\Delta_2$. 
    \opt{sub,subn}{For details, see Section~\ref{AppAA}.}
    \opt{normal}{\DetailedProofOfUniquenessOfDecompositionInReachablyAtomic}
\end{proof}

Conservation laws in ``-atomic''
networks reminds us of a more familiar type of conservation law, which is mass conservation. The next section exhibits some observations on the relationship between these two types of conservation laws.

\section{Mass-Conservation and primitive atomicity}\label{MCPAEquiv}

This section shows that ``primitive atomic'' and ``mass conserving'' are equivalent concepts.
We first formalize what it means for a network to conserve mass:%satisfy %the latter which describes the intuition %concept, that
%``mass can be neither created nor destroyed'':
%``mass does not come from nowhere, nor does it vanish into nowhere'' under our model:

\begin{defn}[Mass Conserving]\label{MCCCCCC}
A chemical reaction network $\calC = (\Lambda, R)$ is \emph{mass conserving} if
    $$
        (\exists \vec{m}\in\R_{>0}^{\Lambda})(\forall (\vec{r},\vec{p}) \in R)
        \sum\limits_{S \in \Lambda} \vec{r}(S) \cdot \vec{m}(S) = 
        \sum\limits_{S \in \Lambda} \vec{p}(S) \cdot \vec{m}(S)
    $$ 
Equivalently, if $\matr{A}$ is the stoichiometric matrix in Definition~\ref{stoi}, 
then $\calC$ is mass conserving if  $(\exists \vec{m}\in\R_{>0}^{\Lambda})\ \matr{A}\cdot \vec{m} = \vec{0}.$ 
We call $\vec{m}$ a \emph{mass distribution vector}.
\end{defn}

Using our familiar example, $\calC =  \{2X+Y+Z\to 2W+Z,X+2Y+W+Z\to Y+W+2Z\}$ is mass conserving with respect to $\vec{m}=(0.5,1,1.5,1)^T$. ``Mass Conserving'' captures the feature that for every reaction in $C$, the total mass of reactants are equal to the total mass of products. Difference between the definitions of Mass Conserving and Primitive Atomic (as well as all ``-atomic'' definitions descended therefrom) become clear if we compare the matrix form of their respective conservation laws: mass conservation requires a single conservation relation ($\matr{A}\cdot \vec{m}=0^{|R|}$), while primitive atomicity requires $|\Delta|$ of them ($\matr{A}\cdot\matr{D}=\matr{0}$ where $\matr{D}$ is a $|\Lambda|\times |\Delta|)$ matrix.

However, %it is very %intuitive
%natural
apparently these two conservation laws are closely related. In fact, the freedom of defining $\Delta$ independent of $\Lambda$ provides us a choice for making $\Delta$ a singleton, which enables us to prove the following equivalence:

\newcommand{\StatementofMCPAEquivalence}{For any network $\calC$, $\calC$ is primitive atomic $\Leftrightarrow$ $\calC$ is mass conserving. Further, there exists an $O(|\lag \matr{A}\rag|^5)$ algorithm to decide if $\calC$ is primitive atomic, with $\matr{A}$ the stoichiometric matrix of $\calC$.
}%End of Statement of MCPA equivalence. Used also in the Appendix, so no change

\begin{prop}\label{primitivemcequiv}
\StatementofMCPAEquivalence
\end{prop}

\newcommand{\RemarkOnMCPAEquivalence}{\begin{rem}\label{EquivalenceNOTTrueForSubsetAtomic}\normalfont{}The equivalence with Mass Conserving is \textbf{not true} for \emph{subset atomic} networks: Consider $\calC = (\{X,Y\},\{((1,0),(0,1)\})$. This network involves a single reaction: $X\rightarrow Y$. There are only three choices for the set of atoms $\Delta$, which is now required to satisfy $|\Delta|>0, \Delta\subseteq \Lambda$: $\Delta = \{X\}$, $\Delta = \{Y\}$ or $\Delta = \{X,Y\}$, respectively. 

$\calC$ is Mass Conserving, the mass distribution vector being $\vec{m}(X)=\vec{m}(Y)=1$; Nonetheless, no choice of $\Delta$ makes $\calC$ subset atomic. 
Take $\Delta = \{X,Y\}$ for example, and consider the atom $X$: 
the reactants have a single atom of $X$ but the products have no $X$ (since $Y$ is itself an atom it cannot be composed of $X$).
If $\Delta=\{X\}$, then $Y$ is not an atom, but the definition of atomic decomposition requires that non-atomic species have a decomposition with at least two atoms, which means that since $X$ is the only atom, $\vec{d}_{Y}$ contains at least two $X$'s, while the reactants have only one $X$.
A similar argument applies if $\Delta=\{Y\}.$
\end{rem}
}%End of MCPA Remark

\newcommand{\DetailedProofOfMCPAEquivalence}{\begin{enumerate}
  \item  primitive atomic $\Leftrightarrow$ mass conserving:

\begin{enumerate}
    \item  mass Conserving $\Rightarrow$ primitive Atomic:

   If a chemical reaction network $\calC$ is mass conserving, then there exists a mass distribution vector $\vec{m}\in\R_{>0}^{\Lambda}$ s.t. $\matr{A}\cdot \vec{m}=0^{|R|}$, with $\matr{A}$ the stoichiometric matrix of $\calC$. We shall 
    exhibit a constructive way of finding a \emph{rational} solution to the linear system $\matr{A}\cdot \vec{m} = 0^{|R|}$, thereby enabling further %algebraic 
    manipulation to finally yield an integral solution that could be used to construct a desired decomposition matrix $\matr{D}$. \begin{enumerate}
         \item \label{sergeiss} Knowing the existence of the distribution vector $\vec{m}\in\R_{>0}^\Lambda$
        \ s.t. $\matr{A}\cdot\matr{m}=0^{|R|}$, we  
         run Sergei's ``strictly positive solution finder for linear system'' algorithm\cite{chubanov2015polynomial, chubanovprivatecommunication} to find $\vec{m}'\in\Q^{\Lambda}_{>0}\subseteq \R^{\Lambda}_{>0}$, s.t. 
       $\matr{A}\cdot \vec{m}' = 0^{|R|}$. The algorithm has complexity $O(\max\{|\Lambda|^4,|R|^4\}\cdot |\lag\matr{A}\rag|) = O(|\lag\matr{A}\rag|^5)$\cite{chubanov2015polynomial, chubanovprivatecommunication}.

        \item \label{sergeiss2} Now that we have obtained a solution vector $\vec{m'}\in\Q_{>0}^{\Lambda}$, for each $S\in\Lambda$ we may write $\vec{m'}(S) = \frac{a_S}{b_S}\mid a_S\in\Z,b_S\in \Z_{\neq 0}$. Compute $c =  \textrm{lcm} \{b_S\}_{S\in\Lambda}$ 
        , and let $\vec{m}'' = \vec{m}'\cdot c$. Since $\matr{A}\cdot\vec{m''}=c\cdot (\matr{A}\cdot\vec{m}')=0^{|R|}$ 
        , by definition $\vec{m}''\in \N^{|\Lambda|}_{>0}$ is also a valid mass distribution vector.

        Complexitywise, one may apply the %arity-$2$ 
        binary Euclidean Algorithm (Let us denote the algorithm as $\gcd(a,b)$) for $|\Lambda|-1$ times, finding $\gcd(b_{S_1},b_{S_2}), \gcd(b_{S_1},b_{S_2},b_{S_3}) = \gcd(\gcd(b_{S_1},b_{S_2}),b_{S_3})$, $\cdots$, $\gcd(b_{S_1},b_{S_2},\cdots, b_{S_{|\Lambda|}})$ dynamically, and then compute $c=$lcm$\{b_S\}_{S\in\Lambda} = \dfrac{\prod_{S\in\Lambda}b_S}{\gcd\{b_S\}_{S\in\Lambda}}$. %to compute the $\gcd$.
        Since the complexity of %two-input 
        Euclidean Algorithm $\gcd(a,b)$ is %linear to the binary size of the sum of the input
        $O((\max\{|\lag a\rag|,|\lag b\rag|\})^3)$\cite{knuth1997seminumerical}, 
        complexity for computing $\gcd(b_1,\cdots,b_{S_{|\Lambda|}})$ as well as $c=$lcm$(b_1,\cdots,b_{S_{|\Lambda|}})$ is $O(|\Lambda|\cdot ( \max\{|\lag\vec{m}'(S)\rag|\})^3)$. 
        
        Lastly, it is shown in the footnote that computing $\vec{m}''=\vec{m}'\cdot c$ expands the binary size of $\vec{m}'$ by at most a factor of $|\Lambda|$,\footnote{To make the notations a little simpler, let the $i$-th entry of $\vec{m}'$ be written as $\frac{c_i}{d_i}$, so the the numerator and denominator are $\log(|c_i|)$ and $\log(|d_i|)$ bits long, respectively ($\forall i\in [1,|\Lambda|]$). Define $$t:=\underbrace{\sum_{j=1}^{|\Lambda|}(\log(c_j)+\log(d_j))}_{\textrm{binary size of }\vec{m'}}$$ 
        %= O(|\Lambda|^4L) $$ . 
        In the worst case where $d_i,d_j(\forall i,j)$ are pairwise co-prime, multiplying each entry with $\gcd(d_1,d_2,\cdots d_{|\Lambda|})$ is equivalent to first setting all $d_i$'s to $1$, then expending each $c_i$ by $(\sum_{j=1}^{|\Lambda|}\log(d_j)) - \log(d_i)$ bits. The net effect is expanding the size of the input by \begin{eqnarray*}\sum_{i = 1}^{|\Lambda|}(((\sum_{j=1}^{|\Lambda|}\log(d_j)) - \log(d_i)) - (\log{d_i}-1)) &=&\sum_{i=1}^{|\Lambda|}((\sum_{j=1}^{|\Lambda|} \log(d_j))-2\log(d_i)+1)\\&= & (|\Lambda|-2)(\sum_{j=1}^{|\Lambda|}\log(d_i)) + |\Lambda|\\ &<& |\Lambda|\cdot t + |\Lambda|\\ &=& O(|\Lambda|\cdot t) \end{eqnarray*} bits, as desired.} so altogether, the complexity for this Step (\ref{sergeiss2}) is $O(|\Lambda|^2\cdot ( \max\{|\lag\vec{m}'(S)\rag|\})^3)$. Since 
        $\vec{m}'$ is obtained from Step (\ref{sergeiss}) with size-bound $|\lag\matr{A}\rag|^5$, we have \begin{equation}\label{boundofconv}O(|\Lambda|^2\cdot (\max\{|\lag\vec{m}'(S)\rag|\})^3) =O(|\lag A\rag|^{17}) 
        \end{equation}
         
 \item\label{sergeiss3} Output $\Delta = \{A_1\}$ where $A_1\not\in\Lambda$, and $\matr{D} = 2\vec{m}''$.

 Note that $\matr{D}$ is a length-$|\Lambda|$ column vector when $|\Delta| = 1$, and so is $2\vec{m}''$, hence their dimensions match. 
    \end{enumerate}

    Intuitively, the last step is assigning mass $1$ to the single atom $A_1$, and then decomposing each $S\in\Lambda$ into $2\vec{m}''(S)$ of $A_1$'s. We formally verify the correctness of this output: \begin{enumerate}
        \item $\matr{A}\cdot \matr{D} = 2(\matr{A}\cdot\vec{m}'') = 0^{|R|}$, so Condition (\ref{primitive1}) for primitive atomicity is satisfied;
        %\item each entry of $\matr{D}$ is at least $2$, so $\|\vec{d}_S\|_{S\in\Lambda}\geq 2$ for all $S\in\Lambda\setminus\Delta = \Lambda$, proving Condition (\ref{primitive2}); 
        \item  $A_1$ as the single atom in $\Delta$ is used in the decomposition of each molecular species, verifying Condition (\ref{primitive3}) as well.
    
    \end{enumerate} 

Complexitywise, Step (\ref{sergeiss2}) dominates, so (\ref{boundofconv}) is the complexity of the whole algorithm.
    \item primitive atomic $\Rightarrow$ mass conserving:%Proof of ($\ref{1katom1b}$):
    
    We shall prove that given any chemical reaction network $\calC$ primitive atomic with respect to $\Delta$ via decomposition matrix $\matr{D}$, there exists a polynomial time algorithm to compute a valid mass distribution vector $\vec{m}$. For each $S\in\Lambda$, let $\vec{m}(S) = \sum\limits_{A\in\Delta} \vec{d}_S(A) = \|\matr{D}_{(S,:)}\|_1$; that is, we assign mass $1$ to all atoms, making the mass of a molecular species equal to the total count of atoms in its atomic decomposition. %species $S$ as its mass. 
    Then for each $(\vec{r},\vec{p})\in R$, the entry 
    \begin{eqnarray*}
    (\matr{A}\cdot \vec{m})_{(\vec{r},\vec{p})} &=& \sum\limits_{S\in\Lambda}(\matr{A}_{((\vec{r},\vec{p}),S)}\cdot \sum\limits_{A\in\Delta}\vec{d}_S(A))\\
    &=&\sum\limits_{S\in\Lambda}\sum\limits_{A\in\Delta}(\matr{A}_{((\vec{r},\vec{p}),S)}\cdot \matr{D}_{(S,A)})\\&=&\sum\limits_{A\in\Delta}\sum\limits_{S\in\Lambda}(\matr{A}_{((\vec{r},\vec{p}),S)}\cdot \matr{D}_{(S,A)})\\
    &=& \|(\matr{A}\cdot\matr{D})_{((\vec{r},\vec{p}),:)}\|_1\\
    &=&0,
    \end{eqnarray*}
as desired. Note that $\|\lag\vec{m}\rag\|_1 = O(|\lag \matr{D}\rag|) %|\Lambda|\cdot (n\cdot k)$
$, which gives the upper bound of complexity.
\end{enumerate}

\item There exists an $O(|\lag A\rag|^5)$ algorithm to decide if $\calC$ is primitive atomic. %primitive atomicity is polynomial time decidable. 
    
    By the equivalence relationship shown right above, deciding primitive atomicity is equivalent to deciding mass conservation, which is in turn equivalent to deciding if there exists a strictly positive solution to the linear system $\matr{A}\cdot \vec{m} = 0^{|R|}$ with $\matr{A}$ the stoichiometric matrix of $\calC$. The latter problem is decidable by an $O(|\lag\matr{A}\rag|^5)$ algorithm, as mentioned in Step \ref{sergeiss} \cite{chubanov2015polynomial, chubanovprivatecommunication}.\qed%\qed
\end{enumerate}
}%End of Detailed proof of MCPA equiv 

\begin{proof}
Intuitively, the ``$\implies$'' direction is shown by assigning mass $1$ to each atom, as ``homogenizing'' the atoms preserves the original conservation law; for the ``$\impliedby$'' direction, one may %convert the mass distribution vector (originally in $\mathbb{R}^{{\Lambda}}$) into an integral vector, then interpret it as a $|\Lambda|\times 1$ decomposition matrix; so 
essentially %we are creating 
create a $\Delta$ of cardinality $1$ with respect to which the network is primitive atomic. The proof also reflects the difference in number of conservation relations addressed two paragraphs above. %converting a mass distribution vector into integral and interpreting it as composition vectors proves the ``mass conservation $\Rightarrow$ primitive atomicity'' direction, while assigning mass $1$ to each atom for a primitive atomic chemical reaction network gives the mass distribution vector.  
    \opt{sub,subn}{See Section \ref{AppAA} in the Appendix for details and some more remarks.}%See Paragraph \ref{statementofMCPAEquivalence} in the appendix for details and some remark.}
    \opt{normal}{\DetailedProofOfMCPAEquivalence
    \RemarkOnMCPAEquivalence}
\end{proof}

Recall that subset atomicity imposes the restriction that $\Delta\subseteq\Lambda$. As we'll show in the following section, this single restriction increases the computational complexity of the decision problem ``is a network `(\emph{prefix})-atomic' ''. %deciding subset atomicity, compared to primitive atomicity.
\section{Complexity of Subset Atomic}\label{CSFA}
 We shall determine in this chapter the computational complexity for deciding the subset atomicity of networks. First, we define the relevant languages: %of encodings of chemical reaction networks with subset atomicity:
 
\begin{defn}\label{SubAtom}
 We define the following languages:
 \begin{eqnarray*}
 \SubsetAtomic &=& \{\lag\Lambda,R\rag\mid (\exists\Delta\subseteq \Lambda)((\Lambda, R)\textrm{ is subset atomic with respect to }\Delta)\}\\
 \SubsetFixedAtomic &=& \{\lag\Lambda,R,\Delta\rag\mid (\Lambda, R)\textrm{ is subset atomic with respect to }\Delta\}
 \end{eqnarray*}
 \end{defn}

 By definition, $\SubsetAtomic$ is the language whose elements are the encoding of a \emph{subset atomic} chemical reaction network.  \SubsetFixedAtomic, on the other hand, is the language consisting of the encoding of a (network, atom set) pair where the network is \emph{subset atomic} with respect to the given atom set. In this section we determine the complexity classes of these languages.
\subsection{$\SubsetFixedAtomic$ and $\SubsetAtomic$ are in NP}
 
It is not %obviously true
immediately obvious that there exists a short witness for either language (which if true would %have proved
imply that both languages are in $\NP$ immediately), so we reduce $\SubsetFixedAtomic$ to $\textsc{Integer-Programming}$, which is %known to be 
in $\NP$~\cite{papadimitriou1981complexity}.

\newcommand{\StatementOfLemmaSSFixedAReducibleToIP}{
\SubsetFixedAtomic %\ is polynomial-time many-one reducible to 
$\leq_{m}^{p}\textsc{Integer-Programming}$ (hereinafter, ``\IP'').
}%End of Statement of Lemma Subset Fixed Atomic reducible to IP. In the Appendix. and full.

\begin{prop}\label{SFANP}\StatementOfLemmaSSFixedAReducibleToIP
\end{prop}

\newcommand{\DetailedProofSSFixedAReducibleToIP}{%Only in the Appendix and full
    We transform the input into the encoding of the following equivalent linear system: 
    In the following, $b_1,b_2,\cdots b_{m-n},c_1,c_2,\cdots,c_n\in \N$ are slack variables so that we may express an inequality as an equality together with a nonnegativity constraint on $b_s$, $c_a$ ($1\leq s\leq m-n,1\leq a\leq n$).

    \begin{equation}\label{series}\begin{array}{cccccc}
    \textrm{for }a \in [1,n]: & & & & &\\
    & \textrm{for }r \in [1,k]: & & & & \\
     & & & \sum\limits_{s=1}^{m} \vv_r(S_s)\cdot x_{sa}   & = & 0 \\
     %& \} & & & &\\
     %\} & & & & & \\
    \textrm{for }s \in [1, m-n]: & & & & & \\
    & & & \left( \sum\limits_{a=1}^{n}x_{s,a} \right) - b_{s} &=& 2\\
    & & & b_s &\geq & 0\\
    \textrm{for }a \in [1, n]: & & & & & \\
    & & & \left( \sum\limits_{s=1}^{m-n}x_{s,a} \right) - c_{a} &=& 1\\
    & & & c_a &\geq & 0\\
     %  \} & & & & & \\
    \textrm{for }a \in [1, n]: & & & & & \\
    & \textrm{for }a' \in [1,n]: & & & & \\
    & & \textrm{if }a=a'& & & \\
    & & & x_{a'+m-n,a} &=& 1\\
    %& & \} & & &\\
    & & \textrm{else} & & &\\
    & & & x_{a'+m-n,a} &=& 0\\
    %& &\} & & &\\
    \textrm{for }s \in [1,m]: & & & & & \\
    &\textrm{for }a \in [1,n]: & & & & \\
    & & & x_{sa}& \geq& 0\\
    \end{array}\end{equation}

    The equivalence between the linear system and the subset atomicity of $(\Lambda,R,\Delta)$ follows from Definitions \ref{SubAtom} and \ref{subatom}. In fact, the first equation of (\ref{series}) asks the number of each atom to be preserved across each reaction; the second equation and third inequality prescribes that for all non-atomic species $S\in\Lambda\setminus\Delta$, $||\vec{d}_S||_1\geq 2$; the fourth equation and fifth inequality prescribes that for all atomic species $A\in\Delta$, $\sum_{S\in\Lambda\setminus\Delta}\vec{d}_S(A)\geq 1$, which is equivalent to saying that each $A$ appears in at least one atomic decomposition for some molecules; the sixth and seventh equation translates to $\forall A\in\Delta$, $\vec{d}_A =\vec{e}_A$ (recall $\ve_A$ is the unit vector that is 1 on $A$ and 0 on everything else); and the last inequality restricts atomic counts in species to be nonnegative integers.

    Note, also, that $\vec{d}_S\neq 0^{\Delta}$ has been ensured by the respective restrictions on decomposition for non-atom species and atoms.

    Let us analyse the complexity of the reduction. To construct a constraint system 
    \begin{equation*}\label{IntMat}
         \matr{A}\vec{x}=\vec{b},\ %\vec{x}\geq 0,\ 
         \vec{x}\in\N^n
    \end{equation*}
 
    we first observe that 
    \begin{equation*}
        |\vec{x}| = |\{x_{s,a}\mid s\in[1,m], a\in[1,n]\}| + |\{d_s\mid s\in [1,m-n]\}| + |c_a\mid a\in [1,n]| = mn + m
    \end{equation*}

    and that 
    \begin{eqnarray*}
         \nonumber |\vec{b}| &=& \textbf{number of equations in (\ref{series})}\\\nonumber
         & = & nk + m + n^2
    \end{eqnarray*}

    which means the matrix $\matr{A}$ is $(nk+m+n^2)\times (mn+m)$, a polynomial in $m,n,k$. Further, each entry of $\vec{b}$ is an integer in $[0,2]$, while the absolute value of each integral entry of $\matr{A}$ is bounded by $\max\{\max\limits_{r\in R, s\in[1,m]}\{\vec{v}_{r}(S_s)\}, 1\}$. This shows that the linear system is of size polynomial in $m,n,k$ and the binary size of $\lag\Lambda,R,\Delta\rag$, and the process reducing $\lag\Lambda,R,\Delta\rag$ to the linear system is also of time polynomial to the same parameters.

    Since $\lag\Lambda, R,\Delta\rag \in$ 
    $\SubsetFixedAtomic$ if and only if $\Delta\subseteq \Lambda$ and the linear system (\ref{series}) has a solution,
    and since constructing the encoding of the linear system (\ref{series}) takes polynomial time, 
    we conclude that %\emph{Subset-Fixed-Atomic} 
    $\SubsetFixedAtomic\leq_{m}^{p}\IP$. %is polynomial-time many-one reducible to Integer Programming.  \qed%\renewcommand{\qedsymbol}{\square}
}%End of Command: Detailed Proof of SSFA <= IP.
 
\newcommand{\DetailedProofOFSAInNP}{From Corollary \ref{cor:SFANPNP} we know that
 there exists a polynomial time verifier $V'$ for the language \SubsetFixedAtomic, who takes an instance $\langle \Lambda, R,\Delta\rag$ and a witness $\lag\matr{D}\rag$, the latter being encoding of a decomposition matrix $\matr{D}$, and \textsc{accepts} (resp. \textsc{rejects}) if and only if $\langle \Lambda, R,\Delta\rag\in$ \SubsetFixedAtomic\ (resp. $\langle \Lambda, R,\Delta\rag\not\in$ \SubsetFixedAtomic) via $\matr{D}$.

 We %first 
 exhibit a polynomial time verifier $V$ for \SubsetAtomic. On instance $c = \lag\Lambda, R\rag$ and witness $w = \lag\Delta,\matr{D}\rag$, the algorithm $V$:
 
 \begin{algorithm}[H]

\emph{Parses $\lag c,w\rag$ into $\lag c',w'\rag$ where $c'=(\Lambda, R,\Delta)$ and $w' = \matr{D}$}\;

\emph{Runs $V'$ on $\lag c',w'\rag$ and echos its output}\;
 \caption{Verifier $V$ for \SubsetAtomic}\label{VSFA} 
\end{algorithm}

A valid witness $w=\lag\Delta,\matr{D}\rag$\footnote{In fact the witness could even be a single $\lag\matr{D}\rag$, as one may read each row of $\matr{D}$ and decide if the species represented by that row is a molecule (sum of entries in the row is at least $2$) or an atom (the row would be a unit vector) immediately, thereby determining $\Delta$.} has a size polynomial of the input size since a valid $\Delta\subseteq \Lambda$, while Corollary \ref{cor:SFANPNP} ensures that a valid $\matr{D}$ has size polynomial of the input as well. Step $1$ therefore takes linear time and by Corollary \ref{cor:SFANPNP} again, step $2$ takes polynomial time too. \qed%\renewcommand{\qedsymbol}{\square}
}

\begin{proof}
    The proof is done by exhibiting a polynomial time algorithm to transition the %axioms/constraints
    conditions in Definition \ref{subatom} %the definition of subset atomicity 
    into a linear system. Note that the atom set $\Delta$ is given as input. %, and finding that such a transition is polynomial time.%system of linear equations or inequalities.%, and then solving this system.  
    \opt{sub,subn}{For details, see Section \ref{AppAA}.}%.Paragraph \ref{statementOfLemmaSSFixedAReducibleToIP}.}
    \opt{normal}{\DetailedProofSSFixedAReducibleToIP}
\end{proof}

%We immediately obtain t
%The following corollary is straightforward: 
%\begin{corollary}\label{cor:SFANPNP}
\begin{cor}\label{cor:SFANPNP}
    \SubsetFixedAtomic,\ \SubsetAtomic\ $\in \NP$.
\end{cor}%ollary}

\begin{proof} 
    %Papadimitriou  
    It is proved (e.g., in \cite{papadimitriou1981complexity}) that $\IP\in\NP$, hence so is %given a $m\times n$ integer matrix $\matr{A}$ and $m$-vector $\vec{b}$, the problem %of deciding whether 
    %of whether there exists $\vx \in \N^n$ such that
   % \begin{equation*}
    %     \matr{A}\vec{x}=\vec{b},\ %\vec{x}\geq 0,\ 
   %      \vec{x}\in\N^n
  %  \end{equation*}
  %  is in $\NP$. %Since \SubsetFixedAtomic\  $\leq_{m}^P$ \IP, b
   % By the polynomial time reduction in Proposition \ref{SFANP}, the conclusion that 
   $\SubsetFixedAtomic$.%\in\NP$. %follows.
%\end{proof}

%\begin{corollary}
%\label{cor:SFANPNP}
 %   \SubsetAtomic\ $\in$ $\NP$.
%\end{corollary}

%\begin{proof}
The proof that $\SubsetAtomic\in\NP$ is given by an polynomial time verification algorithm using the polynomial-time verifier of $\SubsetFixedAtomic$ as an oracle and taking as witness both the atom set and decomposition matrix.     \opt{sub,subn}{For details, see Section \ref{AppAA}.}%Paragraph \ref{corollarySSAInNPStatementOnly}.}
    \opt{normal}{\DetailedProofOFSAInNP}
\end{proof}

\subsection{Subset-Fixed-Atomic is NP-hard}

Our proof shall be based on reduction from $\MonotoneOneThree$. 
Recall that a monotone $3$-CNF $C$ is a conjunctive normal form with no negations, 
and a $1$-in-$3$ satisfying assignment for $C$ is an assignment of Boolean values to all variables such that for each clause in $C$, exactly one variable is assigned true.

As a well-established result, the following language is $\NP$-complete~\cite{GareyJohnson79}.
\begin{eqnarray}
     \nonumber\MonotoneOneThree &= &\{\lag V,C \rag\mid %B\subset \N, |B|<\infty,
     C\textrm{ is a monotone }3\textrm{-CNF over }V=\{v_1,v_2,\cdots,v_n\}\textrm{,}\\\nonumber
     && \textrm{ and there exists a }1\textrm{-in-}3 \textrm{ satisfying assignment for C}\}
\end{eqnarray}
     
     %We shall prove that $\MonotoneOneThree$ is polynomial time many-one reducible to $\SubsetFixedAtomic$.
     
\begin{prop}\label{PartitionSFA}
     $\MonotoneOneThree%$ %is polynomial-time many-one reducible to $
     \leq_{m}^{p}\SubsetFixedAtomic$.
\end{prop}

\newcommand{\DetailedProofMonotoneOneThreePolyRedSSFA}{For each instance $\lag V,C\rag$ of $\MonotoneOneThree$, 
   
     let $\Delta = \{T,F,P,Q\}$, $\Lambda = \{S_1,S_2,\cdots, S_n, X_1, X_2,\cdots, X_n\}\cup\Delta$. 
     
     To construct $R$, we denote $C=c_1\wedge c_2\wedge \cdots \wedge c_k$. For the $i$-th clause $c_i\in C$, let $v_{ij}$ denote the $j$-th literal of $c$. Same indexing convention applies for $\{S_i\}_{i=1}^n$ and $\{X_i\}_{i=1}^n$, hence each $S_{ij}$ (resp. $X_{ij}$) denotes an element in $\{S_i\}_{i=1}^{n}$ (resp. $\{X_i\}_{i=1}^{n}$)%, given that $C$ is monotone
     .\footnote{For example, for $V = \{v_1,v_2,\cdots, v_5\}$, $C = (v_1\vee v_3\vee v_4)\wedge(v_3\vee v_2\vee v_5)$, $v_{11}=v_1, v_{12}=v_3, v_{13}=v_4$, $\cdots$, $v_{23} = v_5$. Correspondingly, $S_{11}=S_1,X_{11}=X_1$, $\cdots$, $S_{23}=S_5, X_{23}=X_5$.} The set $R$ contains the following reactions\footnote{To continue the example in the previous footnote, the set of reactions shall be:
     \begin{eqnarray}
             \nonumber 3P+2F+T&\rightarrow& S_1+ S_3 + S_4\\\nonumber
 3P+2F+T&\rightarrow& S_3 + S_2 + S_5\\\nonumber
 3Q+2F+T&\rightarrow& X_1 + X_3 + X_4\\\nonumber
 3Q+2F+T&\rightarrow& X_3 + X_2 + X_5\\\nonumber
 S_i + Q &\rightarrow& X_i + P\ (i = 1,2,\cdots, 5)
     \end{eqnarray}}:
         \begin{eqnarray}
        3P+2F +T %\left( \sum\limits_{i=1}^n c_i \right) 
          & \rightarrow &S_{m1}+S_{m2}+S_{m3}\ (\forall m\in[1,k])\label{r1} \\
        3Q+2F +T %\left( \sum\limits_{i=1}^n c_i \right) 
          & \rightarrow& X_{m1}+X_{m2}+X_{m3}\ (\forall m\in[1,k])\label{r2} \\
        S_i + Q &\rightarrow & X_i + P\ (\forall i\in[1,n])\label{r3} %S_i&\not\rightarrow& a + b, \forall i
        \end{eqnarray}

     Note that we only need to construct $4+2n$ species and $2k+n$ reactions whose coefficients are constant, so this transformation is polynomial in time and space. %$2\mid \|B\|_1$ so $\frac{\|B\|_1}{2}\in\N$. 
     We argue that $\lag V,C \rag\in\MonotoneOneThree$ if and only if $\lag\Lambda, R,\Delta\rag\in \SubsetFixedAtomic$.
     
     $\Rightarrow$: If $\lag V,C \rag\in\MonotoneOneThree$, there exists a $\phi:V\rightarrow \{0,1\}$ under which $\exists (n_1,n_2,\cdots, n_q)\subsetneq (1,2,\cdots, n)$ s.t. 
     $\phi(v_{n_i}) = 1\ (\forall i\in[1,q])$, $\phi(v_{j})=0 \ (j\in([1,n]\setminus(n_1,n_2,\cdots, n_q)))$, and for each $c_i\in C$, exactly one in three of its literals evaluates to $1$. Let: \begin{eqnarray}
             \nonumber
             \Lambda_{TP} &=& \{S_{n_j}\mid j\in[1,q]\}\nonumber\\\nonumber
           \Lambda_{FP} &=& \Lambda\setminus (\Delta\cup \Lambda_{TP}\cup \{X_i\mid i\in[1,n]\})\nonumber\\\nonumber  \Lambda_{TQ} &=& \{X_{n_j}\mid j\in[1,q]\}\nonumber\\\nonumber
           \Lambda_{FQ} &=& \Lambda\setminus (\Delta\cup \Lambda_{TQ}\cup \Lambda_{TP}\cup\Lambda_{FP}) 
     \end{eqnarray}
     
     Then $\lag\Lambda, R, \Delta\rag\in\SubsetFixedAtomic$ because with all atoms listed in the order: 
  $\{T,F,P,Q\}$, one may make the following decomposition:
     \begin{eqnarray}\nonumber
     \vec{d}_{U}%\restriction\Delta 
     &=&(1,0,1,0)^T,\ \forall U\in\Lambda_{TP}\nonumber\\
     \vec{d}_{V}%\restriction\Delta 
     &=&(0,1,1,0)^T,\ \forall V\in\Lambda_{FP}\nonumber\\
     \vec{d}_{W}%\restriction\Delta 
     &=&(1,0,0,1)^T,\ \forall W\in\Lambda_{TQ}\nonumber\\
     \vec{d}_{Z}%\restriction\Delta 
     &=&(0,1,0,1)^T,\ \forall Z\in\Lambda_{FQ}\nonumber
     \end{eqnarray}%decompose all species in $\Lambda_a$ into $2a$ and all species in $\Lambda_b$ into $2b$, or vice versa.
     
     Because of the way $\{n_j\}_{j=1}^{q}$ was taken, for each reaction in (\ref{r1}), exactly one of the product species decompose to $1T$ and $1P$, with the other two decomposing to $1F$ and $1P$. Similar argument applies for reactions in (\ref{r2}), substituting $X_i$ for $S_i$ and $Q$ for $P$. Arithmetics show that all three reactions %in (\ref{PartRed})
     (\ref{r1}) through (\ref{r3}) preserve the number of atoms, each atom appears in the decomposition of at least one molecular species, and the number of atoms in the decomposition of each species complies with the Definition \ref{SubAtom}. Therefore $\lag \Lambda, R,\Delta\rag \in \SubsetFixedAtomic$.
     
     $\Leftarrow$: If $\lag \Lambda, R,\Delta\rag \in \SubsetFixedAtomic$, ($\ref{r1}$) ensures that each molecular species contains exactly two atoms. That is because for each $i\in [1,n]$, \begin{equation}\label{defnsum}\vec{d}_{S_i}(T) + \vec{d}_{S_i}(F)+\vec{d}_{S_i}(P) + \vec{d}_{S_i}(Q) \geq 2\end{equation}
     
     by Definition \ref{SubAtom}, so for each $m\in[1,k]$, \begin{equation}
         \label{sumdsS}
 \sum\limits_{j=1}^{3}(\vec{d}_{S_{mj}}(T) + \vec{d}_{S_{mj}}(F)+\vec{d}_{S_{mj}}(P) + \vec{d}_{S_{mj}}(Q)) \geq 3\times 2 = 6
     \end{equation}
     
     However the total number of atoms on the left hand side of \ref{r1} is exactly $6$, meaning the equal sign has to be taken everywhere in (\ref{sumdsS}) for any $m\in[1,k]$, forcing (\ref{defnsum}) to take equal sign as well. %todo{SP: Polish.}
     
     Similarly, $(\ref{r2})$ ensures $||\vec{d}_{X_i}||_1 = 2$ for each $i\in[1,n]$.
     
     The reaction series (\ref{r3}) implies that each $S_i$ has at least one $P$ and each $X_i$ has at least one $Q$. Furthermore, %if 
     \begin{enumerate}
         \item if any $S_i$ decomposes to $2P$, its corresponding $X_i$ shall be composed of $PQ$, contradicting %the second reaction 
     (\ref{r2}) which says that no $X_i$ contains any $P$;
         \item if any $S_i$ decomposes to $PQ$, 
         it contradicts with (\ref{r1}) which says that no $S_i$ contains any $Q$.
     \end{enumerate} 
     
     Therefore all $S_i$ shall decompose to %\emph{exactly}
     either $\{1F,1P\}$ ($(0,1,1,0)^T$) or $\{1T, 1P\}$ ($(1,0,1,0)^T$). Correspondingly, $X_i$ decompose to either $(0,1,0,1)^T$ or $(1,0,0,1)^T$.
     
     Construct $\phi$ such that $\phi(v_j) = 1$ for all $v_j\in \{v_j\mid \vec{d}_{S_j}%\restriction\Delta 
     = (1,0,1,0)^T\}$, and $\phi(v_m) = 0$ for all $v_m\in V\setminus \{v_j\mid \vec{d}_{S_j}%\restriction\Delta
     = (1,0,1,0)^T\}$. Exactly one in three of the products in the right hand side of (\ref{r1}) decomposes to $(1,0,1,0)$, so exactly one in three of the variables (literals) in each clause of $C$ evaluates to $1$. It follows that $\lag V,C\rag\in\MonotoneOneThree$.\qed
     }%End of Detailed: Monotone 1-3 \leq SSFA.

\begin{proof}
    Given an instance $\lag V,C\rag$, we design a chemical reaction network $\calC$ where %such that:%whose 
    \begin{enumerate}
        \item Each molecular species consists of $2$ atoms $T$ and $F$ (representing ``True'' and ``False'' respectively), and
        \item reactions guarantees the equivalence: %following: %the network to be 
    $\calC$ is subset-$\Delta$-atomic if and only if the %original input instance is in 
    $\lag V,C\rag\in \MonotoneOneThree$. 
    \end{enumerate}
   % where %while atoms in $\Delta$ represent ``True'' and ``False''
    %, and whose  
    \opt{sub,subn}{
    
    For details, see Section \ref{AppAA} in the Appendix.}%Paragraph \ref{statementOnlyMonotoneThreeOnetoSSFA}.}
    \opt{normal}{\DetailedProofMonotoneOneThreePolyRedSSFA}
\end{proof}

\opt{normal}{
    We notice that the coefficients of all species in all the reactions $(\ref{r1})\sim(\ref{r3})$ are constants, so the numerical parameters -- entries of each $(\vec{r},\vec{p})\in R$ -- of the instances $\lag \Lambda, R, \Delta\rag$ constructed above are bounded by the constant $3$, which is again bounded by a polynomial of the length of $\lag V,C\rag$, presuming the encoding scheme is ``reasonable and concise''~\cite{garey1978strong}.
    We therefore conclude that:
}

\opt{sub,subn}{
    The full proof of Proposition~\ref{PartitionSFA} uses only coefficients of size $O(1)$ with respect to $|\lag V,C\rag|$, %the problem size instance, 
    which combined with Corollary \ref{cor:SFANPNP} establishes the following: %corollary:
}

\begin{cor}\label{strong}
$\SubsetFixedAtomic$ is strongly $\NP$-hard (and hence strongly $\NP$-complete).
\end{cor}

%Combining Corollary \ref{strong} and Corollary \ref{cor:SFANPNP}, we immediately obtain:
%\begin{prop}\label{SFANPC}
%$\SubsetFixedAtomic$ is strongly $\NP$-complete. \qed
%\end{prop}

\newcommand{\LongCommentSSFARemainsNPCWhenBiMolecular}{%Only appears in appendix and full.
In fact, one may convert, for each $m\in[1,k]$ (recall that $k$ is the number of clauses in $C$), any reaction in the series (\ref{r1}) ($3P+2F +T 
           \rightarrow S_{m1}+S_{m2}+S_{m3}\ (\forall m\in[1,k])$) into the following series:

\begin{eqnarray} \nonumber T+F&\leftrightharpoons& M_{m1}\\
\nonumber M_{m1}+F&\leftrightharpoons& M_{m2}\\
\nonumber M_{m2}+P&\leftrightharpoons& M_{m3}\\
\nonumber M_{m3}+P&\leftrightharpoons& M_{m4}\\
\nonumber M_{m4}+P&\rightarrow& M_{m5} + S_{m1}\\
\nonumber M_{m5}&\rightarrow& S_{m2}+S_{m3}
\end{eqnarray}

And apply similar methods to the $X_i$ species. Such conversion creates $2\times 10k=20k$ extra reactions and $2\times 5k= 10k$ intermediate species, which is polynomial in both time and space. 
}%End of Long comment: bimolecular.

\begin{rem}\label{SSFARemainNPCWhenBiMolecular}
%Another interesting thing to note is that
$\SubsetFixedAtomic$ remains $\NP$-complete even restricted to instances where $R$ contains only unimolecular and bimolecular reactions. 
    \opt{sub,subn}{To see %this, one only needs to translate the network $\calC$ designed in Proposition \ref{PartitionSFA} into an equivalent $\calC'$ whose reactions are bimolecular or unimolecular. For 
    details on this, see Section \ref{AppAA} in the Appendix.}%Paragraph \ref{sSFARemainNPCWhenBiMolecular}.}

    \opt{normal}{\LongCommentSSFARemainsNPCWhenBiMolecular}
\end{rem}

%However, %the case of $\SubsetAtomic$ is more complicated. Note that in many cases, for some $\lag\Lambda, R\rag\in\SubsetAtomic$, there exist multiple possible $\Delta$'s with respect to whom $(\Lambda, R)$ is subset atomic, so t
%There is no obvious reduction from $\SubsetFixedAtomic$ to $\SubsetAtomic$, and %It remains open as to what complexity class $\SubsetAtomic$ belongs to
The lower bound of the complexity of $\SubsetAtomic$ therefore remains open, but we conjecture that \SubsetAtomic\ is $\NP$-\text{hard} (hence $\NP$-$\Complete$).
%below is our conjecture:
%we have the following conjecture:

%\begin{conj}\label{SANPC}
%\SubsetAtomic\ is %$\NP$-$\Complete$.
%\end{conj}

\section{Complexity of reachably atomic}\label{52cva}

Without %regurgitating
repeating the intuition of the definition of reachably %atomicity
atomic which has been explained in Subsection \ref{previous111} and Section \ref{newdefns}, we proceed with the corresponding definition of languages for deciding reachable atomicity and the reachability problem in reachably atomic networks. %[SENTENCE NEEDS MODIFICATION]

\begin{defn}\label{VisAtom}
 We define the following languages:
 \begin{eqnarray*}
 \ReachablyAtomic &=& \{\lag\Lambda,R\rag\mid (\exists\Delta\subseteq \Lambda)((\Lambda, R)\textrm{ is reachably atomic with respect to }\Delta)\}
 \\
 \ReachablyFixedAtomic &=& \{\lag\Lambda,R,\Delta\rag\mid (\Lambda, R)\textrm{ is reachably atomic with respect to }\Delta\}
  \end{eqnarray*}
 \end{defn}

 Distinction between $\ReachablyFixedAtomic$ and $\ReachablyAtomic$ is analogous to %that between  
 ``$\SubsetFixedAtomic$ v.s. $\SubsetAtomic$''. However, %given the uniqueness of choice of $\Delta$ as well as the composition matrix $\matr{D}$ (
 by Lemma \ref{uniquereachable} there is no semantic 
 reason to distinguish between ``$\ReachablyFixedAtomic$'' and $\ReachablyAtomic$. Hence, %for subsequent discussion 
 we shall only consider $\ReachablyAtomic$.
 
\subsection{$\ReachablyAtomic$ is in P}

As mentioned before, the requirement that $\{1S\}\Rightarrow^* \vec{d}_S\ (\forall S\in\Lambda)$ %seems strong
ensures some interesting results. The complexity results in this subsection confirm this.

\newcommand{\DetailedProofLemmaShowingExistenceOfOneStepDecomposibleElement}{Suppose not, then for all reactions with $\vec{r}=\{1S\}$ for some $S\in\Lambda\setminus\Delta$, either $[\vec{p}]\cap (\Lambda\setminus\Delta)\neq \emptyset$, or $[\vec{p}]\subseteq\Delta$ but $\vec{p}\neq\vec{d}_S$. The last case cannot happen, due to the uniqueness of atomic decomposition for reachably atomic networks (Recall Lemma \ref{uniquereachable}). Hence for all $(\vec{r},\vec{p})$ with $\vec{r}=\{1S\}$ for some $S\in\Lambda\setminus\Delta$, $[\vec{p}]\cap (\Lambda\setminus\Delta)\neq \emptyset$ [*]. 
 
 %But since each molecular species eventually splits (albeit possibly after some isomerization reactions, which are reactions shaped like $S_1\rightarrow S_2$, $S_1,S_2\in\Lambda\setminus\Delta$), 
 [*], together with the reachable-atomicity, implies that for each $S\in\Lambda\setminus\Delta$ one may find a $S'$ s.t. $\|\vec{d}_S\|_1>\|\vec{d}_{S'}\|_1$ $[**]$. To see this, consider an arbitrary $S_i\in\Lambda\setminus \Delta$: any $(\vec{r},\vec{p})$ with $\vec{r}=\{1S_i\}$ has either $\|\vec{p}\|_1=1$, or $\|\vec{p}\|_1\geq 2$. In the second case we are done, for any $S_j\in[\vec{p}]\cap(\Lambda\setminus\Delta)$ satisfies $\|\vec{d}_{S_j}\|_1<\|\vec{d}_{S_i}\|_1$; in the first case, we have found some $S_{i+1}$ s.t. $\vec{d}_{S_{i+1}}=\vec{d}_{S_i}$ (and we call such $(\vec{r},\vec{p})$ an \emph{isomerization reaction}), so we recursively inspect into all $(\vec{r}_1,\vec{p}_1)$ with $\vec{r}_1=\{1S_{i+1}\}$. Such a recursion must finally terminate with some $S_{i+m}$
that satisfies $(\exists(\vec{r}_m,\vec{p}_m)\mid \vec{r}_m=\{1S_{i+m}\})\|\vec{p}_m\|_1\geq 2$, for otherwise $\vec{d}_{S_i}$ would not be reachably decomposible into $\vec{d}_{S_i}$ via any reaction sequence. It follows that any $S_{i+m+1}\in [\vec{p}_m]\cap (\Lambda\cap\Delta)$ satisfies $\|\vec{d}_{S_i}\|_1 > \|\vec{d}_{S_{i+m+1}}\|_1$.
 
 We have argued that our initial assumption (for the sake of contradiction) implies $[**]$. But $[**]$ would imply that there exists no molecular species with minimal size, contradicting the finiteness of $\Lambda$.
}%detail: lemma of 1-step decomposible.

\newcommand{\StatementLemmaShowingExistenceOfOneStepDecomposibleElement}{%Only in Appendix/full. 
    If a network $\calC=(\Lambda, R)$ is reachably atomic with respect to $\Delta$ via decompositin matrix $\matr{D}$ (or equivalently, via the set of decomposition vectors $\{\vec{d}_S\}_{S\in\Lambda}$), then $\exists S\in \Lambda\setminus\Delta$ and $(\vec{r},\vec{p})\in R$ s.t. $\vec{r}=\{1S\}$ and $\vec{p}=\vec{d}_{S}$.
}%Lemma: exist 1-step decomposible species.

\begin{lem}\label{viswithdirect}
\StatementLemmaShowingExistenceOfOneStepDecomposibleElement
\end{lem}

\begin{proof}
    The claim is saying that if a network is reachably atomic, then there exists a molecular species that can be decomposed into its atomic decomposition in \emph{one} single reaction. Proof is done by assuming otherwise and chasing the decomposition sequence to find an infinite descending chain of species ordered by the size of their decomposition vectors, contradicting the finiteness of species set. 
    \opt{sub,subn}{For details, see Section \ref{AppAA} in the Appendix.}%Paragraph \ref{statementLemmaShowingExistenceOfOneStepDecomposibleElement}.}
    \opt{normal}{\DetailedProofLemmaShowingExistenceOfOneStepDecomposibleElement}
\end{proof}

\newcommand{\ReachableAtomicIsInP}{$\ReachablyAtomic$ $\in$ $\mathsf{P}$.}%Statement of thm appearing in Appendix.

\begin{theorem}\label{VisAtomP}
\ReachableAtomicIsInP
\end{theorem}

\newcommand{\ProofReachableAtomicIsInP}{describe this polynomial time algorithm in more details, argue for its correctness, and then exhibit the pseudo-code.

 Our algorithm will compose of the following steps:
 
 \begin{enumerate}
 
     \item Enumerate all reactions in $R$: for each %disassociation 
     reaction starting with $\{1S\}$ for some $S\in\Lambda$, put $S$ in the set $M$ of molecules. After the enumeration, define $\Delta = \Lambda\setminus M$. If $\Delta = \emptyset$ %(no atoms) 
     or $M=\emptyset$, reject. %(every species is an atom, in which case the only way for the system to preserve atoms)
 
     \item Find if there exists some molecular species $S\in M$ such that $S$ is decomposible into some $\vec{d}_S$ with $[\vec{d}_S]\subseteq\Delta$ by execution of a single reaction. If any of such $\vec{d}_S$ is of size $0$ or $1$, reject. Keep track of each decomposition vector $\vec{d}_S$;
     
     Make a subset $M'$ of $M$ s.t. $M'$ contains all molecular species which CANNOT be decomposed into $\vec{d}_S$ with $[\vec{d}_S]\subseteq \Delta$ by execution of one reaction.  If $M'=M$, then reject.
     
     \item while $M'$ is not empty, iterate and try to find an $S\in M'$ that satisfies this condition: $\exists (\vec{r},\vec{p})\in R$ s.t. $\vec{r}=\{1S\}$ and $[\vec{p}]\subseteq (M\setminus M')\cup \Delta$. %This says that $S$ can be decomposed into complexes consisting only of species that are known to be decomposible (that is, 
     Note that all elements $S''$ in $M\setminus M'$ satsifies $\{1S''\}\Rightarrow^* \vec{d}_{S''}$, hence if $S$ can be decomposed into complexes consisting solely of elements in $(M\setminus M')\cup\Delta$ via execution of one reaction, then $S$ itself satisfies $\{1S\}\Rightarrow^* \vec{d}_S$ as well. Keep track of $\vec{d}_S$ and exclude such $S$ from $M'$.
     
     If in some iteration we cannot find any such $S\in M'$, then %\textsc{reject}
     reject; else, the iteration will finally halt excluding all such $S$'s and making $M'$ empty.
     
     \item We have obtained $\vec{d}_S$ for each $S\in\Lambda$ (note that $\vec{d}_A=\vec{e}_A$ for all $A\in\Delta$) s.t. $[\vec{d}_{S}]\subseteq \Delta$ and $\{1S\}\Rightarrow^* \vec{d}_S$. By Lemma \ref{uniquereachable}, such set of decomposition is unique. Now, check if this decomposition conserves atoms by composing stoichiometric matrix $\matr{A}$ and decomposition matrix $\matr{D}$ and examine if $\matr{A}\cdot\matr{D}=\matr{0}$, and finally check if each atom $A$ appears at least once in some $\vec{d}_S$ for some $S\in M$.
 \end{enumerate}

 We first prove that if reachably atomic then the algorithm will halt in \textsc{accept}: 
 
 if $\calC$ is reachably atomic, then we claim that the set $\{S\in \Lambda\mid \exists (\vec{r},\vec{p})\in R\textrm{ s.t. }\{1S\}=\vec{r}\}$ is exactly the set of molecules $M$, with its complement $\Delta=\Lambda\setminus M$. To see this, recall that we prescribed there being no ``$\vec{r}\rightarrow\vec{r}$'' reactions in $R$, so all reactions $(\vec{r},\vec{p})\in R\mid \|\vec{r}\|=1$ is either an isomerization reaction ($\|\vec{p}\|_1=1,\vec{p}\neq\vec{r}$) or disassociation reaction ($\|\vec{p}\|_1\geq 2$). Both types of reactions can only happen when $S\in [\vec{r}]$ is a molecule; it follows that $\{S\in \Lambda\mid \exists (\vec{r},\vec{p})\in R\textrm{ s.t. }\{1S\}=\vec{r}\}\subseteq M$. Conversely, when $S\in M$, reachably atomicity gives $S\in\{S\in\Lambda\mid \exists (\vec{r},\vec{p})\in R\ \textrm{ s.t. }\{1S\}=\vec{r}\}$.\footnote{ We point out that the set of atoms $M\neq \{S\in\Lambda\mid \exists (\vec{r},\vec{p})\in R\ \textrm{ s.t. }(\{1S\}=\vec{r})\wedge(\|\vec{p}\|\geq 2)\}$, so we have to test the $\|\vec{d}_S\|\geq 2$ condition in later steps. This is because it might be the case that the only reaction $(\vec{r},\vec{p})$ with $\vec{r}=\{1S\}$  turns out to be an isomerization reaction. A counter example would be: \begin{eqnarray*}A&\rightarrow&B\\B&\rightarrow& 2C\end{eqnarray*}
 
 By our definition $M = \{S\in\Lambda\mid \exists (\vec{r},\vec{p})\in R\ \textrm{ s.t. }\{1S\}=\vec{r}\}$, we shall correctly identify $M=\{A,B\}$, yet the added condition $\|\vec{p}\| \geq 2$ would make $M=\{B\}$, a mis-identification.}

 Neither $M$ nor $\Delta$ would be empty, for $(\Delta=\emptyset)\Rightarrow (\calC$ is not reachably atomic) and $(M=\emptyset)\Rightarrow (R=\emptyset)$. Hence the algorithm passes Step $1$, correctly identifying the partition $(M,\Delta)$ of $\Lambda$.
 
 By Lemma \ref{viswithdirect}, reachable-atomicity implies that the algorithm will find at least one molecular species $S$ that directly decomposes to its atomic decomposition $\vec{d}_S$ and grantedly, $\|\vec{d}_S\|\geq 2$, so Step $(2)$ will be passed.
 
 Further, applying the same argument in Lemma \ref{viswithdirect} on the set $M'$, the while loop must shrink the cardinality of $M'$ by at least $1$ per iteration, and finally exit by making $M'$ empty, \footnote{That is, if ($\forall S\in M'$)($\forall(\vec{r},\vec{p})\in R$) ($\vec{r}=\{1S\}\Rightarrow [\vec{p}]\cap M'\neq\emptyset$), then for each species $S$ in $M'$ there will be $S'\in M'$ s.t. $\|S\|_1>\|S'\|_1$, contradicting the finiteness of $M'$.} passing Step $(3)$;
 
 Finally, the decomposition must preserve atoms for all reactions, and all atoms must appear in the decomposition of at least one molecule, which make both tests in %the algorithm 
 Step $(4)$ passed. 
 
 It remains to show that if $\calC$ is not reachably atomic then our algorithm will halt in \textsc{reject}.  
 We claim that: if $\calC$ is not reachably atomic, then exactly one of the following will be true:

 \begin{enumerate}
 
     \item There is no valid separation of $\Lambda$ into $M$ and $\Delta$. That is, either $\{S\in \Lambda\mid \exists (\vec{r},\vec{p})\in R\textrm{ s.t. }\{1S\}=\vec{r}\}=\emptyset$ (no species is the single reactant of an isomerization or disassociation reaction, so no species $S$ is decomposible from $\{1S\}$), or $\{S\in \Lambda\mid \exists (\vec{r},\vec{p})\in R\textrm{ s.t. }\{1S\}=\vec{r}\}=\Lambda$ (every species is the single reactant of some isomerization or disassociation reaction, which contradicts the definition of subset atomicity for atoms should be neither isomerizable nor decomposible).%and hence reachably atomic). 
      An example where $\{S\in\Lambda\mid \exists (\vec{r},\vec{p})\in R\ \textrm{ s.t. }\{1S\}=\vec{r}\}=\emptyset$ would be $(\Lambda = \{A,B,C\}, R=\{2A+3B\rightarrow 4C\})$, while $(\Lambda' = \{A,B\}, R'=\{A\rightarrow B, B\rightarrow A\})$ would be an instance where $\{S\in\Lambda\mid \exists (\vec{r},\vec{p})\in R\ \textrm{ s.t. }\{1S\}=\vec{r}\}=\Lambda$.
     
     Observe such a valid separation $(M,\Delta)$ of $\Lambda$, if existing, is unique for a certain $\calC=(\Lambda,R)$, since $S\in M$ if and only if $\exists (\vec{r},\vec{p})\in R$ s.t. $\vec{r}=\{1S\}$, and this property is uniquely decided by $R$.
 
     \item There exists the unique valid separation $(M,\Delta)$ of $\Lambda$, but there exists no molecular species directly decomposible into its atomic decomposition via execution of one single reaction. That is, $(\forall S\in M)(\forall (\vec{r},\vec{p})\in R)(\vec{r}=\{1S\}\Rightarrow [\vec{p}]\cap M\neq\emptyset)$. 
An example of this is $(\Lambda = \{A,B,C\}, R=\{A\rightarrow B+5C, B\rightarrow A+5C\})$. We would successfully identify $M=\{A,B\}, \Delta = \{C\}$, but for all reactions we $(\vec{r},\vec{p})\in R$ have  $[\vec{p}]\cap M\neq \emptyset$.

\item There exists the unique valid separation $(M,\Delta)$ of $\Lambda$ and $\{S\in M\mid(\exists (\vec{r}_S,\vec{p}_S)\in R)((\vec{r}_S=\{1S\})\wedge ([\vec{p}_S]\subseteq \Delta))\}\neq\emptyset$, but for some $S\in \{S\in M\mid(\exists (\vec{r}_S,\vec{p}_S)\in R)((\vec{r}=\{1S\})\wedge ([\vec{p}_S]\subseteq \Delta))\}$, $\|\vec{p}\|\leq 1$. That is, we have some reaction $S_1\rightarrow A_1$ with $S_1\in M$ and $A_1\in\Delta$, or $S_1\rightarrow \emptyset$. In this case, either a molecule decomposes to a single atom, or it vanishes. 

Typical Examples are: $\calC_1 = (\{A,B,C\},\{A\rightarrow B, B\rightarrow C\})$, $\calC_2 = (\{A,B,C\},\{A\rightarrow 2C, B\rightarrow \emptyset\})$. In both cases one would identify $M_{\calC_1}=M_{\calC_2} = \{A,B\}$; for both networks, $\{S\in M\mid(\exists (\vec{r}_S,\vec{p}_S)\in R)((\vec{r}=\{1S\})\wedge ([\vec{p}_S]\subseteq \Delta))\}=\{A,B\}$. But In $\calC_1$, $B$ decomposes to a single atom $C$; in $\calC_2$, $B$ vanishes. This disqualifies both sets from being reachably atomic by placing them in the third case.

     \item There exists the unique valid separation $(M,\Delta)$ of $\Lambda$, and $\{S\in M\mid(\exists (\vec{r}_S,\vec{p}_S)\in R)((\vec{r}_S=\{1S\})\wedge ([\vec{p}_S]\subseteq \Delta))\}\neq\emptyset$; further, each $S\in \{S\in M\mid(\exists (\vec{r}_S,\vec{p}_S)\in R)((\vec{r}_S=\{1S\})\wedge ([\vec{p}_S]\subseteq \Delta))\}$ satisfies $\|\vec{p}\|_1\geq 2$. However, there exists some indecomposible molecular species.  %which cannot be decomposed into a state with support in $\Delta$. 
      %So this case means for the (unique) separation $(M,\Delta)$,
      That is, $\exists$ $S'\in M$ s.t. $\forall \vec{c}\in\N^\Lambda$ where $\{1S'\}\reach\vec{c} $, $[\vec{c}]\cap M\neq \emptyset$. 
      
      An example for this case is $\calC = (\{A,B,C,D,E\},\{A\rightarrow B,  B\rightarrow C, D\rightarrow 3E\})$. One may identify $M = \{A,B,D\}$ and find $\{D\}= \{S\in M\mid (\exists (\vec{r},\vec{p})\in R)((\vec{r} = \{1S\})\wedge([\vec{p}]\subseteq \Delta))\}$. Further, the reaction $D\rightarrow 3E$ where $\vec{r}=\{1D\}$ satisfie $\|\vec{p}\|=3$. This network does not belong to any of the first few cases but it does belong to Case $4$, for $\forall \vec{c}$ where $\{1A\}\Rightarrow^* \vec{c}$, $[\vec{c}]\subseteq \{B,C\}\subseteq M$. %This standard is only determined by $R$
     
     \item There is a unique valid decomposition $(M,\Delta)$ of $\Lambda$ and $(\forall S\in M)$ $(\exists \vec{c}_S$ with $[\vec{c}_S]\subseteq\Delta)$ $(\{1S\}\reach\vec{c}_S)\wedge(\|\vec{c}_S\|_1\geq 2)$, but the decomposition does not preserve atoms for some reaction. That is, with $\matr{A}$ the stoichiometric matrix and $\matr{D}$ the decomposition matrix (row vectors being the $\vec{c}_S$'s restricted to $\Delta$), $\matr{A}\cdot\matr{D}\neq \matr{0}$. Note that for reachably atomic networks, atomic decomposition vectors (or equivalently, matrix) should be unique. %Observe, again, that 
    
    One example of this would be $(\{A,B,C,D\},\{ A\rightarrow B, B\rightarrow 3C, A+B\rightarrow 5C+D\})$. Here we have $M=\{A,B\}$ and $\{3C\}\underbrace{=}_{B\rightarrow 3C}\vec{d}_B\underbrace{=}_{A\rightarrow B}\vec{d}_A\underbrace{=}_{A+B\Rightarrow 5C+D, B\rightarrow 3C}\{2C+D\}$, contradicting the preservation of atoms. Note that this happens to be another example where a network is mass conserving (Just set $\vec{m}(A)=\vec{m}(B)=3\vec{m}(C)=3\vec{m}(D)= 3$) but not subset atomic (and hence not reachably atomic). \item There is a unique valid decomposition $(M,\Delta)$ of $\Lambda$ and $(\forall S\in M)$ $(\exists \vec{c}_S$ with $[\vec{c}_S]\subseteq\Delta)$ $(\{1S\}\reach\vec{c}_S)\wedge(\|\vec{c}_S\|_1\geq 2)$, and the decomposition preserves atoms ($\matr{A}\cdot\matr{D}=0)$, but some atoms are ''redundant'': $\exists A\in\Delta$ s.t. $\forall S\in M$, $A\notin[\vec{d}_S]=[\vec{c}_S]$.
    
    One may ``fix'' the last example into this case: $(\{A,B,C,D\},\{ A\rightarrow B, B\rightarrow 3C, A+B\rightarrow 6C\})$. Here we have $M=\{A,B\}$ and $\{3C\}\underbrace{=}_{B\rightarrow 3C}\vec{d}_B\underbrace{=}_{A\rightarrow B}\vec{d}_A\underbrace{=}_{A+B\Rightarrow 6C, B\rightarrow 3C}\{3C\}$, so every condition for reachably atomic is satisfied, except that $(\forall S)D\not\in[\vec{d}_S]$.%, contradicting the preservation of atoms.
 \end{enumerate}%the algoirhthm will return no because 

All six cases exclude each other, so \emph{at most} one case could hold; on the other hand, the negation of the disjunction of all six cases says that there exists a non-empty proper subset of $\Lambda$ and a decomposition matrix $\matr{D}$, such that all three conditions of primitive atomicity holds with respect to $\Lambda$ via $\matr{D}$, and $\{1S\}\Rightarrow ^*\vec{d}_S$ for all $S$. This is exactly the definition of reachably atomicity. So taking contraposition, non-reachable-atomicity implies \emph{at least} one of the six cases hold.

Instances satisfying Case $(1)$ will be rejected in Step $(1)$, while Cases $(2)$ and $(3)$ will get rejected in Step $(2)$. In case $(4)$, the loop for finding decomposition vectors must terminate before $M'$ is emptied, so it will get rejected by Step $(3)$; Cases $(5)$ and $(6)$ triggers rejection in Step $(4)$.

We exhibit the following pseudocode for the decider:

 \begin{algorithm}[H]
 
\DontPrintSemicolon
\emph{Initialize global set $M,M',\Delta, D=\emptyset$\ //$D$: $\{$decomposition vectors$\}$.}\;
\For{$(\vec{r},\vec{p})\in R$}{
    \If{$(\exists S\in \Lambda)\vec{r}=\{1S\}$}{
        $M\leftarrow M\cup \{S\}$
        }
    }
$\Delta\leftarrow\Lambda\setminus M$\;

$M'\leftarrow M$\;

\If{$M=\emptyset$ or $\Delta = \emptyset$}{
    $\textsc{Reject}$
}

$D\leftarrow D\cup \{\vec{e}_{A}\}_{A\in\Delta}$\;

\For{$(\vec{r},\vec{p})\in R$ where $(\exists S\in M)\vec{r}=\{1S\}$}{
    \If{$[\vec{p}]\subseteq \Delta$}{
        \If{$\|\vec{p}\|_1\leq 1$}{\textsc{reject}
        }
    }
    $D\leftarrow D\cup \{\lag\vec{d}_{S} = \vec{p}\rag\}$\;
    $M'\leftarrow M'\setminus\{S\}$\;
}

\If{$M'=M$}{
    \textsc{reject}
}
\While{$M'\neq\emptyset$}{
    \If{$(\forall S\in M')$  $(\forall (\vec{r},\vec{p})\in R\mid \vec{r}=\{1S\})$  $([\vec{p}]\cap M'\neq\emptyset)$
    %%%\not\subseteq \Delta\cup(M\setminus M'))
}{
        \textsc{reject}
    }
    \Else{\
        \For{($S\in M'\mid (\exists(\vec{r},\vec{p})\in R\mid \vec{r}=\{1S\}\textrm{ and }[\vec{p}]\cap M'=\emptyset))$}{
            $D\leftarrow D\cup \{\lag\vec{d}_{S} = \sum_{S'\in[\vec{p}]}\vec{d}_{S'}\rag\}$\;
        
            $M'\leftarrow M'\setminus\{S\}$ \;
        }
    }      
}
   \emph{Compose $\matr{A}$ (stochiometric matrix) and $\matr{D}$ (decomposition matrix)}\;    
   
   \If{$\matr{A}\cdot \matr{D}\neq \matr{0}$}{
        \textsc{reject}
    }
   \If{$(\exists A\in\Delta)(\forall S\in M)A\not\in[\vec{d}_S]$}{
        \textsc{reject}
    }
   \textsc{accept}\;

\caption{Decider for $\ReachablyAtomic$}\label{VisAtomPolyTime}
\end{algorithm}

 Let us briefly mention that this is a polynomial time algorithm. The first $\texttt{for}$-loop takes $O(|R|)$ time; the second $\texttt{for}$-loop takes at most $O(|R|)$ iterations, and each iteration takes $O(|\Lambda|^3)$ time; as for the $\texttt{while}$ loop, note that it either shrinks the size of $M'$ by $1$ per iteration, or $\textsc{rejects}$. Hence the \texttt{while} loop takes at most $O(|\Lambda|)$ to exit. The $\texttt{if}$-statement inside the $\texttt{while}$-loop takes $O(|\Lambda|\cdot |R|\cdot |\Lambda|^2)$ to evaluate. Lastly, composing and multiplying $\matr{A}\cdot \matr{D}$ takes $O(|R||\Lambda|\cdot |\Lambda|^2)$ time, and verifying each $A\in\Delta$ is ``used'' by the decomposition of some molecule is $O(|\Lambda|\cdot|\Lambda|\cdot|\Lambda|)$. The times complexity is therefore dominated by the $\texttt{while}$, which is $O(|R||\Lambda|^4)$. No input, output or intermediate encoding takes more than polynomial space to record, so $\ReachablyAtomic\in \textsf{P}$, as desired.
 }%End of Proof that Reachably Atomic is in P

\begin{proof}
    %Intuitively, we would 
    We need to exhibit a polynomial time algorithm %to identify if 
    that decides whether there exists a separation of $\Lambda$ into two non-empty, disjoint sets $M$ (molecules) and $\Delta$ (atoms), with elements in $M$ decomposable  via sequences of reactions into combination of elements in $\Delta$. 
    
    To achieve this goal, we set $M=\{S\in\Lambda\mid (\exists(\vec{r},\vec{p})\in R)\vec{r}=\{1S\}\}$, %which would later be %justified as 
    %proved to be the very set of molecular species
    the subset of species which are the single reactant of some reaction; %, the rationale of which will be explained later
    apparently $M$ is non-empty for reachably-atomic networks, by Lemma
    %upon successfully obtaining such a separation (should it exist), we first check if there exists a molecular species decomposable into its atomic composition\emph{ via one single reaction}, which should exist if the network is reachably atomic, by Lemma 
    \ref{viswithdirect}. Then recursively, we check if there exist elements in $M$ that can be decomposed into combination of atoms %and previously identified decomposable molecules \textit{via one single reaction}.
    via a reaction sequence of length $i=1,2,\cdots$, and reject if  we succeed to do so at $i=k$ but fails at $i=k+1$ while not all elements in $M$ have been examined. %[EXPLAIN UNINTUITIVE OPERATIONS AND DEFINITIONS.] %TO BE CONTINUED TOMORROW.] %If 
    When this process terminates (note that $M$ is finite) %[EXPLAIN INTUITION WHY IT MUST TERMINATE] 
    finding (candidate) atomic decomposition for all molecules, we %proceed to check if the compositions preserve molecules and if each atom appears in the composition of at least one molecule
    verify if the necessary conservation laws hold. %We only accept when the network passes all those tests.  
    \opt{sub,subn}{Details of the proof are included in Section \ref{AppAA} of the Appendix.}%Paragraph \ref{reachableAtomicIsInP}.}
    \opt{normal}{\ProofReachableAtomicIsInP}
\end{proof}

\subsection{$\ReachableReach$ is $\Pspace$-$\Complete$}

We shall first introduce %begin this subsection by introducing %(without proof)
%a relevant lemma, followed by the proof of the finiteness of the 
the definition of configuration reachability graphs, followed by a result proved in %relevant literature
\cite{mayr2014framework} (see also Subsection \ref{previous111}), %the idea of 
based on which %shall be used to 
we prove %the $\Pspace$-$\Complete$ness of 
$\ReachableReach$ (see Definition \ref{RRReach}), a problem motivated by restricting relevant problems such as ``exact reachability'' \cite{leroux2011vector}, is $\Pspace$-$\Complete$.

\begin{defn}[Configuration Reachability Graph] An \emph{$\vec{i}$-initiated Configuration Reachability Graph} $G_{\calC,\vec{i}}$ of the chemical reaction network $\calC=(\Lambda, R)$ is a directed graph $(V,E)$, where:\begin{enumerate}
    \item each $v_{\vec{c}}\in V$ \emph{(}$\vec{c}\in\N^\Lambda$\emph{)} is labeled by a reachable configuration %/configuration 
    $\vec{c}$%\in\N^{\Lambda}$
    \ of $\calC$ ;
    \item $v_{\vec{i}}\in V$\emph{(}$\vec{i}\in\N^\Lambda$\emph{)} is the vertex labeled by the initializing configuration $\vec{i}$;
    \item the ordered pair $(v_{\vec{c}_1},v_{\vec{c}_2})\in E$ if and only if $\vec{c}_1\Rightarrow^{1}\vec{c}_2$.
\end{enumerate}  
\end{defn}

\begin{rem}\normalfont{}
For the sake of simplicity, we use $G_{\mathcal{C},\vec{i}}$ as shorthand for $G_{\mathcal{C},v_{\vec{i}}}$.
\end{rem}

\newcommand{\EgConfigGraph}{%Only in full and appendix. Note that it uses the attached graph!
\begin{example}\label{egConfigGraph}\normalfont{}
Consider $C=(\Lambda, R)$ where $\Lambda = \{S_1,S_2,S_3,A_1\}$ $($we use the order exhibited above for %the order in 
$\N^{\Lambda})$, and $R$, in its explicit form, consists of
\begin{eqnarray}
 S_1&\rightarrow & 4A_1\label{dec1}\\
S_2&\rightarrow& 9A_2\label{dec2}\\
2S_1 + S_2&\rightarrow& S_3\label{bigger}
\end{eqnarray}

\quad Now consider two intialization vectors, $\vec{i} = (0,2,1,0)^T$ $($that is, $\{2S_2, 1S_3\})$ and $\vec{i}^{'}=(2,1,3,0)^T($i.e., $\{2S_1,1S_2,3S_3\})$. For $\vec{i}$, the only possible reaction is $(\ref{dec2})$ and the only possible reaction path is two consecutive executions of $(\ref{dec2})$. Hence $G_{C,\vec{i}}=(V,E)$ where $V=\{v_{\vec{i}},v_{(0,1,1,9)^T},v_{(0,0,1,18)^T}\})$ and $E=\{(v_{\vec{i}},v_{(0,1,1,9)^T}),(v_{(0,1,1,9)^T}, v_{(0,0,1,18)^T}\}$.

\quad The case for $\vec{i}^{'}$ is more complicated. Potential execution paths include: $(\ref{bigger})$; $(\ref{dec1})\rightarrow(\ref{dec1})\rightarrow(\ref{dec2})$; $(\ref{dec1})\rightarrow(\ref{dec2})\rightarrow(\ref{dec1})$; $(\ref{dec2})\rightarrow(\ref{dec1})\rightarrow(\ref{dec1})$. We may construct the Configuration Reachability Graph $G_{C,\vec{i}^{'}}$ following these paths. Figures of $G_{C,\vec{i}}$ and $G_{C,\vec{i}^{'}}$ are shown below: 
\begin{figure}[H]\label{crg2}
\includegraphics[height=0.3%93024
\textheight,width=%0.989184
0.75\textwidth]{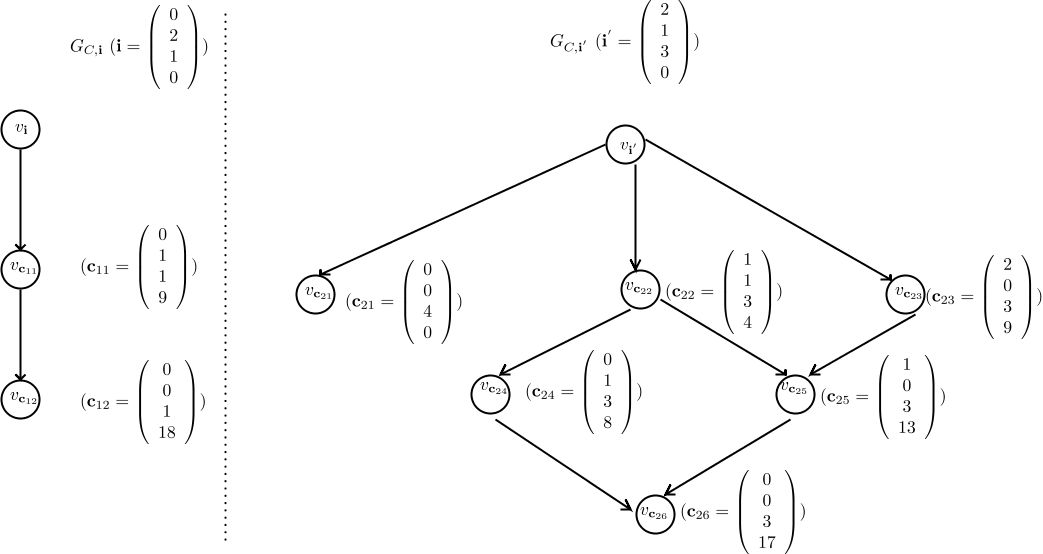}\caption{Configuration Reachability Graphs $G_{C,\mathbf{i}},G_{C,\mathbf{i}^{'}}$}
\end{figure}
\end{example}
}%End of E.g. of config graphs.

 For the same $\calC$, Configuration Reachability Graphs can be far from isomorphic due to parameterization by different initialization vectors.
  \opt{sub,subn}{We have included an example (Example \ref{egConfigGraph}) in Section \ref{Auxiliary}.}

    \opt{normal}{\EgConfigGraph}
    
%We are about ready to 
We will soon prove the conclusion on the complexity of the reachability problem for reachably atomic networks. But first, we point out that the following is %equivalent to 
a %direct application 
straightforward translation of a finding in \cite{mayr2014framework}, giving the complexity class of reachability problems for mass-conserving chemical reaction networks.

\begin{obs}[A result proved in \cite{mayr2014framework}]\label{MCComplete}
    For all mass conserving chemical reaction networks $\calC$ and initial configuration $\mathbf{i}$ of $\calC$,  $|\lag G_{\calC,\mathbf{i}}\rag|\in O(2^{poly(|\lag\calC, \vec{i}\rag|)})$. That is, the binary size of the encoding of the configuration reachability graph $G_{\calC,\mathbf{i}}$ is at most exponential to the binary size of the encoding of the pair $(\calC, \vec{i})$. 
    
    Furthermore, reachability problem for mass conserving networks is $\Pspace$-$\Complete$. That is, it is $\Pspace$-$\Complete$ to decide if an instance is in the following language:
    \begin{equation*}
        \{\lag \Lambda, R, \vec{c}_1,\vec{c}_2\mid (\Lambda, R) \underline{\textrm{ is mass conserving}}; \ \vec{c}_1,\vec{c}_2\in \N^\Lambda; \vec{c}_1\Rightarrow^*\vec{c}_2\rag\}
    \end{equation*}
\end{obs}

Built on Observation \ref{MCComplete}, we now exhibit the proof that the decision problem ``Given a Reachably Atomic network, is $\vec{c}_2$ reachable from $\vec{c}_1$'' is $\Pspace$-$\mathsf{Complete}$. %More formally,

\begin{defn}[$\ReachableReach$] \label{RRReach}
    We define the language 
    \begin{equation*}
        \textsc{Reachable-Reach}=\{(\Lambda, R, \vec{c}_1,\vec{c}_2)\mid (\Lambda,R) \underline{\textrm{ is reachably atomic}};\ \vec{c}_1,\vec{c}_2\in\N^{\Lambda};\vec{c}_1\reach\vec{c}_2\}
    \end{equation*}
\end{defn}

%Note the difference between the underlined parts in Observation \ref{MCComplete} and definition \ref{RRReach}.

%Referencing the proof ideas and conclusions in Observation \ref{MCComplete}, we shall show the complexity class of \ReachableReach\ below, using the exponential size bound on configuration reachability graph of reachably atomic networks and a simulation of polynomial space Turing Machine by reachably atomic network.

\newcommand{\StatementOfTheoremReachableReachIsPspaceComplete}{%Statement reused in Appendix. Kept.
$\ReachableReach$ is $\Pspace$-$\Complete$. 
}%reachable Reach pspace complete: statement. Used in both text and appendix.

\begin{prop}\label{ReachableReach}
\StatementOfTheoremReachableReachIsPspaceComplete
\end{prop}

\newcommand{\DetailedProofOfReRePspaceComplete}{
Let $\lag\Lambda, R,\vec{c}_1,\vec{c}_2\rag$ be an instance, and let $Z:= |\lag\Lambda, R,\vec{c}_1,\vec{c}_2 \rag|$. In Theorem \ref{VisAtomP} we proved that $\ReachablyAtomic\in P$ so we may run the polynomial decider on $\lag\Lambda,R\rag$ and $\textsc{reject}$ if $\lag\Lambda,R\rag\not\in\ReachablyAtomic$.

If the $\ReachablyAtomic$ decider halts in $\textsc{accept}$, we would obtain $\Delta\subseteq \Lambda$ with respect to which $(\Lambda,R)$ is reachably atomic, as well as the set $\{\vec{d}_S\}_{S\in\Lambda}$ of decomposition vectors. Further, we would have confirmed that $\mathcal{C}=(\Lambda,R)$ is mass-conserving, for this is implied by reachably atomicity. 
Recall from \cite{mayr2014framework} that the number of vertices in a configuration reachability graph $G_{\mathcal{C},\vec{c}_1}$ for mass-conserving network $\mathcal{C}$ is at most exponential to the binary size of the input%, and reachabily-atomicity implies the mass-conserving property
. Now, let $n=|V_{G_{\mathcal{C},\vec{c}_1}}|$, then by Savitch's Theorem\cite{papadimitriou2003computational}, $$\ReachableReach\in \textbf{\textrm{SPACE}}((\log n)^2)=\Space((\log((O(2^{\textrm{poly}(Z)})))^2) = \Space(O(\textrm{poly}(L))^2)$$

It follows that $\ReachableReach\in\Pspace$.

As for the $\Pspace$ hardness, we shall prove by simulating a polynomial-space Turing Machine. That is, consider the language $$L:=\{\lag M,x,0^{|x|^c}\rag\mid M\textrm{ is an }O(|x|^c)-\textrm{space, clocked Turing Machine, }x\in\{0,1\}^*: M(x)\rightarrow 1\}$$

Just to clarify the notation, ``$M(x)\rightarrow 1$'' means $M$ on the input $x$ runs for $O(|x|^c)$ time and $\textsc{accepts}$. 
We shall construct an $\ReachableReach$ instance $\lag \Lambda, R,\vec{c}_1,\vec{c}_2\rag$ by a polynomial time reduction from an instance $\lag M,x,0^{|x|^c}\rag$, and show that $\lag M,x,0^{|x|^c}\rag\in L$

if and only if $\lag \Lambda, R,\vec{c}_1,\vec{c}_2\rag\in \ReachableReach$. 

Without loss of generality, assume the initial configuration of $M$ is $q_1\in Q_M=\{q_1,q_2,\cdots, q_{t-2},q_A,q_R\}$ (where $t:=|Q_M|$, $Q_A$ is the accept state and $Q_R$ is the reject state), and that the TM blank the tape cells and return the tape head to the %starting 
leftmost position before halting. Let $p$ denote the maximum number of tape cells that $M$ may use on input $x$ (Note that $p\in O(|x|^c)$).
Define the following set of species:

\newcommand{\blank}{{\llcorner\negthinspace\lrcorner}}

        $$\Lambda = \{A, \underbrace{Q_1,\ldots,Q_{t-2},Q_A,Q_R}_{\textrm{machine states}}, \underbrace{P_1,\ldots,P_p}_{\textrm{head positions}}, \underbrace{T_1^0, T_1^1, \ldots,  T_p^0, T_p^1}_{\textrm{tape contents}} \}$$

    and configurations:
    \begin{eqnarray*}
        \vec{c}_1 &=& \{
        P_1,Q_1,T_{1}^{x[1]}, \ldots, T_{|x|}^{x[|x|]}, T_{|x|+1}^\blank, \ldots T_p^{\blank}\}\\
        \vec{c}_2 &=& \{1Q_A, T_1^{\blank},\ldots,T_{p}^\blank,P_1\}
\end{eqnarray*}

Recalling that the transition function $\delta_M: Q_M\setminus\{Q_A,Q_R\}\times \Gamma\rightarrow Q_M\times \Gamma\times \{-1,+1\}$, we construct the set $R$ of reactions in the following way: 

\begin{algorithm}[H]
\DontPrintSemicolon
\For{\emph{($\forall q_i\in Q_M$)($\forall b\in\{0,1,\blank\}$)($\forall k\in\{1,2,\cdots,p\}$)}}{
\If{$\delta(q_i,b)\rightarrow(q_j,b',m)$}{Add Reaction ${Q_i+T^{b}_{k}+P_k\rightarrow Q_j+T^{b'}_{k}+P_{k+m}}$\ //$m\in\{\pm 1\}$ : tape head moving direction.}}
\For{$S\in\Lambda\setminus\{A\}$}{Add Reaction ${S\rightarrow 2A}$}
 \caption{Construction of $R$}\label{SimPSpace}
\end{algorithm}

Observe that $(\Lambda, R)$ is a reachably atomic network with respect to $\Delta=\{A\}\subseteq\Lambda$, for any molecular species can be decomposed to $\{2A\}$ explicitly via Lines $11$-$12$, $A$ appears in the decomposition of all molecular species, and all reactions preserve the number of atoms. 

Further,
\begin{eqnarray*}
\lag M,X,0^{|x|^c}\rag\in L&\Leftrightarrow& M(x)\rightarrow 1\\
&\Leftrightarrow& \exists \textrm{computation path } (q_1,(x[1],x[2],\cdots,x[|x|],\underbrace{\blank,\cdots,\blank}_{p-|x|}))\Rightarrow^*(q_A,(\underbrace{\blank,\blank,\cdots,\blank}_{p}))\\
&\Leftrightarrow& \vec{c}_1\Rightarrow^*\vec{c}_2\\
&\Leftrightarrow&\lag\Lambda,R,\vec{c}_1,\vec{c}_2\rag\in\ReachableReach
\end{eqnarray*}

Finally, $|\Lambda| = 1+t+3p$; $|R| \in O(3pq+|\Lambda|)$, 
$\|\vec{c}\|_1 =\|\vec{c}\|_2 = p+2$. All coefficients of reactions are constant Hence this reduction is polynomial in $Z$ both timewise and spacewise.
}%DetailedProof of re re pspace.

\begin{proof}
%We prove the containment in $\Pspace$ by the $\Pspace$-$\Complete$ness of reachability problem in a mass conserving chemical reaction network (Observation \ref{MCComplete})
$\ReachableReach\in\Pspace$ is a direct %straightforward 
application of Observation \ref{MCComplete} -- note that all reachably-atomic chemical reaction networks are primitive atomic, and hence mass conserving (Proposition \ref{primitivemcequiv}). Hardness is shown by simulating %arbitrary 
polynomial space Turing Machines via reactions.  
    \opt{sub,subn}{For details, see Section \ref{AppAA}.}%Paragraph \ref{reachableReach}.} %Section \ref{AppAA}.}
    \opt{normal}{\DetailedProofOfReRePspaceComplete}
\end{proof}

\begin{rem}\normalfont{}
The fact that the coefficients of all reactions involved in the proof of Proposition \ref{ReachableReach} are constant also implies that $\ReachableReach$ is $\Pspace$-hard (and hence $\Complete$) in the strong sense. \opt{sub,subn}{Another remark on the irreversibility of reactions may be found in Section \ref{Auxiliary}.}%is attached in Section 
\end{rem}

\newcommand{\cantreversable}{\begin{rem}\normalfont{}
    The argument %above
    for Proposition \ref{ReachableReach} fails %for Manoj's definition as he 
    if we require all reactions to be reversible. %In fact, if that were the case
    In that case, one may decompose each molecular species to $2A$ and create any new `input species'' (species representing the tape content), so it is possible that $\vec{c}_1\reach \vec{c}_2$ (indicating {\tt accept}) yet $M(x) = 0$.  %run the disassociation reactions first and run the reversed disassociation reactions to create any ``input species'' (species representing the tape cell) as wanted, breaking the equivalence. [Go thru PROOF and EXPLAIN BETTER] [This MAY BE A WRONG ARGUMENT]
\end{rem}}
\opt{normal}{\cantreversable}
%Original remark is probably wrong since your tape content goes back to normal when you execute the reversed disassociation reaction.

\opt{sub,subn}{
    We also found connections between our definitions of %subset atomic and reachably 
    ``-atomic'' and the concept of ``core composition'', addressed by Gnacadja\cite{gnacadja2011reachability} and detailed in Section \ref{acc}. Some interesting results are:%In particular, 
    \begin{enumerate}
        \item \textbf{Lemma \ref{EConservative}} states that a network is subset atomic if and only if it admits a ``near-core composition'' with certain restrictions;
        \item \textbf{Lemma \ref{ReachableConstructive}} in the same section says reachable-atomicity implies admitting a core composition;
       \item \textbf{Theorem \ref{rvsaequiv}} gives the equivalence between ``reversibly-reachable atomic'' and ``explicitly-reversibly constructive with no isomeric elementary species''.
   \end{enumerate}   
    
    %For more details, %of relevant definitions and arguments can be found in 
    %see Section \ref{acc}.
}%A brief summary of the chapter about core compositions.

\newcommand{\SectionAtomicityWithCoreCompositionAsAppendix}{
 
 As mentioned in Subsection \ref{previous111}, there are some interesting relationship between some property of network defined in \cite{gnacadja2011reachability} and ours, which we shall look into in this section.
 
 We begin by introducing some definitions in \cite{gnacadja2011reachability}. We disclaim that all the following definitions and notations (but not remarks) before Subsection (\ref{subsection61}) are from \cite{gnacadja2011reachability}, possibly with slight modification of notationsand/or interpretations:

 \begin{defn}
 A \emph{species composition map}, or simply a \emph{composition} of chemical reaction network $\calC$ is a map $\E:\Lambda \rightarrow \N^n \setminus \{0^n\}$, where $n\in\N_{>0}$.
\end{defn}

\begin{rem}\normalfont{}
We donnot confuse $\E$ with $\vec{d}$ because they are completely different mappings. In particular, $\E$ could map different species to the same $\vec{e}_i$, which means $\E\restriction\E^{-1}(\vec{e}_i)$
 could be non-injective, while we donnot allow this for atomic composition $\vec{d}$. This point is further illustrated in the following Lemmas, such as Lemma \ref{EConservative}.  
 \end{rem}%/ is a positive integer.
 
 \begin{defn}\quad
 
 \begin{enumerate}
 
     \item A species $S\in \Lambda$ is $\E$-elementary if $\|\E(S)\|_1=1$;
     \item A species $S\in \Lambda$ is $\E$-composite if $\|\E(S)\|_1\geq 2$;
     \item $S,S'\in\Lambda$ are $\E$- isomeric if $\E(S)=\E(S')$; equivalently, we say $S,S'$ belong to the same $\E$-isomeric class.
 \end{enumerate}
 \end{defn}
 
 \begin{defn}
The \emph{extended composition} $\tme$ of $\E$ is defined as in Equation (\ref{EEEE}).  
 \end{defn}

 We denote $\Theta = $span($\{(\vec{p}-\vec{r})\}_{(\vec{r},\vec{p})\in R}$) $\subseteq\R^\Lambda$. With slight abuse of notation, we sometimes also write $\Theta = $span($R$) with $R$ the set of reaction vectors.

 \begin{defn}\label{nknk}
 a composition $\E$ is \emph{near-core} if:
 
 \begin{enumerate}
     \item $\ker(\tme)\supseteq \Theta$, which is equivalent to saying $\calC$ is $\mathcal{E}$-\emph{conservative}; and 
     \item $\vec{e}_1,\cdots,\vec{e}_n\in \range(\mathcal{E})$.
 \end{enumerate}
 
 \end{defn}

 \begin{defn}
 A composition $\E$ is \emph{core} if $\E$ is near-core and further, $\ker(\tme)=\Theta$.
 \end{defn}
 
 \begin{rem}\normalfont{}
Intuitively, a core composition $\E$ is the composition whose linear extension $\tme$ has the \textbf{smallest} kernel containing the reaction vector space $\Theta=span(R)$ as subspace. That ensures the uniqueness (up to isomorphism) 
of $\tme$%)
, and avoids including vectors that are not reachable by reactions into %non-reaction-reachable vectors in 
the kernel. Theorem $2.11$, $2.12$ of \cite{gnacadja2011reachability} has a %formal 
detailed and more formal discussion on this matter.%(Avoids making vectors non-reachable by reactions conservative.)
\end{rem}

 \begin{defn}
 A reaction network $\calC$ is \emph{constructive} if it admits a core composition.
 \end{defn}

 \begin{defn}\quad 
 
 \begin{enumerate}
     \item A species $Y$ is \emph{explicitly constructible (resp. explicitly destructible)} if there are
isomerization reactions $Y_0 \rightarrow\cdots\rightarrow Y_l$ (resp. $Y_l \rightarrow\cdots\rightarrow Y_0$), where $l\in\N$,
such that $Y_0$ is the target of a binding reaction (resp. the source of a dissociation
reaction) and $Y_l = Y$. 

Binding reactions are $(\vec{r},\vec{p})\in R$ s.t. $\|\vec{r}\|_1\geq 2$ and $\|\vec{p}\|_1=1$, and dissasociation reactions have similar definition with $\vec{r}$ and $\vec{p}$ swapped. Isomerizations are $(\vec{r},\vec{p})\in R$ s.t. $|\vec{r}|_1=|\vec{p}|_1 = 1$: note that the reactant and product of an isomerization reaction are isomers.
 
 \item A species $X$ is explicitly constructive (resp. explicitly destructive) if there is a
binding reaction $Q \rightarrow Y$ (resp. a dissociation reaction $Y \rightarrow Q$ such that $X\in [Q]$.
\end{enumerate}
 \end{defn}

\begin{rem}\normalfont{}
Intuitively, a species is explicit constructibile if it is ``eventually'' a product of a binding reaction (up to having some isomerization reactions in between),
while explicit constructivity means a species directly participate in a binding reaction as reactant. 
\end{rem}

 And finally, 
 
 \begin{defn}\label{ere}
 A chemical reaction network $\calC$ is \emph{explicitly-reversibly constructive} if:
 \begin{enumerate}
     \item $\calC$ is constructive;
     \item Each composite species is both explicitly constructible and explicitly destructible; and
     \item Each elementary species is both explicitly constructive and explicitly destructive.
 \end{enumerate}
 \end{defn}

 \subsection{Atomic chemical reaction networks with core or near core compositions}\label{subsection61}

 We would like to begin this section by showing an equivalence relationship between our definition of subset atomicity and \cite{gnacadja2011reachability}'s definition of networks admitting near-core compositions with certain restrictions. %admitting near-core co

 \begin{lem}\label{EConservative}
 A chemical reaction network $\mathcal{C}=(\Lambda, R)$ is subset atomic if and only if $\exists n\in\N_{>0},\mathcal{E}:\Lambda\rightarrow\N^{n}\setminus\{0\}^n$, s.t.
 \begin{enumerate}  
     \item $\E$ is a near-core composition of $\calC$;
     \item \label{522}$\E\restriction \E^{-1}(\{\vec{e}_i\}_{i=1}^{n})$ is one to one, and
     \item\label{523} $(\forall i\in [1,n])$ $(\exists S\in \Lambda\setminus \bigcup_{i=1}^{n}\E^{-1}(\vec{e}_i)) (\E(S))_i> 0$. 
 \end{enumerate}  %defined above satisfies $$%with respect to $\Delta$
 \end{lem}

 \begin{rem}\normalfont{}
 Adopting the definition that $\mathscr{X}_i:=\{S\in\Lambda \mid \E(S)=\vec{e}_i\}$, %then 
 condition $\ref{522}$ is saying that $p_i:=|\mathscr{X}_i|=1$ for each $i$. It intuitively translates to ``no isomerization is allowed for $\E$-elementary species''. Note also the similarity between Condition \ref{523} above and Condition \ref{primitive3} of Definition \ref{primitiveatomic}. They will translate to each other by construction in the proof below.
 \end{rem}

 \begin{proof} We note that atomic decomposition in a subset atomic network $\calC$ describes a similar phenomenon of a near-core composition of $\calC$, and we prove the equivalence by translating between the definitions.%mathematically translating %the constraints in 
 %the definitions into each other.%
 %one definition into the other, and vice versa. 
 \begin{description}
 
     \item[$\implies$:]
     
         Suppose $\calC$ is subset atomic with respect to $\Delta$ via decomposition matrix $\matr{D}$. Let $n:=|\Delta|$.
         Consider
      
         \begin{eqnarray}
         \mathcal{E}: \Lambda &\rightarrow& \N^{n}\setminus \{0^n\}:\label{EEE}\\
         S &\mapsto& \vec{d}_S,\ \forall S\in \Lambda\setminus A\nonumber\\
         A &\mapsto& \vec{d}_A=\vec{e}_{A},\ \forall A\in \Lambda
        \nonumber 
        \end{eqnarray}
         
         We shall argue that $\mathcal{E}$ has the desired property. By construction, $\vec{e}_{1} = \vec{e}_{A_1},\cdots \vec{e}_{n} = \vec{e}_{A_n}\in \mathcal{E}(\Lambda)$, and $\E$ restricted to the preimage of $\{\vec{e}_i\}_{i=1}^{n}$ is one-to-one; subset atomicity inherits Condition (\ref{primitive3}) of Definition (\ref{primitiveatomic}) (Primitive Atomic), which implies that $$(\forall i\in [1,n])(\exists S\in \Lambda\setminus \bigcup_{i=1}^{n}\E^{-1}(\vec{e}_i)=\Lambda\setminus\Delta) (\E(\vec{S}))_i=\vec{d}_S(A_i)> 0$$
         It remains to show that  $\calC$ is $\E$-conservative, which, by \cite{gnacadja2011reachability},%\todo{Add resources citing his paper.}
         \ is equivalent to $\ker(\tme)\supseteq\Theta$ where $\Theta$ is the span of the reaction vectors and %second requirement is satisfied. Now, consider 
         $\tme$ is the linear extension $\tilde{\mathcal{E}}$ of $\mathcal{E}$:
         
         \begin{eqnarray}
             \tilde{\mathcal{E}}: \R^\Lambda &\rightarrow& \R^{n}:\label{EEEE}\\
             \vec{c} &\mapsto& (\sum_{S\in\Lambda}\vec{c}(S)\cdot (\E(S))_1,\cdots, \sum_{S\in\Lambda}\vec{c}(S)\cdot (\E(S))_n)^T,\ \vec{c}\in\R^\Lambda\nonumber
         \end{eqnarray}
     
         Observe that $\mathcal{\tilde{E}}(\cdot)$ operates on %a configuration vector $\vec{c}$ or reaction vector $\vec{p}-\vec{r}$ 
         $\vec{c}\in\R^\Lambda$ as left mulitplication by $\matr{D}^T$, the transpose of the decompsition matrix.%\todo{Or its inverse $\matr{D}^T$. Unify later. No biggie. }
         By definition of subset atomicity, any reaction preserves the count/concentration of each atom, so $\matr{D}^T\cdot (\vec{p}-\vec{r}) = 0^n$ ($\forall (\vec{r},\vec{p})\in R$). Hence $a_1(\vec{p}_1-\vec{r}_1) + \cdots + a_k(\vec{p}_k-\vec{r}_k)\in\ker{\tme}$ ($k=|R|$) for any $a_1(\vec{p}_1-\vec{r}_1) + \cdots + a_k(\vec{p}_k-\vec{r}_k)\in\Theta$, as desired.
 
     \item[$\impliedby$:]
     
         Suppose we have a function $\E:\Lambda\rightarrow \N^n\setminus\{0\}^n$ satisfying the three described properties. Then define 
     
         \begin{equation}  \label{FindDelta}
        \Delta = \{\E^{-1}(\vec{e}_i)\}_{i=1}^{n}
        \end{equation}
        
        We argue that $\calC$ is subset atomic with respect to $\Delta$. Indeed, define the decomposition vectors:
        
        \begin{eqnarray}\label{54432}
         \vec{d}: \Lambda &\rightarrow& \N^n\setminus\{0\}^n,\\
         S&\mapsto& \E(S),\ \forall S\in\Lambda\nonumber
        \end{eqnarray}
        
        Since $\vec{d}$ coincides with $\mathcal{E}$ everywhere and since $\calC$ is $\mathcal{E}$-conservative, for each $S\in\Lambda$, $(\vec{r},\vec{p})\in R$ and $A_i\in\Delta$, we have

        \begin{eqnarray}
        \sum\limits_{S \in \Lambda} (\vec{p}(S)-\vec{r}(S)) \cdot \vec{d}_{S}(A_i) &=&\sum\limits_{S \in \Lambda} (\vec{p}(S)-\vec{r}(S)) \cdot (\E(S))_i\nonumber\\
        &=& (\tme(\vec{p}-\vec{r}))_i\nonumber\\
        &\underbrace{=}_{\E\textrm{ conservative}\Rightarrow \vec{p}-\vec{r}\in\ker{\tme}}&0,
        \end{eqnarray} 
        which gives the atom-preservation condition of subset atomic. %"," above, so donot do new paragraph.

        By construction, $(\forall i)A_i\in\Delta\subseteq \Lambda$,  $\vec{d}_{A_i} = \vec{e}_i$%, so the third condition is satisfied.
        ; To see that $\|\vec{d}(S)\|_1\geq 2$ for all $S\in\Lambda\setminus\Delta$, recall that $\E\restriction \E^{-1}(\{\vec{e}_i\}_{i=1}^{n})$ is one to one, which means $\forall S\in \Lambda\setminus\Delta$, $\E(S)\not\in\{\vec{e}_i\}_{i=1}^n$. Given that $\range(\E) = \N^{n}\setminus\{0\}^n$, this means $\vec{d}(S)$ is some non-trivial linear combination of $\vec{e}_i$'s, which gives $\|\vec{d}(S)\|_1\geq 2$ as desired;
        
        Lastly, $(\forall i\in [1,n])$ $(\exists S\in \Lambda\setminus \bigcup_{i=1}^{n}\E^{-1}(\vec{e}_i)) (\E(S))_i> 0$ translates to $(\forall A\in\Delta)(\exists S\in\Lambda\setminus\Delta)A\in [\vec{d}_S]$ by definition of $\vec{d}$.\qed
 \end{description}%\renewcommand{\qedsymbol}{\square} %This implies that $S$, defined as span of reaction vectors, is contained in $\ker(\mathcal{\tilde{E}})$.
 \end{proof}
 
 In order to further describe the relationship between atomic networks and core compositions, we make the following definitions first.

\begin{defn}[Associated Composition]
Given a subset atomic chemical reaction network $\mathcal{C}=(\Lambda, R)$ with respect to $\Delta$ via decomposition matrix $\matr{D}$%^{\Delta\times\Lambda}$
, the \emph{associated composition} of $\matr{D}$ is the function $\mathcal{E}$ constructed in (\ref{EEE}). %the proof of Lemma \ref{EConservative}. 
$\E$'s unique linear extension $\tme$, constructed in (\ref{EEEE})%the same proof
, is defined as the \emph{extended associated composition} of $\matr{D}$.
\end{defn}

\begin{rem}\normalfont{}
Note that $\tme$, the linear extension of $\E$, is defined in (\ref{EEEE}) independent of atomic decompostions.
\end{rem}

Next we prove that reachably atomic networks admits a core composition. But first, we exhibit some auxiliary definitions.

\begin{defn}
Given a subset atomic chemical reaction network $(\Lambda, R)$ with respect to $\Delta$ via $\matr{D}$, a \emph{single-molecule decomposition vector} $\vec{d}^{'}_{S_i}\in \N^{\Lambda}\setminus \{0^n\}$ is defined as $\vec{d}_{S_i} - \vec{e}_{S_i}$. The set of single-molecule decomposition vector is denoted as $U := \{\vec{d}^{'}_{S_i}\}_{S_i\in \Lambda\setminus\Delta}$.
\end{defn}

Recall that $\vec{d}_{S_i}$ is the decomposition vector of $S_i$ whose first $|\Lambda|-|\Delta\cap\Lambda|$ coordinates are $0$ and last $|\Delta|$ coordinates correspond to the count of each atom in the molecule. for subset atomic networks, $\vec{d}_{S_i}\in \N^{\Lambda}\setminus \{0^n\}$, so $\vec{d}^{'}_{S_i}$ is well-defined by replacing $0$ with $-1$ on the $i$-th molecular position.

We explore the relationship between atomicity and core compositions by inspecting into the relationship between $\ker(\tme)$ and $\Theta$, the span of reaction vectors. This is in turn done by inspecting the relation between $\ker{\tme}$ and the space spanned by $U$.%But first, 

Next, let $\Upsilon$ denote the vector space spanned by $U$, as a subspace of $\R^n$.%\end{nt}
We now give the first approach to the implication ``reachably atomicity $\Rightarrow$ Core-Composition Admission''. To do this, we will prove that for subset atomic networks with $\tme$ defined as previously defined in (\ref{EEEE},\ref{EEE}), $\ker{\tme} = \Upsilon$; for reachably atomic networks (which are by definition also subset atomic), $\Upsilon \subseteq \Theta$. The two relations combined would give $\ker{\tme} \subseteq \Theta$, which is exactly the missing bit from near-core to core compositions. The proofs will be carried out from Lemma \ref{8basis} through Lemma \ref{ReachableConstructive}.

\begin{lem}\label{8basis}
   Vectors in $U$ are linearly independent. Since they also span $\Upsilon$, this means $U$ is a basis for $\Upsilon$.
\end{lem}

\begin{proof}
Observe that the $-1$ on the $i$-th position ($\forall 1\leq i\leq  |\Lambda|-|\Delta|$) cannot be obtained by linear combination of other vectors in $U$, the $i$-th position of which are all $0$'s.

This also shows that  $\dim(\Upsilon)=|U|=|\Lambda-\Delta| = |\Lambda|-|\Delta|$.
\end{proof}

\begin{lem}[Kernel-Span Equivalence]\label{kerspan}
For subset atomic networks with $\tme$ defined as in (\ref{EEEE},\ref{EEE}), $\ker{(\tme)}= \Upsilon$.
\end{lem}%The ``$\Leftarrow$'' direction of (\ref{12}) is trivial. 

\begin{proof} By verification of definitions.
\begin{enumerate}
    \item $\Upsilon\subseteq \ker{(\tme)}$:
    
    Take $\vec{u} = a_1\vec{d}^{'}_{S_1} +\cdots +a_{|\Lambda|-|\Delta|}\vec{d}^{'}_{S_{|\Lambda|-|\Delta|}}\in \Upsilon$. 
    Then $\forall i\in[1,n]$, we have
    \begin{eqnarray*}
     (\tme(\vec{u}))_i&=&\sum_{S_j\in\Lambda}\vec{u}(S_j)\cdot (\E(S_j))_i\\
     &\underbrace{=}_{(\ref{EEE})}& \sum_{S_j\in\Lambda\setminus\Delta}\vec{u}(S_j)\cdot \vec{d}_{S_j}(A_i) + \sum_{A_k\in\Delta}\vec{u}(A_k)\cdot \vec{d}_{A_k}(A_i)\\
     &=& \sum_{S_j\in \Lambda\setminus\Delta} (-a_j)\cdot\vec{d}_{S_j}(A_i)  + \vec{u}(A_i)\cdot 1\\
     &=& \sum_{S_j\in \Lambda\setminus\Delta} (-a_j)\cdot\vec{d}_{S_j}(A_i)  + \sum_{S_j\in\Lambda\setminus\Delta} a_j\cdot\vec{d}_{S_j}(A_i)\\
     &=&0
    \end{eqnarray*}
    
    \item $\Upsilon\supseteq \ker{(\tme)}$:
    
    Take $\vec{v}\in\ker{(\tme)}$, and let $a_{1},\cdots,a_{|\Lambda|-|\Delta|}$ denote the first $|\Lambda|-|\Delta|$ coordiantes of $\vec{v}$. We claim that 
    
    \begin{equation} \label{vdp}\vec{v} = -\sum_{i=1}^{|\Lambda|-|\Delta|}a_i\vec{d}^{'}_{S_i}
    \end{equation}

    Indeed, for each $S_j\in \Lambda\setminus\Delta$, $$\vec{v}(S_j) = -a_j\cdot (-1) = \sum_{i\neq j}(-a_i)\cdot\underbrace{\vec{d}^{'}_{S_i}(S_j)}_{0} + (-a_j)\cdot\underbrace{\vec{d}^{'}_{S_j}(S_j)}_{-1} = -\sum_{i=1}^{|\Lambda|-|\Delta|}a_i\vec{d}^{'}_{S_i}(S_j)$$
   It remains to verify that $(\ref{vdp})$ holds the last $|\Delta|=n$ positions. %of $\vec{v}$ and $-\sum_{i=1}^{|\Lambda|-|\Delta|}a_i\vec{d}^{'}_{S_i}$. 
    Let $b_1,\cdots, b_n$ denote the last $n$ positions of $\vec{v}$, $c_1,\cdots, c_{n}$ denote the last $n$ positions of $-\sum_{i=1}^{|\Lambda|-|\Delta|}a_i\vec{d}^{'}_{S_i}$. Because $\vec{v}\in\ker{(\tme)}$, we know that for each $i\in[1,n]$,
    
    \begin{eqnarray*}
     (\tme(\vec{v}))_i&=&\sum_{S_j\in\Lambda}\vec{v}(S_j)\cdot (\E(S_j))_i\\
     &\underbrace{=}_{(\ref{EEE})}& \sum_{S_j\in\Lambda\setminus\Delta}\vec{v}(S_j)\cdot \vec{d}_{S_j}(A_i) + \sum_{A_k\in\Delta}\vec{v}(A_k)\cdot \vec{d}_{A_k}(A_i)\\
     &=& \sum_{S_j\in \Lambda\setminus\Delta} (a_j)\cdot\vec{d}_{S_j}(A_i)  + \vec{v}(A_i)\cdot 1\\
          &=& \sum_{S_j\in \Lambda\setminus\Delta} (a_j)\cdot\vec{d}_{S_j}(A_i)  + b_i\cdot 1\\
     &=& 0,
    \end{eqnarray*}%Recall that for any $j\neq i$, $\vec{d}^{'}_{S_j}(S_i) = 0$.
    which gives that
    \begin{eqnarray*}
    b_i &=& -\sum_{j\in 1}^{|\Lambda|-|\Delta|} a_j\cdot \vec{d}_{S_j}(A_i)\\
    &=& -\sum_{j\in 1}^{|\Lambda|-|\Delta|} a_j\cdot \vec{d}^{'}_{S_j}(A_i)\\
    &=& c_i
    \end{eqnarray*}
    as desired, completing the proof.\qed
\end{enumerate}%\renewcommand{\qedsymbol}{\square}
\end{proof}

\begin{rem}\normalfont{}
The lemmas above directly imply that $\dim(\ker(\tme)) = \dim(\Upsilon)= |\Lambda|-|\Delta|$.
\end{rem}

The Remark above resonates Gilles' Theorem $3.3$, which gives a %more general 
equation for general cases where $\E\restriction\E^{-1}(\{\vec{e}_i\}_{i=1}^{n})$ is not necessarily one-to-one. Since the general cases are not directly related to our discussion on network atomicity,\footnote{For the fact that any species cannot have two atomic decompositions in a single decomposition matrix. Note that molecular species may admit different atomic decompositions, but they belong to different decomposition matrices (that is, atomic decompositions of other species have to change accordingly).} we refrain from further discussion thereon. %omit detailed discussion thereof.%. As usual, this is the special case where $p_i=1$ for all $i$.

\begin{lem}\label{ReachableConstructive}
If a %subset atomic  
chemical reaction network $\calC$ is reachably atomic, then $\calC$ admits a core composition. That is, $\calC$ is constructive.
\end{lem}

\begin{proof}
By Lemma \ref{EConservative}, since $\calC$ is reachably atomic and hence subset atomic, it admits a near-core composition. By Lemma \ref{kerspan} and Definition \ref{nknk}, $\ker(\tme)=\Upsilon\supseteq\Theta$, so it suffices to prove $\Upsilon\subseteq \Theta$ when $\calC$ is reachably atomic.

%To prove $\Upsilon\subseteq \Theta$, r
Recall that by Lemma \ref{8basis}, $U$ is a basis for $\Upsilon$. Hence we only need to  argue that each basis vector $\vec{d}'_{S_i}\in U$ is a linear combination of reaction vectors. Indeed, $\forall S_i\in\Lambda\setminus\Delta$, $\vec{e}_{S_i}\reach\vec{d}_{S_i}$, so there exists $\vec{p}_1-\vec{r}_1,\cdots,\vec{p}_k-\vec{r}_k\in R$ for some $k$ with $\vec{r}_1=\vec{e}_{S_i}$ s.t. $\vec{e}_{S_i}+\sum_{k}(\vec{p}_k-\vec{r}_k) = \vec{d}_{S_i}$. But this means $\vec{d}'_{S_i} = \vec{d}_{S_i}-\vec{e}_{S_i} = \sum_{k}(\vec{p}_k-\vec{r}_k)\in \Theta$.

Having proved $\ker{\tme} = \Upsilon = \Theta$, we conclude that $\calC$ admits a core composition, as desired.
\end{proof}

\begin{rem}\normalfont{} 
    We exhibit an alternative approach for Lemma \ref{ReachableConstructive} by directly applying Theorem $4.2$ of \cite{gnacadja2011reachability}.

    Theorem $4.2$ of \cite{gnacadja2011reachability} states that if $\E: \Lambda\rightarrow \N^n\setminus\{0^n\}$ satisfies the following, then $\E$ is a core composition:
    \begin{enumerate}
    \item $\E$ is near-core;
    \item all $\E$-elementary species of the same $\E$-isomeric class are stoichiometrically-isomeric ($X,Y\in\mathscr{X}_i\Rightarrow Y-X\in \Theta=\textrm{span}(R)$); here, $\E$-isomeric classes are defined as species having the same $\E(\cdot)$-value. Restricted to $\E$-elementary species, the $\E$-isomeric classes are $\mathscr{X}_i:=\{S\in\Lambda \mid \E(S)=\vec{e}_i\}$ ($i\in [1,n]$).
    \item\label{300} For every $\E$-composite $Y$ with $\E(Y)=\alpha\in \Z_{\geq 0}^n\setminus\{0^n\}$, there exist one ``representative'' elementary species from each elementary isomeric class such that $Y$ and the $\alpha$-linear combination of these elementary species are stoichiometrically compatible. That is, $\exists W_1\in\mathscr{X}_1,\cdots,W_n\in\mathscr{X}_n$, s.t. $Y-\sum_{i=1}^{n}\alpha_iW_i\in \Theta$. 
    \end{enumerate}

    Note that for a reachably atomic network $\calC$ and its associated composition $\E$, all elementary $\E$-isomeric classes are singleton, so $(\forall i)(\forall A_i\in\mathscr{X}_i)(A_i-A_i=0^n\in \Theta)$\footnote{Note that this doesn't contradict the assumption that $(\vec{r},\vec{r})\not\in R$, since $0^n = 1\cdot (\vec{p}-\vec{r})+(-1)\cdot (\vec{p}-\vec{r})$ can be obtained by linear combination of reactions $(\vec{r},\vec{p})\in R$ where $\vec{r}\neq\vec{p}$.}; Also, for each $Y\in\Lambda\setminus\Delta$, $Y-\sum_i\underbrace{\vec{d}_{Y}(A_i)}_{\alpha_i}A_i = (-1)\cdot \vec{d}'_Y\in \Theta$, as $\vec{d}'_Y = \vec{d}_Y-\vec{e}_Y=\sum_j(\vec{p}_j-\vec{r}_j)\in\Theta$ is guaranteed by the reachably atomicity. The reachably atomicity implies subset atomicity, which again implies $\E$ is near-core. Therefore $\calC$ adopts a core composition, as desired.
\end{rem}

 Although in the alternative proof exhibited above we talked about the concept of ``span'' as linear combinations with real coefficients, in fact our model of reachably atomic networks is not necessarily equipped with the property that each reaction is reversible. This fact does not break the proof, but it justifies our decision to keep both approaches. Further, the above observation indicates the possibility that networks with the reversible property may itself guarantee some more interesting structures. In fact, let us study the relationship between reversibly-reachably atomic networks and explicitly-reversibly constructive networks in the following section.

  \subsection{Reversibly-Reachably Atomicity and Explicitly-Reversibly Constructiveness}
  In this section we show an equivalence between our definition of \emph{reversibly-reachably atomic} and \cite{gnacadja2011reachability}'s definition of \emph{explicitly-reversibly constructive} with an additional restriction.
  
  We first note that the following subclass of reachably atomic networks has such a property that each molecular species can be explicitly constructed from its atomic makeup via reactions.

\begin{defn}[Reversibly-Reachably Atomic]\label{rvsa} A chemical reaction network $\calC=(\Lambda, R)$ is \emph{reversibly-reachably saturated atomic} if:
    \begin{enumerate}
    \item It is reachably atomic with respect to some $\Delta\subseteq \Lambda$ via the decomposition matrix $\matr{D}$;%vectors $\vec{d}_{S}\in \N^{\Lambda}$ ($S\in \Lambda$);
    \item \label{1bb}$\forall S_j\in\Lambda\setminus\Delta$, $\vec{d}_{S_j}\Rightarrow^*\{1S_j\}$.
    \end{enumerate}
\end{defn}

By definition, reversibly-reachably atomic networks $\subsetneq$ reachably atomic networks $\subsetneq$ subset atomic networks $\subsetneq$ primitive atomic networks; or equivalently, for a given $\calC$, primitive atomicity $\Rightarrow$ subset atomicity $\Rightarrow$ reachably atomicity $\Rightarrow$ reversibly-reachably atomicity, while the reversed arrows do not necessarily hold.

Correspondingly, we define the following language of encodings of networks with the reversibly-reachably atomic property:

\begin{defn}
    \begin{eqnarray*}
    \RevReachablyAtomic &=& \{\lag\Lambda,R\rag\mid (\exists\Delta\subseteq \Lambda)((\Lambda, R)\textrm{ is reversibly-reachably-} \\&&\textrm{atomic with respect to }\Delta)\}
    \end{eqnarray*}
\end{defn}

\begin{cor}\label{420} $\textsc{Reversibly-reachably atomic}\in\textsf{P}$.
\end{cor}

\begin{proof}
    The proof %is almost mirroring 
    highly resembles the previous proof of Theorem \ref{VisAtomP}. It takes polynomial time to decide if an instance $\lag\Lambda,R\rag\in\ReachablyAtomic$, as shown above. We extend the $\ReachablyAtomic$ decider to decide whether $\calC$, %after 
    having been confirmed to be reachably atomic, further satisfies $\vec{d}_S\Rightarrow^*\{1S\}$ for each $S\in \Lambda\setminus\Delta = M$:
    
    Construct $M''=M$. While $M''$ is not empty, iterate and try to find an $S''\in M''$ that satisfies this condition: $\exists (\vec{r},\vec{p})\in R$ s.t. $\mathbf{\vec{p}=\{1S''\}}$ and $\mathbf{[\vec{r}]\subseteq (M\setminus M')\cup \Delta}$. Note that all elements $S''$ in $M\setminus M''$ satsifies $\{1S''\}\Leftarrow^* \vec{d}_{S''}$, hence if $\{S''\}$ is the product of a reaction whose reactants %decomposed into complexes 
     consist solely of elements in $(M\setminus M'')\cup\Delta$ %via execution of one reaction
     , then $S''$ itself satisfies $\{1S''\}\Leftarrow^* \vec{d}_{S''}$ as well. Keep track of $\vec{d}_{S''}$ and exclude such $S''$ from $M''$.
     
     If in some iteration we cannot find such $S''\in M''$, then \textsc{reject}; else, the iteration will finally halt excluding all such $S''$'s and making $M''$ empty, in which case we \textsc{accept}.
     
     Proof of correctness works analogously as the proof above, with the $\Rightarrow^*$ reversed to $\Leftarrow^*$ and disassociation reactions changed into association reactions for consideration. %The same ``finiteness of $M''$ implies '' 
     
     We omit the pesudocode for the algorithm described above, as it highly resembles Lines $25$-$35$ of Algorithm \ref{VisAtomPolyTime} %\todo{Number} 
     (with $\vec{r}$ and $\vec{p}$ reversed) and can be reconstructed from the verbal description above. Lastly, the complexity is dominated by the reachably atomic deciding process.

\end{proof}

\begin{theorem}\label{rvsaequiv}
For a chemical reaction network $\calC$, the following are equivalent:
    \begin{enumerate}
        \item $\calC$ is reversibly-reachably  atomic; 
        \item\label{200} $\calC$ %has a core composition $\E$ 
        is explicitly-reversibly  constructive, with %$\E\restriction \E^{-1}(\{\vec{e}_i\}_{i=1}^{n})$ one to one 
        $p_i=1\ (\forall i\in[1,n])$ where $p_i:=|\mathscr{X}_i| = |\E^{-1}{(e_i)}|$.  \end{enumerate}   
\end{theorem}

Note that condition (\ref{200}) translates to ``no $\mathcal{E}$-elementary species is $\mathcal{E}$-isomeric''.

 \begin{proof} 
    Proof is done by similar techniques used for Lemma \ref{viswithdirect}: that is, assuming otherwise, then there will be an infinite descending chain of species ordered by number of atoms in their respective decomposition. This contradicts the fact that the set of species is finite.
    \begin{description}
     \item[$1 \implies 2$:]
         By Definition $\ref{rvsa}$, $\calC$ is reachably atomic with respect to some $\Delta\subseteq \Lambda$. Let $|\Delta|=n$, then $\calC$ is constructive by lemma (\ref{ReachableConstructive}). By Lemma (\ref{EConservative}), since a reachably atomic network with $|\Delta|=n$ is subset atomic, $p_i=1$ for each $i\in [1,n]$.
         
         the weakly reversibility of $\{1S_i\}\Rightarrow^* \vec{d}_{S_i}$ ensures that each composite species is both explicitly constructible and explicitly destructible, because it is ensured that each molecular species directly or indirectly participates in at least one disassociation (resp. binding) reaction as reactant (resp. product);

          We also claim that each elementary species is both explicitly constructive 
          and explicitly destructive:
          
          suppose for the sake of contradiction $A_i\in\Delta$ is not explictly destructive. Then for any reactions where $A_i$ participates as product, the reactant has to contain at least $2$ species. In particular, since $A_i\in [\vec{d}_{S_i}]$ for some $S_i\in\Lambda\setminus\Delta$, if we consider the  
          sequence $\{1S_i\}\Rightarrow^*\vec{d}_{S_i}$, there must be a reaction in this sequence written as 
          
          \begin{equation}\label{YQ}\sum_{j=1}^k a_jY_j \rightarrow  \sum_{j=1}^s b_sQ_s
          \end{equation}
          with the multisets $\{a_jY_j\}_{j=1}^{k}\neq \{b_jQ_j\}_{j=1}^{s}$ and  $A_i\in\{Q_j\}_{j=1}^{s}$, s.t. 
          
          \begin{equation}\label{smallerYQ}\exists S_j\in(\Lambda\setminus\Delta)\cap \{Y_j\}_{j=1}^{k}\ with\ A_i\in[\vec{d}_{S_j}].\end{equation}

          This is because $A_i$ cannot be directly obtained from $\{1S_i\}$, and hence must be obtained from some intermediate molecular species. Apparently $||\vec{d}_{S_j}||_1 < \|\vec{d}_{S_i}\|_1 $ by conservativity of atoms. But then consider the decompsoition series $\{1S_i\}\Rightarrow^*\vec{d}_{S_j}$ and apply the same argument, we obtain $S'_{j}\in \Lambda\setminus\Delta$ s.t. $\|\vec{d}_{S'_{j}}\|_1<||\vec{d}_{S_j}||_1 < \|\vec{d}_{S_i}\|_1$. By the finiteness of $\Lambda\setminus\Delta$, this repeated process terminates with a ``smallest'' molecule containing $A_i$; that is, there exists some $S_m\in\Lambda\setminus\Delta$ where $A_i\in[\vec{d}_{S_m}]$ and $\forall n\neq m$, if $A_i\in[\vec{d}_{S_n}]$, then $\|\vec{d}_{S_m}\|_1<\|\vec{d}_{S_n}\|_1$. 
          
          But now we once again apply the argument (involving equations (\ref{YQ},\ref{smallerYQ}) above, getting some $S'_{m}\in \Lambda\setminus\Delta$ with $A_i\in [\vec{d}_{S'_m}]$ s.t. $\|\vec{d}_{S'_m}\|_1<\|\vec{d}_{S_m}\|_1$, a contradiction.
          
          Symmetric argument with ``products'' and ``reactants'' swapped, given the reachability of $\vec{d}_{S_i}\Rightarrow^*\{1S_i\}$ for each $S_i\in\Lambda\setminus\Delta$, proves that all $A_i\in\Delta$ also have to be explicitly constructive.
      
     \item[$2 \implies 1$:]
     
         Because $\calC$ admits a core composition, in particular it admits a near-core composition. The one-to-one condition allows us to define the set of atoms $\Delta$ and decomposition $\vec{d}$ as (\ref{FindDelta}) and $(\ref{54432})$ in the proof of Lemma (\ref{EConservative}).
         
         We first argue that condition $(\ref{primitive3})$ of Definition $\ref{primitiveatomic}$ holds. Suppose not, then $\exists A_i\in\Delta$ s.t. $\forall S_j\in\Lambda\setminus\Delta$, $A_i\not\in[\vec{d}_{S_j}]$. But then consider the binding reaction $Q\rightarrow Y$ where $|Y|=1$ and $A_i\in [Q]$; such a reaction has to exist because of the explicit constructivity of $A_i$. $Y$ has to be a molecule containing $A_i$ in $[\vec{d}_{Y}]$, a contradiction.
         
         By definition $\ref{54432}$, the above implies that 
         $(\forall i\in [1,n])$ $(\exists S\in \Lambda\setminus\Delta = \Lambda\setminus \bigcup_{i=1}^{n}\E^{-1}(\vec{e}_i)) \vec{d}_S(A_i)=(\E(S))_i> 0$.
         Together with the condition $p_i=1\ (\forall i\in[1,n])$ and the near-core property, we use the equivalence in Lemma (\ref{EConservative}) to conclude that $\calC$ is subset atomic.

         Next, we argue that $(\ref{200})\Rightarrow \forall S_j\in\Lambda\setminus\Delta$, $\vec{d}_{S_j}\Rightarrow^* \{1S_j\}$ and $\{1S_j\}\Rightarrow^*\vec{d}_{S_j}$. Note that the latter reachability would also imply the reachable-atomicity of $\calC$, given that $\calC$ is subset atomic and that atomic decompisition is natually unique for reachably atomic networks, by Lemma \ref{uniquereachable}.
         
         Consider an arbitrary $S_i\in\Lambda\setminus\Delta$. For the sake of contradiction, assume $\{1S_i\}\not\Rightarrow^*\vec{d}_{S_i}$. Since $S_i$ is explicitly destructible, some reaction sequence starting with $S_i$ has to eventually split into a complex $Q_i$ with $|Q_i|\geq 2$. On the other hand, since $\calC$ is subset atomic, any reaction sequence starting with $\{1S_i\}$ either reaches $\vec{d}_{S_i}$ or reaches some configuration $\vec{c}$ with $[\vec{c}]\cap (\Lambda\setminus\Delta)\neq\emptyset$. Since we assumed that $1S_i\not\Rightarrow^*\vec{d}_{S_i}$, it must be the second case, and in particular, there exists $S_j\in [Q_i]\cap (\Lambda\setminus\Delta)$. Apparently $\|\vec{d}_{S_j}\|_1 < \|\vec{d}_{S_i}\|_1$ by conservativity of number of atoms.
         
         We apply the same argument to $S_j$ and obtain $S'_j$ s.t. $\|\vec{d}_{S'_j}\|_1<\|\vec{d}_{S_j}\|_1 < \|\vec{d}_{S_i}\|_1$. Repeat this process, and by finiteness of $\Lambda\setminus\Delta$, we'll find some $S_m\in\Lambda\setminus\Delta$ satisfying $\forall m'\neq m$, $||\vec{d}_{S_{m'}}||_1>||\vec{d}_{S_m}||_1$. Application of the same argument to $S_m$ yields the contradiction as to the size of decomposition vector.
         
         The same ``infinite descending chain'' argument applies to the other direction, with all the arrows reversed and the explicit constructibility property applied. This proves the reachable-atomicity as well as condition (\ref{1bb}) in Definition $\ref{rvsa}$.\qed
 \end{description}%\renewcommand{\qedsymbol}{\square}
 \end{proof}

%\begin{corollary}
\begin{cor}
The problem ``Given a chemical reaction network, is it explicitly reversibly constructible with no isomeric elementary species'' 
 as well as the problem ``Given a chemical reaction network $\calC$, is $\calC$ reversibly-reachably atomic'' are both polynomial time decidable.
\end{cor}%ollary}

\begin{proof}
Immediate from Corollary \ref{420} and Theorem \ref{rvsaequiv}.
\end{proof}
}%End of Section Atomicity with Core Composition as Appendix.

\opt{normal}{\section{Atomicity with Core Composition}\label{acc}\SectionAtomicityWithCoreCompositionAsAppendix}%Put this chapter in the appendix and full. 

\section{Open Problems}\label{openopen}

\begin{conj}
\SubsetAtomic\ $\in\NP$-$\Complete$.
\end{conj}

One may note that there are two sources of indeterminancy in the problem $\SubsetAtomic$: the choice of $\Delta$ and %the choice of 
$\matr{D}$. For example, the network constructed in the proof of $\NP$-hardness of $\SubsetFixedAtomic$ would remain subset atomic if we define $\Delta = \{T,F\}$, and let $\vec{d}_P=\vec{d}_Q = \{kT , sF\}$ for any $k,s\geq 2$.   

There is a formal sense in which chemical reaction networks have been shown to be able to compute functions $f:\N^k\to\N$~\cite{CheDotSolNaCo} and predicates $\N^k\to\{0,1\}$~\cite{angluin2006passivelymobile}.
A function/predicate can be computed ``deterministically'' 
(i.e., regardless of the order in which reactions occur) 
$\iff$ it is semilinear (see~\cite{ginsburg1966} for a definition).

\begin{prob}\label{powerpower}
%What is the computational power of atomic chemical reaction networks%? 
What semilinear functions/predicates can atomic chemical reaction networks compute \emph{deterministically}, and how efficiently?
What general functions/predicates can atomic chemical reaction networks compute \emph{with high probability}, and how efficiently?
\end{prob}

\newcommand{\allsemilinear}{In fact, in their proof of the %following Lemma
lemma that any semilinear function $f:\N^k\rightarrow \N$ can be stably computed by a chemical reaction network, Chen, Doty and Soloveichik \cite{CheDotSolNaCo} designed a chemical reaction network which can be made primitive-atomic by a slight modification: %constructed a chemical reaction system which can be shown primitive-atomic:
%\begin{lemma}abc
%\end{lemma}
%
%The chemical reaction system constructed are:
\begin{eqnarray}
T_i + \widehat{Y}_{i,j}^P&\rightarrow& T_i + Y_{i,j}^P + Y_j\label{704}\\
F_i + Y_{i,j}^P &\rightarrow &F_i + M_{i,j}\\
Y_j + M_{i,j} &\rightarrow& \widehat{Y}_{i,j}^P\\
T_i + \widehat{Y}_{i,j}^C&\rightarrow& T_i + Y_{i,j}^C\\
F_i  + Y_{i,j}^C
&\rightarrow& F_i+ \widehat{Y}_{i,j}^C\\
Y_{i,j}^{P}+Y_{i,j}^{C}&\rightarrow& K_j\label{709}\\
K_j + Y_j&\rightarrow& W_j\label{badboy}\end{eqnarray}
%
%Let us first modify the apparently non-mass conserving reaction, (\ref{badboy}), into
%\begin{eqnarray}
%%    K_j + Y_j&\rightarrow& W_j\label{gooddog}
%\end{eqnarray}

The modified chemical reaction network $(\ref{704})\sim (\ref{709})$ still stably computes the same function $f$, as the waste product $W_j$'s do not participate in any other reactions. This network is %apparently 
mass-conserving, via the mass distribution function
\begin{eqnarray}
    \vec{m}: \{T_i,F_i,Y_j,K_j,W_j,Y_{i,j}^P,\widehat{Y}_{i,j}^P, Y_{i,j}^C, \widehat{Y}_{i,j}^C, M_{i,j}\}_{i,j}&\rightarrow& \R\\
    T_i,F_i,Y_j,Y_{i,j}^P, M_{i,j},Y_{i,j}^C,\widehat{Y}_{i,j}^C &\mapsto& 2,\\
    \widehat{Y}_{i,j}^P, K_j&\mapsto& 4,\\
    W_j &\mapsto& 6
\end{eqnarray}

And by setting $\Delta = \{A\}$ and $\vec{d}_{S}(A) = m(S)$ $(\forall S\in \{T_i,F_i,Y_j,K_j,W_j,Y_{i,j}^P,\widehat{Y}_{i,j}^P, Y_{i,j}^C, \widehat{Y}_{i,j}^C, $ $M_{i,j}\}_{i,j} = \Lambda$), we find that the network above is also primitive atomic. This shows that any semilinear function can be stably computed by a primitive-atomic chemical reaction network.

Redefining $\Lambda \leftarrow \Lambda\cup \Delta$, we obtain a subset-atomic network that stably computes $f$, which implies that the computation power of subset-atomic chemical reaction networks are no weaker than primitive-atomic chemical reaction networks.
} %Assigning mass
%Maybe need $W_j \rightarrow 6A$ in order to maintain the ``saturated'' condition. He doesn't seem to care, though.
\begin{rem}\label{PASAWorksForAllSemilinear}\normalfont{} A partial answer for Problem \ref{powerpower} based on results in \cite{CheDotSolNaCo} says that primitive atomic networks and subset atomic networks can stably compute any seminilear functions \opt{sub,subn}{(For the proof of this, see Section \ref{AppAA})}, but it is not obvious  how to modify the subset-atomic network into reachably-atomic with the stably-computation property maintained, or whether it is even possible to do so. 
 %for primitive-atomic and subset-atomic chemical reaction networks. 

 %In fact, If $T_i$'s and $F_i$'s above are atoms, then there must be reactions for all molecular species to decompose into $T_i$ $F_i$'s or other atoms like $A$, which breaks the stable computation; if $T_i$'s or $F_i$'s are molecules, they themselves must decompose, which also break the computation.

%Modifying to Subset-Atomic. Mention that this system is not trivially modifable into reachably-atomic with the stably computation property maintained.

\opt{normal}{\allsemilinear}
\end{rem}

%\subparagraph*{Acknowledgements.}
\paragraph{Acknowledgements}{}
The authors are thankful to Manoj Gopalkrishnan, Gilles Gnacadja, Javier Esparza, Sergei Chubanov, Matthew Cook, and anonymous reviewers for their insights and useful discussion.

\opt{sub,subn}{
    \newpage
}

\bibliography{tam,MPP}
\newpage{}

\opt{sub,subn}{%The appendix, yay.
    \section{Appendix}
        \appendix
        \section{More Discussion on Related Works}% Addressing Different Aspects of Atomicity }
        \label{relatedbutnotsame}
        \Egrelatedbutnotsameone
        
        \Egrelatedbutnotsametwo
        \section{Proof Details of Theorems}\label{AppAA}

    \newcommand{\StatementofUniquenessOfDecompositionInReachAtomicCRNJustTheStatement}{If $(\Lambda, R)$ is reachably atomic, then the choice of $\Delta$ with respect to which $(\Lambda, R)$ is reachably atomic is unique. Moreover, $(\forall S\in\Lambda) \vec{d}_S$ is unique, meaning $$(\forall S\in \Lambda)(\forall \vec{c}\in\N^\Lambda) (((\{1S\}\Rightarrow^*\vec{c})\wedge (\vec{c}\neq\vec{d}_{S}))\Rightarrow [\vec{c}]\cap (\Lambda\setminus\Delta)\neq\emptyset)$$
    }%For the appendix.

        \paragraph{Lemma \ref{uniquereachable}: \StatementofUniquenessOfDecompositionInReachAtomicCRNJustTheStatement}\label{statementofUniquenessOfDecompositionInReachAtomicCRNJustTheStatement}
        
        \begin{proof}
        \DetailedProofOfUniquenessOfDecompositionInReachablyAtomic
        \end{proof}
        
        \paragraph{Proposition \ref{primitivemcequiv}: \StatementofMCPAEquivalence}\label{statementofMCPAEquivalence}
        
        \begin{proof}
        \DetailedProofOfMCPAEquivalence
        \end{proof}
        
        \RemarkOnMCPAEquivalence
        
        \paragraph{Proposition \ref{SFANP}: \StatementOfLemmaSSFixedAReducibleToIP}\label{statementOfLemmaSSFixedAReducibleToIP}

        \newcommand{\NotationOfProofSSFixedAReducibleToIP}{
%Only used in Appendix. Kept.
Before proving the lemma, we make the following notations for the sake of simplicity.%clarity.
 
 \begin{nt}\label{not1}\normalfont{}
 Throughout the proof of Proposition \ref{SFANP}, Corollary \ref{cor:SFANPNP} and Corollary \ref{cor:SFANPNP}, we assume $|\Lambda| = m$, $|R| = k$, $|\Delta| = n$. 
 We label the elements of $\Lambda$ as $S_1,S_2,\cdots, S_m$, the elements of $R$ as $(\vec{r}_1,\vec{p}_1),(\vec{r}_2,\vec{p}_2),\cdots,(\vec{r}_k,\vec{p}_k)$.
 Without loss of generality, assume the elements of $\Delta$ are
  $S_{m-n+1},\cdots,S_{n}$, 
 which for ease of notation we refer to as $A_1,A_2,\ldots,A_n$.
 Thus the elements of $\Lambda\setminus\Delta$ are $S_1,S_2,\cdots, S_{m-n}$. 
 For each $t\in\{1,\ldots,k\}$, let $\vv_t = \vp_t - \vr_t$ denote the $t$'th \emph{reaction vector}, so that $\vv_t(S) \in \Z$ denotes the amount by which the count of species $S\in\Lambda$ changes if the $t$'th reaction executes once.
 
 For the construction of the desired linear system, let $x_{ij}$ denote $\vec{d}_{S_{i}}(A_{j})$ for $S_i\in\Lambda$, $A_{j}\in \Delta$, $\vec{d}$ being a candidate decomposition vector. 
 \end{nt}

}%End of Notation of SSFA in NP. Used in full version and appendix.

        \NotationOfProofSSFixedAReducibleToIP

Now we come back to the proof of Proposition \ref{SFANP}.
        \begin{proof}
        \DetailedProofSSFixedAReducibleToIP
        \end{proof}

        \paragraph{Corollary \ref{cor:SFANPNP}: \SubsetAtomic\ $\in$ $\NP$}\label{corollarySSAInNPStatementOnly}
        
        \begin{proof}
        \DetailedProofOFSAInNP
        \end{proof}
  
  \newcommand{\StatementOnlyMonotoneThreeOnetoSSFA}{%Statement's reused in Appendix.
$\MonotoneOneThree$ is polynomial-time many-one reducible to $\SubsetFixedAtomic$}

        \paragraph{Proposition \ref{PartitionSFA}: \StatementOnlyMonotoneThreeOnetoSSFA}\label{statementOnlyMonotoneThreeOnetoSSFA}
        
        \begin{proof}
        \DetailedProofMonotoneOneThreePolyRedSSFA
        \end{proof}
        
        \paragraph{Remark \ref{SSFARemainNPCWhenBiMolecular}: Details}\label{sSFARemainNPCWhenBiMolecular}
        
        \LongCommentSSFARemainsNPCWhenBiMolecular
        
        \paragraph{Lemma \ref{viswithdirect}: \StatementLemmaShowingExistenceOfOneStepDecomposibleElement}\label{statementLemmaShowingExistenceOfOneStepDecomposibleElement}
        
        \begin{proof}
        \DetailedProofLemmaShowingExistenceOfOneStepDecomposibleElement
        \end{proof}
        
        \paragraph{Theorem \ref{VisAtomP}: \ReachableAtomicIsInP}\label{reachableAtomicIsInP}
%\end{theorem}
        \begin{proof}
        We shall \ProofReachableAtomicIsInP
        \end{proof}
        
        \paragraph{Proposition \ref{ReachableReach}: \StatementOfTheoremReachableReachIsPspaceComplete}\label{reachableReach}
        
        \begin{proof}
        \DetailedProofOfReRePspaceComplete
        \end{proof}
        
        \paragraph{Remark \ref{PASAWorksForAllSemilinear}}
        
        \begin{proof}
        \allsemilinear
         %\DetailedProofOfReRePspaceComplete
        \end{proof}

        \section{Atomicity with Core Composition}\label{acc}
        
        \SectionAtomicityWithCoreCompositionAsAppendix
        
        \section{Auxiliary Examples and Remarks}\label{Auxiliary}
        \RemarkWhyConditionTwo
        
        \ExampleOfStoichiometricMatrix
        
        \EgAndRemarkOnDecompositionMatrix
        
        \EgConfigGraph
        
        \cantreversable
    }%End of sub, Appendix.

\end{document}